\pdfoutput=1
\documentclass[a4paper,USenglish]{lipics-v2021}

\hideLIPIcs
\nolinenumbers %uncomment to disable line numbering

\usepackage{amsmath}
\usepackage{amsfonts}
\usepackage{amsthm}
\usepackage{subcaption}
\usepackage{graphicx}
\usepackage{xcolor}
\usepackage{algorithm}
\usepackage[noend]{algpseudocode}
\usepackage{comment}
\usepackage[title]{appendix}
\usepackage{tikz-cd}
\usetikzlibrary{patterns}

\author{Shunhao Oh}{School of Computer Science, Georgia Institute of Technology}{ohoh@gatech.edu}{}{}

\author{Dana Randall}{School of Computer Science, Georgia Institute of Technology}{randall@cc.gatech.edu}{https://orcid.org/0000-0002-1152-2627}{Supported by the National Science Foundation (NSF) award  CCF-2106687 and by U.S.\ Army Research Office (ARO) award MURI W911NF-19-1-0233.}

\author{Andr\'ea W.\ Richa}{School of Computing and Augmented Intelligence, Arizona State University}{aricha@asu.edu}{https://orcid.org/0000-0003-3592-3756}{Supported in part by the National Science Foundation (NSF) award CCF-2106917 and by U.S.\ Army Research Office (ARO) award MURI W911NF-19-1-0233.}

\authorrunning{S. Oh, D. Randall and A. W. Richa} %TODO mandatory. First: Use abbreviated first/middle names. Second (only in severe cases): Use first author plus 'et al.'

\Copyright{Shunhao Oh, Dana Randall and Andr\'ea W.\ Richa} %TODO mandatory, please use full first names. LIPIcs license is "CC-BY";  http://creativecommons.org/licenses/by/3.0/


 %TODO mandatory; please add comma-separated list of keywords

\title{Adaptive Collective Responses to Local Stimuli in Anonymous Dynamic Networks}
\EventEditors{David Doty and Paul Spirakis}
\EventNoEds{2}
\EventLongTitle{2nd Symposium on Algorithmic Foundations of Dynamic Networks (SAND 2023)}
\EventShortTitle{SAND 2023}
\EventAcronym{SAND}
\EventYear{2023}
\EventDate{June 19--21, 2023}
\EventLocation{Pisa, Italy}
\EventLogo{}
\SeriesVolume{257}
\ArticleNo{6}
\newcommand{\ex}{\mathbb{E}}

\newcommand{\mathbbNzero}{\mathbb{N}_{\geq 0}}

\algnewcommand{\algorithmicswitch}{\textbf{switch}}
\algnewcommand{\algorithmiccase}{\textbf{case}}
\algdef{SE}[SWITCH]{Switch}{EndSwitch}[1]{\algorithmicswitch\ #1\ \algorithmicdo}{\algorithmicend\ \algorithmicswitch}%
\algdef{SE}[CASE]{Case}{EndCase}[1]{\algorithmiccase\ #1:}{\algorithmicend\ \algorithmiccase}%
\algtext*{EndSwitch}%
\algtext*{EndCase}%

\newcommand{\aware}{\textsc{Aware}}
\newcommand{\unaware}{\textsc{Unaware}}
\newcommand{\alerttoken}{{alert token}}
\newcommand{\alerttokens}{{alert tokens}}
\newcommand{\cleartoken}{{all-clear token}}
\newcommand{\cleartokens}{{all-clear tokens}}
\newcommand{\allclear}{{all-clear}}
\newcommand{\Alerttoken}{{Alert token}}

\newcommand{\Cleartoken}{{All-clear token}}

\newcommand{\bstates}{{\widehat{\sigma}}}
\newcommand{\witset}{{\mathcal{W}}}

\newcommand{\behaviorgroup}{{behavior group}}
\newcommand{\behaviorgroups}{{behavior groups}}

\newcommand{\mobileb}{\textsc{Mobile}}
\newcommand{\unawareb}{\textsc{Unaware}}
\newcommand{\immobileb}{\textsc{Immobile}}

\newcommand{\stateU}{\mathcal{U}}
\newcommand{\stateAbase}{\mathcal{A}}
\newcommand{\stateA}{\mathcal{A}_{\emptyset}}
\newcommand{\stateAA}{\mathcal{A}_{\{A\}}}
\newcommand{\stateAW}{\mathcal{A}_{\{W\}}}
\newcommand{\stateAAW}{\mathcal{A}_{\{{A,W}\}}}
\newcommand{\stateAC}{\mathcal{A}_{\{C\}}}
\newcommand{\stateBU}{\widehat{\mathcal{U}}} % Unaware
\newcommand{\stateBA}{\widehat{\mathcal{M}}} % Aware
\newcommand{\stateBI}{\widehat{\mathcal{I}}} % Immobile
\newcommand{\NA}{N_{\stateAbase}}

\newcommand{\junk}[1]{}
\begin{document}
\maketitle

\begin{abstract}
We develop a framework for self-induced phase changes in programmable matter in which a collection of agents with limited computational and communication 
%and movement 
capabilities can collectively perform appropriate global tasks in response to local stimuli that dynamically appear and disappear.  Agents reside on graph vertices, where each stimulus is only recognized locally, and agents communicate via token passing along edges to alert other agents to transition to an \aware{} state when stimuli are present and an \unaware{} state when the stimuli disappear. We present an Adaptive Stimuli Algorithm that is robust to competing waves of messages as multiple stimuli change, possibly adversarially.  Moreover, in addition to handling arbitrary stimulus dynamics, the algorithm can handle agents reconfiguring the connections (edges) of the graph over time in a controlled way.

As an application, we show how this Adaptive Stimuli Algorithm on reconfigurable graphs can be used to solve the {\em foraging problem}, where food sources may be discovered, removed, or shifted at arbitrary times. We would like the agents  to consistently self-organize, 
using only local interactions, such that
if the food remains in a position long enough, the agents transition to a {\em gather phase} in which many collectively form a single large component with small perimeter %and compress 
around the food.
Alternatively, if no food source has existed recently, the agents should undergo a self-induced phase change and switch to a {\em search phase} in which they distribute themselves randomly throughout the lattice region to search for food.
Unlike previous approaches to foraging, this process is indefinitely repeatable, withstanding competing waves of messages that may interfere with each other. 
Like a physical phase change,  microscopic changes such as the deletion or addition of a single food source trigger these macroscopic, system-wide transitions as agents share information about the environment and respond locally to get the desired collective response.
%In this paper, we present a foraging algorithm that 
%is fully stochastic and builds upon the compression algorithm of Cannon et al.\ (PODC'16), motivated by a phase change that occurs in the fixed magnetization Ising model from statistical physics.
%When the food is present, agents incrementally  enter the gather phase where they attach and compress around the food; when the food disappears or moves, the agents transition to the search phase where they relax their ferromagnetic attraction 
%and instead move according to a simple exclusion process that causes the agents to disperse and explore the domain in search of new food.
%This algorithm is the first to leverage {\em self-induced phase changes as an algorithmic tool}. A key component of our algorithm is a careful token passing mechanism that ensures that the rate at which the compressed cluster around the food dissipates outpaces the rate at which the cluster may continue to grow, ensuring that the broadcast wave
%of the ``dispersion token'' will always outpace that of a compression wave. 
\end{abstract}
%\nopagenumber

\newpage
\section{Introduction}
\label{sec:intro}
Self-organizing collective behavior of interacting agents is a fundamental, nearly ubiquitous phenomenon across fields,  reliably producing rich and complex coordination.  In nature, examples at the micro- and nano-scales include coordinating cells, including our own immune system or self-repairing tissue (e.g.,~\cite{alberti}), and bacterial colonies~(e.g., \cite{Liu2015,Prindle2015}); 
%micro-scale swarm robotics, 
%and interacting particle systems in physics; 
at the macro-scale it can represent flocks of birds~\cite{Chazelle14}, shoals of fish aggregating to intimidate predators~\cite{Magurran1990}, fire ants forming rafts to survive floods~\cite{Mlot2011},
%coordination of drones~\cite{??},  
and human societal dynamics such as segregation~\cite{Schelling1971}.
Common characteristic of these disparate systems is that they are all self-actuated and respond to simple, local environmental stimuli to collectively change the ensemble's behavior.

In 1991, Toffoli and Margolus coined {\it programmable matter} to realize a physical computing medium composed of simple,  homogeneous entities that can dynamically alter its physical properties in a programmable fashion, controlled either by user input or its own autonomous sensing of its environment~\cite{Toffoli1991}. 
There are formidable challenges to realizing such
%programmable matter computationally
collective tasks and many researchers in distributed computing and swarm and modular robotics have investigated how such small, simply programmed entities can coordinate to solve complex tasks and exhibit useful emergent collective behavior (e.g.,~\cite{Sahin2005}).
%Additional algorithms based on this approach, which also exhibit similar phase changes, found applications to the {\em separation} problem~\cite{Cannon2019}, where colored particles would like to separate into clusters of same color particles, {\em aggregation}~\cite{Li2021-bobbots} where the system would like to compress but no longer needs to remain connected, shortcut bridging~\cite{Arroyo2018}, locomotion~\cite{Savoie2018}, alignment~\cite{kedia2022} and transport~\cite{Li2021-bobbots}.
A more ambitious goal, suggested by self-organizing collective systems in nature, is to design programmable matter systems of {\em self-actuated} individuals that {\em autonomously respond to continuous, local changes in their environment}.

%\subsection
%\paragraph*
%\vskip.1in
%{\bf The Dynamic Stimuli Problem: \ }
\vskip.1in
\noindent {\bf The Dynamic Stimuli Problem:} \ 
\label{section:model}
As a distributed framework for agents collectively 
self-organizing in response to changing stimuli, we consider the {\it dynamic stimuli problem}
in which we have a large number of agents that collectively respond  to local signals or stimuli. 
These agents have limited computational capabilities and each only communicates with a small set of immediate neighbors.
We represent these %communicating 
agents via a dynamic graph $G$ on $n$ vertices,
where the agents reside at the vertices and edges represent pairs that can perceive and interact with each other. At arbitrary points of time, that may be adversarially chosen, {\em stimuli} 
dynamically 
appear and disappear at the vertices of~$G$
 --- these can be a threat, such as 
 %the appearance of a 
 an unexpected predator, or an opportunity, such as new food or energy resources.
 
 An agent present at the same vertex as a stimulus acts as a {\em witness} and alerts other agents. 
If any agent continues to witness some stimulus over an extended period of time, we want all agents to eventually be alerted, switching to the \aware{} state; on the other hand, once witnesses stop sensing a stimulus for long enough, 
%if there are no witnesses for long enough, 
all agents should return to the \unaware{} state. Such collective state changes may repeat indefinitely as stimuli appear and disappear over time.
Converging to these two global states enables agents to
%The motivation behind this is for agents to be able to 
carry out {\em differing behaviors in the presence or absence of stimuli}, as observed by the respective witnesses.  As a notable and challenging example, in {\it the foraging problem,}  ``food'' may appear or disappear at arbitrary locations in the graph over time and we would like the collective to {\it gather around food} (also known as {\it dynamic free aggregation}) or {\it disperse in search of new sources}, depending on the whether or not an active food source has been identified.  Cannon {\it et al}$.$~\cite{Cannon2016} showed how computationally limited agents can be made to gather or disperse; 
however there the desired goal,  aggregation or dispersion, is fixed in advance and the algorithm cannot easily be adapted to move between these  according to changing  needs.

In addition to the stimuli dynamics, we assume that the agents 
may {\it reconfigure the connections (edges) of the graph $G$ over time}, 
but in a controlled way that still allows the agents to successfully manage the waves of state changes. 
 In a nutshell, $G$ is {\em reconfigurable} over time if it maintains recurring local connectivity  of the \aware{} agents (i.e., if it makes sure that the 1-hop neighborhood sets of \aware{} agents stay connected over time), as its edge set changes. 
The edge dynamics may be fully in the control of an adversary, or may be controlled by the agents themselves, depending on the context (e.g., in the foraging problem presented in Section~\ref{sec:foraging}, the agents control the edge dynamics).

In this framework, we assume that at all times there are at most a constant number~$w$ of stimuli present at the vertices of~$G$. Agents are anonymous and each  acts as a finite automaton with constant-size memory, constant degree, and no access to global information other than $w$ and a constant upper bound $\Delta$ on the maximum degree.
Individual agents are activated according to their own Poisson clocks and perform instantaneous actions upon activation, a standard way to allow sites to update independently and asynchronously (since the probability two Poisson clocks tick at the exact same instant in time is negligible).  For ease of discussion, we may assume the Poisson clocks have the same rate, so this model is equivalent to a {\em random sequential scheduler}  that chooses an agent uniformly at random to be activated at discrete iterations $t \in \{1,2,3,\ldots\}$.
We denote by $G_t$ the configuration of the reconfigurable graph at iteration~$t$.
When {\em activated} at iteration $t$, an agent perceives its own state and the states of its current neighbors in $G_t$, performs a bounded amount of computation, and can change its own and its neighbors states, including any ``tokens'' (i.e.,  constant size messages received or sent).\footnote{Since we assume a sequential scheduler, such an action is justified; in the presence of a stronger adversarial scheduler, e.g., the asynchronous scheduler, one would need a more detailed message passing mechanism to ensure the transfer of tokens between agents, and resulting changes in their states.}  At each iteration $t \in \{1,2,3,\ldots\}$, we denote  by $\witset_t \subseteq V$ the set of witnesses. The sets $\witset_t$ can change arbitrarily (adversarially) over time, but $|\witset_t|$ is always bounded by the constant $w$,
%,a constant, 
for all $t$, since $w$ is an upper on the number of concurrent stimuli.

%%%%%%%%%%%%%%%%%%%%%%%%%%%%%%%%%%%%%%%%%
\vskip.1in
\noindent{\bf Overview of Results:} \
%\subparagraph{Our results.}
Our contributions are two-fold. First, we present an efficient, robust algorithm  for the dynamic stimuli problem for 
%static and 
a class of reconfigurable graphs.
(The precise details of what constitutes a \emph{valid reconfigurable graph} will be given in Section~\ref{sec:reconfig}.)
%in Section~\ref{sec:alg}. 
 Whenever an agent encounters a new stimulus, the entire collective
 %connected component 
efficiently transforms  to the \aware{} state, so the  agents can implement an appropriate collective response. After a stimulus vanishes, they all return to the \unaware{} state, recovering their neutral collective behavior.  

We show that the system will always converge to the appropriate state (i.e., \aware{} or \unaware{}) once the stimuli stabilize for a sufficient period of time. Specifically,  
if there are no witnesses in the system for a sufficient period of time,
 then all agents in the Adaptive Stimuli Algorithm will reach and remain in the \unaware{} state in  $O(n^2)$
 %polynomial 
 expected time (Theorem~\ref{theorem:mainresultnostimulidynamic}).
% \item 
 Likewise, if the set of witnesses (and stimuli) remains unchanged for a sufficient period of time,
 %sufficiently long,
 %$|\witset_T| > 0$, 
 %\marginpar{\tiny DR: informal def of recurring rate?}
 all agents will reach and remain in the \aware{} state in
 expected time that is a polynomial in $n$ and the ``recurring rate'' of $G$, which captures how frequently disconnected vertices come back in contact with each other
 (Theorem~\ref{theorem:mainresultwithstimulidynamic}). In particular, if $G$ is static or if $G_t$ is connected, for all $t$, the expected convergence time until all agents transition to the \aware{} state is $O(n^5)$ (Theorem~\ref{theorem:mainresultwithstimuli}) and $O(n^6 \log n)$ (Corollary~\ref{cor:connected-graphs}), respectively.
Moreover, the system can recover if the witness set changes before the system converges to the \aware{} or \unaware{} states.

While the arguments are simpler for sequences of connected graphs generated by, say, an oblivious adversary, 
the extension
to the broader class of reconfigurable graphs includes graphs possibly given by non-oblivious adversaries that may occasionally disconnect. This generalization provides
more flexibility  for agents in the \aware{} and \unaware{} states to implement more complex behaviors,
%This distinction allowing for reconfigurable graphs 
as was used for applying the Adaptive Stimuli Algorithm to foraging, which we describe next.

Our second main contribution is the first efficient algorithm for the {\it foraging problem}, where food dynamically appears and disappears over time at arbitrary sites on a finite $\sqrt{N}\times \sqrt{N}$ region of the triangular lattice.
Agents want to gather around any discovered food source (also known as {\it dynamic free aggregation}) or disperse in search of food.  The algorithm of Cannon {\it et al}$.$~\cite{Cannon2016, Li2021-bobbots} uses insights from the high and low temperature phases of the ferromagnetic Ising model from statistical physics
%use solve the individual problems of aggregation and dispersion on the triangular lattice {\it in the non-adaptive setting} where there is a fixed preset collective goal; motivated by a phase change in the ferromagnetic Ising model;  the algorithm has been shown 
to provably achieve either desired collective response: %the algorithm uses a preset 
there is a preset global parameter $\lambda$ 
%given 
related to inverse 
temperature, and the algorithm provably achieves aggregation when $\lambda$ is sufficiently high and dispersion when sufficiently low. 
We show here that by applying the Adaptive Stimuli Algorithm, the  phase change (or bifurcation) for aggregation and dispersion %\cite{Cannon2016, Li2021-bobbots} 
can be self-modulated based on local environmental cues that are communicated through the collective to induce desirable system-wide behaviors in polynomial time in $n$ and $N$, as stated in Theorems~\ref{thm:compressionnofood} and \ref{thm:compressionhasfoodprecise}.

Collectively transitioning between \aware{} and \unaware{} states enables agents to {\it correctly self-regulate system-wide adjustments} in their bias parameters when one or more agents notice the presence or depletion of food to induce the appropriate global coordination to provably transition the collective between macro-modes when required.  We believe other collective behaviors exhibiting emergent bifurcations, including separation/integration \cite{Cannon2019} and alignment/nonalignment \cite{kedia2022}, can be similarly self-modulated.

The distinction between polynomial and super-polynomial 
running times is significant here because our algorithms necessarily make use of competing broadcast waves to propagate commands to change states. A naive implementation of such a broadcast system may put us in situations where neither type of wave  gets to complete its propagation cycle. This may continue for an unknown amount of time, so the agents may fail to reach an agreement on their state.  
%A key component of our algorithm is a 
The carefully engineered token passing mechanism  ensures that when a stimulus has been removed, the {\em rate at which the agents ``reset'' to the \unaware{} state
%become aware that the stimulus has been removed 
outpaces the rate at which the cluster of \aware{} agents may continue to grow}, ensuring that the newer broadcast wave  
always supersedes
%will always outpace 
previous ones and completes in expected polynomial time.  Moreover, while a stimulus is present, there is a continuous {\em probabilistic generation of tokens that move according to a $d_{max}$-random walk} among \aware{} agents until they find a new agent to become \aware{}, thus ensuring the successful convergence to the \aware{} state in expected polynomial time.

%%%%%%%%%%%%%%%%%%%%%%%%%%%%%%%%%%%%%%%%%%
\vskip.1in
\noindent {\bf Related work:} \ 
Dynamic networks have been of growing interest recently and have spawned several model variants (see, e.g., the surveys in~\cite{casteigts2018finding} and \cite{KuhnOshman11}). There is also a vast literature on broadcasting, or information dissemination, in both static and dynamic networks (e.g.,~\cite{KuhnLO10,HaeuplerK11,ClementiRT12,dutta2013complexity,DinitzFGN22}), where one would like to disseminate  $k$ messages to all nodes of a graph $G$, usually with unique token ids 
%from 1 to $k$ 
and $k\leq n$, polylog memory at the nodes (which may also have unique ids), and often nodes' knowledge of $k$ and possibly also of $n$. Note that any of these assumptions violates our agents'  memory or computational capabilities. Moreover, since %we would like 
our collective state-changing process runs indefinitely, 
any naive adaptations of these algorithms 
 would need that $k\to \infty$ to ensure that with any sequence of broadcast waves, the latest always wins. 

Broadcast algorithms that do not explicitly keep any information on $k$ (number of tokens or broadcast waves) or $n$ would be more amenable to our agents. Amnesiac flooding is one such broadcast algorithm that works on a network of anonymous nodes without keeping any state or other information as the broadcast progresses. In~\cite{HussakTrehan19}, 
Hussak and Trehan show that amnesiac flooding will always terminate in a static network under synchronous message passing, but may fail on a dynamic network or with non-synchronous executions.

Many studies in self-actuated systems take inspiration from emergent behavior in social insects, but
either lack rigorous mathematical foundations explaining the generality and limitations as sizes scale (see, e.g., \cite{Garnier2005,Garnier2009,Correll2011,Xie2019}), often approaching the thermodynamic limits of computing~\cite{Wolpert2019} and power~\cite{Dario1992}, or rely on long-range signaling, such as microphones or line-of-sight sensors \cite{Soysal2005,Fates2010,Fates2011,Ozdemir2018}.  Some recent work on stochastic approaches modeled after systems from particle physics has been made rigorous, but only when a single, static goal is desired \cite{Cannon2016, Cannon2019, Li2021-bobbots, Arroyo2018, Savoie2018, kedia2022}.

%Several works have looked at this problem form a deterministic (e.g.,~\cite{...}), randomized (e.g.,~\cite{...}), adversarial (e.g.,~\cite{...}), and more recently smoothed analysis (e.g., \cite{...}) perspectives.
%Note that solutions to the broadcastingWhile our problem certainly relates

\section{The Adaptive Stimuli Algorithm}\label{alg-sec}
\label{sec:alg}
The {\em Adaptive Stimuli Algorithm} is designed to efficiently respond to dynamic local stimuli that indefinitely appear and disappear at the vertices of $G$. Recall the goal of this algorithm is to allow the collective to 
%efficiently 
converge to the \aware{} state whenever a stimulus is witnessed for long enough and to the \unaware{} state if no stimulus has been detected recently.
%, then the collective should converge to having all agents in the \unaware{} state. 
The algorithm converges in expected polynomial time in both scenarios under a reconfigurable dynamic setting, as we show in Sections~\ref{sec:static}-\ref{sec:reconfig}, even as the process repeats indefinitely.

\begin{figure}[t]
\begin{minipage}[b]{.48\linewidth}
\begin{center}
\begin{subfigure}[b]{0.3\linewidth}
  \begin{center}
  \begin{tikzpicture}[x=0.8cm,y=0.8cm]
  \node[align=left] at (-1.20902,0.3) {\normalsize 1.};
\draw[black, line width=0.4mm] (-0.606763,-0.440839) -- (-0.202254,-0.146946);
\draw[black, line width=0.4mm] (0.202254,-0.146946) -- (0.606763,-0.440839);
\draw[black, line width=0.4mm] (-0.809017,-0.587785) circle (0.25);
\node[align=left] at (-0.809017,-1.08779) {\scriptsize $\stateU$};
\draw[black, line width=0.4mm] (0,0) circle (0.25);
\node[align=left] at (0,-0.5) {\scriptsize $\stateU$};
\draw[black, line width=0.4mm] (0.809017,-0.587785) circle (0.25);
\node[align=left] at (0.809017,-1.08779) {\scriptsize $\stateU$};
\draw[black!20!red, line width=0.6mm, ] (-0.809017,-0.587785) circle (0.36);
  \end{tikzpicture}
  \end{center}
\end{subfigure}
\begin{subfigure}[b]{0.3\linewidth}
  \begin{center}
  \begin{tikzpicture}[x=0.8cm,y=0.8cm]
  \node[align=left] at (-1.20902,0.3) {\normalsize 2.};
\draw[black, line width=0.4mm] (-0.606763,-0.440839) -- (-0.202254,-0.146946);
\draw[black, line width=0.4mm] (0.202254,-0.146946) -- (0.606763,-0.440839);
\draw[black, line width=0.4mm] (-0.809017,-0.587785) circle (0.25);
\node[align=left] at (-0.809017,-1.08779) {\scriptsize $\stateAW$};
\node[align=left] at (-0.809017,-0.587785) {\scriptsize $\times$};
\draw[black, line width=0.4mm] (0,0) circle (0.25);
\node[align=left] at (0,-0.5) {\scriptsize $\stateU$};
\draw[black, line width=0.4mm] (0.809017,-0.587785) circle (0.25);
\node[align=left] at (0.809017,-1.08779) {\scriptsize $\stateU$};
\draw[black!20!red, line width=0.6mm, ] (-0.809017,-0.587785) circle (0.36);
  \end{tikzpicture}
  \end{center}
\end{subfigure}
\begin{subfigure}[b]{0.3\linewidth}
  \begin{center}
  \begin{tikzpicture}[x=0.8cm,y=0.8cm]
  \node[align=left] at (-1.20902,0.3) {\normalsize 3.};
\draw[black, line width=0.4mm] (-0.606763,-0.440839) -- (-0.202254,-0.146946);
\draw[black, line width=0.4mm] (0.202254,-0.146946) -- (0.606763,-0.440839);
\draw[black, line width=0.4mm] (-0.809017,-0.587785) circle (0.25);
\node[align=left] at (-0.809017,-1.08779) {\scriptsize $\stateAAW$};
\node[align=left] at (-0.809017,-0.587785) {\scriptsize $\times$};
\draw[black, line width=0.4mm] (0,0) circle (0.25);
\node[align=left] at (0,-0.5) {\scriptsize $\stateU$};
\draw[black, line width=0.4mm] (0.809017,-0.587785) circle (0.25);
\node[align=left] at (0.809017,-1.08779) {\scriptsize $\stateU$};
\draw[black!20!red, line width=0.6mm, ] (-0.809017,-0.587785) circle (0.36);
  \end{tikzpicture}
  \end{center}
\end{subfigure}
\begin{subfigure}[b]{0.3\linewidth}
  \begin{center}
  \begin{tikzpicture}[x=0.8cm,y=0.8cm]
  \node[align=left] at (-1.20902,0.3) {\normalsize 4.};
\draw[black, line width=0.4mm] (-0.606763,-0.440839) -- (-0.202254,-0.146946);
\draw[black, line width=0.4mm] (0.202254,-0.146946) -- (0.606763,-0.440839);
\draw[black, line width=0.4mm] (-0.809017,-0.587785) circle (0.25);
\node[align=left] at (-0.809017,-1.08779) {\scriptsize $\stateAW$};
\draw[black, line width=0.4mm] (0,0) circle (0.25);
\node[align=left] at (0,-0.5) {\scriptsize $\stateA$};
\node[align=left] at (0,0) {\scriptsize $\times$};
\draw[black, line width=0.4mm] (0.809017,-0.587785) circle (0.25);
\node[align=left] at (0.809017,-1.08779) {\scriptsize $\stateU$};
\draw[black!20!red, line width=0.6mm, ] (-0.809017,-0.587785) circle (0.36);
  \end{tikzpicture}
  \end{center}
\end{subfigure}
\begin{subfigure}[b]{0.3\linewidth}
  \begin{center}
  \begin{tikzpicture}[x=0.8cm,y=0.8cm]
  \node[align=left] at (-1.20902,0.3) {\normalsize 5.};
\draw[black, line width=0.4mm] (-0.606763,-0.440839) -- (-0.202254,-0.146946);
\draw[black, line width=0.4mm] (0.202254,-0.146946) -- (0.606763,-0.440839);
\draw[black, line width=0.4mm] (-0.809017,-0.587785) circle (0.25);
\node[align=left] at (-0.809017,-1.08779) {\scriptsize $\stateAAW$};
\node[align=left] at (-0.809017,-0.587785) {\scriptsize $\times$};
\draw[black, line width=0.4mm] (0,0) circle (0.25);
\node[align=left] at (0,-0.5) {\scriptsize $\stateA$};
\draw[black, line width=0.4mm] (0.809017,-0.587785) circle (0.25);
\node[align=left] at (0.809017,-1.08779) {\scriptsize $\stateU$};
\draw[black!20!red, line width=0.6mm, ] (-0.809017,-0.587785) circle (0.36);
  \end{tikzpicture}
  \end{center}
\end{subfigure}
\begin{subfigure}[b]{0.3\linewidth}
  \begin{center}
  \begin{tikzpicture}[x=0.8cm,y=0.8cm]
  \node[align=left] at (-1.20902,0.3) {\normalsize 6.};
\draw[black, line width=0.4mm] (-0.606763,-0.440839) -- (-0.202254,-0.146946);
\draw[black, line width=0.4mm] (0.202254,-0.146946) -- (0.606763,-0.440839);
\draw[black, line width=0.4mm] (-0.809017,-0.587785) circle (0.25);
\node[align=left] at (-0.809017,-1.08779) {\scriptsize $\stateAW$};
\node[align=left] at (-0.809017,-0.587785) {\scriptsize $\times$};
\draw[black, line width=0.4mm] (0,0) circle (0.25);
\node[align=left] at (0,-0.5) {\scriptsize $\stateAA$};
\draw[black, line width=0.4mm] (0.809017,-0.587785) circle (0.25);
\node[align=left] at (0.809017,-1.08779) {\scriptsize $\stateU$};
\draw[black!20!red, line width=0.6mm, ] (-0.809017,-0.587785) circle (0.36);
  \end{tikzpicture}
  \end{center}
\end{subfigure}
\begin{subfigure}[b]{0.3\linewidth}
  \begin{center}
  \begin{tikzpicture}[x=0.8cm,y=0.8cm]
  \node[align=left] at (-1.20902,0.3) {\normalsize 7.};
\draw[black, line width=0.4mm] (-0.606763,-0.440839) -- (-0.202254,-0.146946);
\draw[black, line width=0.4mm] (0.202254,-0.146946) -- (0.606763,-0.440839);
\draw[black, line width=0.4mm] (-0.809017,-0.587785) circle (0.25);
\node[align=left] at (-0.809017,-1.08779) {\scriptsize $\stateAW$};
\draw[black, line width=0.4mm] (0,0) circle (0.25);
\node[align=left] at (0,-0.5) {\scriptsize $\stateA$};
\draw[black, line width=0.4mm] (0.809017,-0.587785) circle (0.25);
\node[align=left] at (0.809017,-1.08779) {\scriptsize $\stateA$};
\node[align=left] at (0.809017,-0.587785) {\scriptsize $\times$};
\draw[black!20!red, line width=0.6mm, ] (-0.809017,-0.587785) circle (0.36);
  \end{tikzpicture}
  \end{center}
\end{subfigure}%
\end{center}
\caption{Illustration of how the presence of a witness (circled in red) gradually converts all agents to the \aware{} state through the distribution of \alerttokens. The agent with the ``$\times$'' is the agent activated in that step.}
\label{fig:alerttokens}
\end{minipage}
~~
\begin{minipage}[b]{.48\linewidth}
\begin{center}
\begin{subfigure}[b]{0.3\linewidth}
  \begin{center}
  \begin{tikzpicture}[x=0.8cm,y=0.8cm]
  \node[align=left] at (-1.10902,0.2) {\normalsize 1.};
\draw[black, line width=0.4mm] (-0.606763,-0.440839) -- (-0.202254,-0.146946);
\draw[black, line width=0.4mm] (0.202254,-0.146946) -- (0.606763,-0.440839);
\draw[black, line width=0.4mm] (0.606763,-0.734732) -- (0.202254,-1.02862);
\draw[black, line width=0.4mm] (-0.606763,-0.734732) -- (-0.202254,-1.02862);
\draw[black, line width=0.4mm] (-0.809017,-0.587785) circle (0.25);
\node[align=left] at (-0.809017,-1.08779) {\scriptsize $\stateAW$};
\draw[black, line width=0.4mm] (0,0) circle (0.25);
\node[align=left] at (0,-0.5) {\scriptsize $\stateA$};
\draw[black, line width=0.4mm] (0.809017,-0.587785) circle (0.25);
\node[align=left] at (0.809017,-1.08779) {\scriptsize $\stateA$};
\draw[black, line width=0.4mm] (0,-1.17557) circle (0.25);
\node[align=left] at (0,-1.67557) {\scriptsize $\stateA$};
  \end{tikzpicture}
  \end{center}
\end{subfigure}
\begin{subfigure}[b]{0.3\linewidth}
  \begin{center}
  \begin{tikzpicture}[x=0.8cm,y=0.8cm]
  \node[align=left] at (-1.10902,0.2) {\normalsize 2.};
\draw[black, line width=0.4mm] (-0.606763,-0.440839) -- (-0.202254,-0.146946);
\draw[black, line width=0.4mm] (0.202254,-0.146946) -- (0.606763,-0.440839);
\draw[black, line width=0.4mm] (0.606763,-0.734732) -- (0.202254,-1.02862);
\draw[black, line width=0.4mm] (-0.606763,-0.734732) -- (-0.202254,-1.02862);
\draw[black, line width=0.4mm] (-0.809017,-0.587785) circle (0.25);
\node[align=left] at (-0.809017,-1.08779) {\scriptsize $\stateU$};
\node[align=left] at (-0.809017,-0.587785) {\scriptsize $\times$};
\draw[black, line width=0.4mm] (0,0) circle (0.25);
\node[align=left] at (0,-0.5) {\scriptsize $\stateAC$};
\draw[black, line width=0.4mm] (0.809017,-0.587785) circle (0.25);
\node[align=left] at (0.809017,-1.08779) {\scriptsize $\stateA$};
\draw[black, line width=0.4mm] (0,-1.17557) circle (0.25);
\node[align=left] at (0,-1.67557) {\scriptsize $\stateAC$};
  \end{tikzpicture}
  \end{center}
\end{subfigure}
\begin{subfigure}[b]{0.3\linewidth}
  \begin{center}
  \begin{tikzpicture}[x=0.8cm,y=0.8cm]
  \node[align=left] at (-1.10902,0.2) {\normalsize 3.};
\draw[black, line width=0.4mm] (-0.606763,-0.440839) -- (-0.202254,-0.146946);
\draw[black, line width=0.4mm] (0.202254,-0.146946) -- (0.606763,-0.440839);
\draw[black, line width=0.4mm] (0.606763,-0.734732) -- (0.202254,-1.02862);
\draw[black, line width=0.4mm] (-0.606763,-0.734732) -- (-0.202254,-1.02862);
\draw[black, line width=0.4mm] (-0.809017,-0.587785) circle (0.25);
\node[align=left] at (-0.809017,-1.08779) {\scriptsize $\stateU$};
\draw[black, line width=0.4mm] (0,0) circle (0.25);
\node[align=left] at (0,-0.5) {\scriptsize $\stateU$};
\node[align=left] at (0,0) {\scriptsize $\times$};
\draw[black, line width=0.4mm] (0.809017,-0.587785) circle (0.25);
\node[align=left] at (0.809017,-1.08779) {\scriptsize $\stateAC$};
\draw[black, line width=0.4mm] (0,-1.17557) circle (0.25);
\node[align=left] at (0,-1.67557) {\scriptsize $\stateAC$};
  \end{tikzpicture}
  \end{center}
\end{subfigure}
\vskip.2in
\begin{subfigure}[b]{0.3\linewidth}
  \begin{center}
  \begin{tikzpicture}[x=0.8cm,y=0.8cm]
  \node[align=left] at (-1.10902,0.2) {\normalsize 4.};
\draw[black, line width=0.4mm] (-0.606763,-0.440839) -- (-0.202254,-0.146946);
\draw[black, line width=0.4mm] (0.202254,-0.146946) -- (0.606763,-0.440839);
\draw[black, line width=0.4mm] (0.606763,-0.734732) -- (0.202254,-1.02862);
\draw[black, line width=0.4mm] (-0.606763,-0.734732) -- (-0.202254,-1.02862);
\draw[black, line width=0.4mm] (-0.809017,-0.587785) circle (0.25);
\node[align=left] at (-0.809017,-1.08779) {\scriptsize $\stateU$};
\draw[black, line width=0.4mm] (0,0) circle (0.25);
\node[align=left] at (0,-0.5) {\scriptsize $\stateU$};
\draw[black, line width=0.4mm] (0.809017,-0.587785) circle (0.25);
\node[align=left] at (0.809017,-1.08779) {\scriptsize $\stateU$};
\node[align=left] at (0.809017,-0.587785) {\scriptsize $\times$};
\draw[black, line width=0.4mm] (0,-1.17557) circle (0.25);
\node[align=left] at (0,-1.67557) {\scriptsize $\stateAC$};
  \end{tikzpicture}
  \end{center}
\end{subfigure}
\begin{subfigure}[b]{0.3\linewidth}
  \begin{center}
  \begin{tikzpicture}[x=0.8cm,y=0.8cm]
  \node[align=left] at (-1.10902,0.2) {\normalsize 5.};
\draw[black, line width=0.4mm] (-0.606763,-0.440839) -- (-0.202254,-0.146946);
\draw[black, line width=0.4mm] (0.202254,-0.146946) -- (0.606763,-0.440839);
\draw[black, line width=0.4mm] (0.606763,-0.734732) -- (0.202254,-1.02862);
\draw[black, line width=0.4mm] (-0.606763,-0.734732) -- (-0.202254,-1.02862);
\draw[black, line width=0.4mm] (-0.809017,-0.587785) circle (0.25);
\node[align=left] at (-0.809017,-1.08779) {\scriptsize $\stateU$};
\draw[black, line width=0.4mm] (0,0) circle (0.25);
\node[align=left] at (0,-0.5) {\scriptsize $\stateU$};
\draw[black, line width=0.4mm] (0.809017,-0.587785) circle (0.25);
\node[align=left] at (0.809017,-1.08779) {\scriptsize $\stateU$};
\draw[black, line width=0.4mm] (0,-1.17557) circle (0.25);
\node[align=left] at (0,-1.67557) {\scriptsize $\stateU$};
\node[align=left] at (0,-1.17557) {\scriptsize $\times$};
  \end{tikzpicture}
  \end{center}
\end{subfigure}
\vskip.1in
\end{center}
\caption{Illustration of how \cleartokens{} are broadcast from an agent with the witness flag set that is no longer a witness (the leftmost agent). The agent with the ``$\times$'' is the  agent activated in that step.}
\label{fig:cleartokens}
\end{minipage}
\end{figure}

All agents know two parameters of the system:  $\Delta \geq 1$, an upper bound on the maximum degree of the graph, and $w\geq 1,$ an upper bound on the size of the witness sets $\witset_t$ at all times~$t$ (which is needed to determine the probability $p < 1/w$ for some agents to change states or generate certain tokens). 
Our algorithm defines a carefully balanced token passing mechanism, where a \emph{token} is a constant-size piece of information: 
Upon activation, an \aware{} witness $u$ continuously generates {\em \alerttokens} one at a time, with probability $p$, which will each move through a {\em random walk over \aware{} 
agents until they come in contact with a neighboring \unaware{} agent $u$}: The token is then consumed and $u$ changes its state to \aware{} (Figure \ref{fig:alerttokens}).
On the other hand, if a witness notices that its co-located stimulus has disappeared, it will initiate an {\cleartoken{} broadcast wave} which will proceed through agents in the \aware{} state, switching those to \unaware{} (Figure \ref{fig:cleartokens}).

The differences between these two carefully crafted token passing mechanisms 
%(more details appear in Algorithm~\ref{alg:stimulusalgorithm} and below) 
allow us to 
%prove that if a stimulus is active for long enough, the e
ensure that whenever there has been no stimulus in the network for long enough, the rate at which the agents deterministically learn this (through broadcasts of \cleartokens)
%that the stimulus has been removed 
and become \unaware{} always outpaces the probabilistic rate at which the cluster of \aware{} agents may still continue to grow.  Thus, the \unaware{} broadcast wave  will always outpace any residual \aware{} waves in the system and will allow the collective to correctly converge to the desired \unaware{} state. 
%While an agent remains \unaware, it will not disseminate any state-changing tokens, so the \cleartoken{}  broadcast wave will eventually "die out.". 
On the other hand, if the witness set is non-empty and remains stable for a long enough period of time, all agents will eventually switch to the \aware{} state since, after some time, there will be no \cleartokens{} in $G$, as \unaware{} agents do not ever generate or broadcast tokens. 
%%%%%%%%%%%%%%%%%%

\begin{algorithm}[t]
\caption{Adaptive Stimuli Algorithm.}% $p \in (0,1)$ and $\Delta \geq 1$ are parameters.}
\label{alg:stimulusalgorithm}
\begin{algorithmic}[1]
\Procedure{The Adaptive-Stimuli-Algorithm}{$u$}
\State  Let $p<1/w \in (0,1)$
\State $u.isWitness \gets $\textsc{True} if $u$ is a witness, else $u.isWitness \gets$\textsc{False}
\If{$u.isWitness$ and $u.state \not\in \{\stateAW, \stateAAW\}$}
    \State With probability $p$, $u.state \gets \stateAW$ \Comment{$u$ becomes \aware{} witness with prob.~$p$}
\ElsIf{$\neg u.isWitness$ and $u.state \in \{\stateAW, \stateAAW\}$} \Comment{stimulus no longer at $u$}
    \For{each $v \in N_\stateAbase(u)$} 
    %such that \textcolor{red}{such that $v.state\not= \stateU$} \textcolor{black}} 
    \Comment{$N_\stateAbase(u) =$ current \aware{} neighbors of $u$}
        \State $v.state \gets \stateAC$ \Comment{\cleartoken{} broadcast to \aware{} neighbors of $u$}
    \EndFor
    \State $u.state \gets \stateU$
\Else \Switch{$u.state$}
        \Case{$\stateU$}
            \If{$\exists v \in N_\stateAbase(u), v.state=\stateAbase_S \in \{\stateAA, \stateAAW\}$} \Comment{$v$ has \alerttoken}
            \State $v.state \gets \stateAbase_{S \setminus \{A\}}$ \Comment{$v$ consumes \alerttoken{}}
                \State $u.state \gets \stateA$ \Comment{$u$ becomes aware}
                %{$u$ consumes \alerttoken}
                %\State Assuming $v.state = \stateAbase_S$ for some set $S$, $v.state \gets \stateAbase_{S \setminus A}$
            \EndIf
        \EndCase
        \Case{$\stateAA$ or $\stateAAW$}
            \State $x \gets$ random value in $[0,1]$
            \If{$x \leq d_G(u)/\Delta$} \Comment{$d_{max}$-random walk} \label{alg:line-dmax}
                \State $v \gets$ random neighbor of $u$
                \If{$v.state=\stateAbase_{S'} \in \{\stateA, \stateAW\}$} \Comment{\aware{} state, no \alerttoken}
                \State Let $u.state = \stateAbase_S$; $u.state \gets \stateAbase_{S \setminus \{A\}}$ \Comment{$u$ sends \alerttoken{} to $v$}
                    \State $v.state \gets \stateAbase_{S' \cup \{A\}}$ \Comment{$v$ receives \alerttoken}
                    %\State Assuming $v.state = \stateAbase_{S'}$ for some set $S'$, $v.state \gets \stateAbase_{S' \cup \{A\}}$
                \EndIf
            \EndIf
        \EndCase
        \Case{$\stateAW$}
            \State With probability $p$, $u.state \gets \stateAAW$ \Comment{generate \alerttoken{} with prob. $p$}
        \EndCase
        \Case{$\stateAC$}
            \For{each $v \in N_\stateAbase(u)$} %\textcolor{red}{such that $v.state\not= \stateU$} \textcolor{black}}
                \State $v.state \gets \stateAC$ \Comment{$u$ broadcasts \cleartoken{} to all aware neighbors}
            \EndFor
            \State $u.state \gets \stateU$
        \EndCase
    \EndSwitch\EndIf
\EndProcedure
\end{algorithmic}
\end{algorithm}

\noindent
To define the Adaptive Stimuli Algorithm, we utilize the following flags and states:
 %The two main general states for an agent are \aware{} and \unaware. For agents in the \aware{} state, we define a few additional bits of information:
\begin{itemize}
\item \textbf{\Alerttoken{} flag} ($A$): used to indicate that an agent has an \alerttoken. 

\item \textbf{\Cleartoken{} flag} ($C$): used to indicate that the agent has an \cleartoken.

\item \textbf{Witness flag} ($W$): used to indicate that the agent is a witness.

\item {\bf States}: The \unaware{} state is denoted by $\stateU$, while \{$\stateA$, $\stateAA$, $\stateAW$, $\stateAAW$, $\stateAC$\} denote the \aware{} states.
In the \aware{} states, the subscript denotes the subset of flags that are currently set. Note that the \cleartoken{} flag is only set when the other two are not, giving us six distinct states in total.

\end{itemize}

Algorithm~\ref{alg:stimulusalgorithm} formalizes the actions executed by each agent $u$ when activated. Agent $u$'s actions
%When an agent $u$ is activated at time $t$, as shown in Algorithm~\ref{alg:stimulusalgorithm}, its actions 
depend on the current state it is in and whether it senses a stimulus.
We describe the behavior for each of the possible cases below: 

\begin{itemize}
\item \textbf{Non-matching witness flag}:
This is a special case that occurs if $u$ is currently a witness to some stimuli
%(i.e.,  in set $\witset_t$) 
but its witness flag has not been set yet, or if $u$ has its witness flag set but is no longer a witness. 
%in the current time step $t$ does not match whether $u$ is in a state with the witness flag set ($\stateAW$~or $\stateAAW$). 
This case takes priority over all the other possible cases, since $u$ cannot take any action before its witness status and its state
%witness flag status 
match.
%the other behaviors below, in that none of the other behaviors will be carried out in this situation.
If $u$ is a witness but does not have the witness flag set, then with probability $p$, switch $u$ to state $\stateAW$. On the other hand, if $u$ is not a witness but has the witness flag set, switch $u$ to \unaware{} (by setting $u.state=\stateU$) and broadcast  an \cleartoken{} to all of $u$'s \aware{} neighbors (this token overrides any other tokens the neighbors may have). %

\item \textbf{\unaware{} state ($\stateU$)}: 
%If $u$ is a witness, with probability $p$, switch $u$ to state $\stateAW$. If $u$ is not a witness and there 
If $u$ has an \aware{} neighbor $v$ with an \alerttoken{} (i.e., $v.state\in\{\stateAA,\stateAAW\}$), then $u$ consumes the \alerttoken{} from $v$ (by setting $v$'s \alerttoken{} flag to \textsc{False}) and switches to the state $\stateA$.

\item \textbf{\aware{} state with \alerttoken{} ($\stateAA$, $\stateAAW$)}:
%Pick a neighbor $v$ of $u$ at random. If $v$ is in the \aware{} state without the \alerttoken{} flag set ($\stateA$, $\stateAW$), swap the \alerttoken{} to $v$ by toggling the \alerttoken{} flags on both $u$ and $v$.
Pick a random neighbor of $u$ such that each neighbor is picked with probability $\frac{1}{\Delta}$, with a probability $1 - \frac{d_G(u)}{\Delta}$ of staying at $u$ (this executes a $d_{max}$-random walk~\cite{RandomWalksRecurringTopologies}).
If an \aware{} neighbor $v$ is picked and $v$ does not have an alert nor an \cleartoken{} (that is, $v.state\in \{\stateA,\stateAW\}$), move the \alerttoken{} to $v$ by toggling the \alerttoken{} flags on both $u$ and $v$.

\item \textbf{\aware{} state with witness flag but without an \alerttoken{} ($\stateAW$)}: With probability $p$, $u$ generates a new \alerttoken{} by switching to state $\stateAAW$.

\item \textbf{\aware{} state with \cleartoken{} ($\stateAC$)}: Switch $u$ to the \unaware{} state $\stateU$~and broadcast the \cleartoken{} to all of its \aware{} state neighbors.
%to state $\stateAC$, spreading the ``all clear'' signal.
\end{itemize}

\noindent The use of a $d_{max}$-random walk instead of regular random walk (Line~\ref{alg:line-dmax} of Algorithm~\ref{alg:stimulusalgorithm})
%Basically a $d_{max}$-random walk 
normalizes the probabilities of transitioning along an edge by the maximum degree of the nodes (or a constant upper bound on that), so that these transition probabilities cannot change during the evolution of the graph. %This allows us to use known results for random walks over dynamic networks, in particular, that 
A $d_{max}$-random walk has polynomial hitting time on any connected dynamic network (while a regular random walk might not) \cite{RandomWalksDynamicGraphs}.

\section{Static graph topologies}
%the static graph setting}
\label{sec:static}
For simplicity, we will first state and prove our results for {\em static connected graph topologies} (Theorems~\ref{theorem:mainresultnostimuli} and~\ref{theorem:mainresultwithstimuli}), where the edge set never changes and the dynamics are only due to the placement of stimuli. In Section~\ref{sec:reconfig}, we show that the same proofs apply with little modification to the reconfigurable case as well.

In order to define our main theorems, we must first define the \emph{state invariant}, that we know holds from an initial configuration where every agent is initialized in the \unaware{} state, as we show Lemma~\ref{lem:stateinvariant}.
In the remainder of this paper, a \emph{component} will refer to a connected component of the subgraph induced in $G$ by the set of \aware{} agents.
\begin{definition}[State Invariant]
\label{dfn:stateinvariant}
We say a component 
%of a configuration 
satisfies the {\em state invariant} if it contains at least one agent in the states $\stateAW$, $\stateAAW$~or $\stateAC$. A configuration satisfies the state invariant if every component (if any) of the configuration satisfies the state invariant.
\end{definition}

\begin{lemma}
\label{lem:stateinvariant}
If the current configuration satisfies the state invariant, then all subsequent configurations reachable by Algorithm~\ref{alg:stimulusalgorithm} also satisfy the state invariant.
\end{lemma}

\begin{proof}%[Proof of Lemma~\ref{lem:stateinvariant}]
Starting from configuration where the state invariant currently holds, let $u$ be the next agent to be activated.
%We show that the state invariant continues to hold after both the state change step and the agent movement step of the algorithm.
%
If $u$ is a witness but $u.state\not\in\{\stateAW,\stateAAW\}$, switching $u$ to state $\stateAW$ does not affect the state invariant.
Conversely, if $u$ is not a witness, but $u.state\in\{\stateAW,\stateAAW\}$
%If whether $u$ is a witness does not match its witness flag (i.e. whether $u.state \in\{\stateAW,\stateAAW\}$), switching an agent to $\stateAW$ does not affect the state invariant. However, in the case where an agent is no longer a witness and switches from $\stateAW$~or $\stateAAW$~to the \unaware{} state, this
switching $u$ to state $\stateU$
can potentially split the component it is in into multiple components. However, as all neighbors of $u$ will also be set to state $\stateAC$, each of these new components will contain an agent in state $\stateAC$.

Otherwise, if $u.state=\stateU$, it only switches to the \aware{} state if it neighbors another \aware{} agent. As the component $u$ joins must contain an agent in states $\stateAW$, $\stateAAW$~or $\stateAC$, the new configuration will continue to satisfy the state invariant.
If $u.state=\stateAC$, similar to the earlier case where $u$ is not a witness but $u.state\in\{\stateAW,\stateAAW\}$, activating $u$ may split the component it is in. As before, all neighbors of $u$ will be set to state $\stateAC$, so each of these newly created components satisfy the state invariant.
The remaining possible cases only toggle the \alerttoken{} flag, which does not affect the state invariant.
\end{proof}
\noindent 
This allows us to state our main results in the static graph setting.
We let $T \in \mathbb{N}$ represent the time when the stimuli stabilize long enough to converge (where $T$ is unknown to the agents) and show that we will have efficient convergence. Without loss of generality, for the sake of our proofs, we will assume that $\witset_t = \witset_T$, for all $t \geq T$, 
although this is really just representing a phase where the stimuli are stable.
\begin{theorem}
\label{theorem:mainresultnostimuli}
Starting from any configuration satisfying the state invariant over a static connected graph topology $G$,
if $|\witset_T| = 0$, then all agents will reach and remain in the \unaware{} state in 
$O(n^2)$ expected iterations, after time $T$.
\end{theorem}

\begin{theorem}
\label{theorem:mainresultwithstimuli}
Starting from any configuration satisfying the state invariant over a static connected graph topology $G$,
if $|\witset_T| > 0$, then all agents will reach and stay in the \aware{} state in 
$O(n^5)$ expected iterations, after time $T$.
\end{theorem}

%T: Elimination of Residuals
The proof of these theorems relies on carefully eliminating {residual} aware agents from previous broadcasts.
%, which are agents that remain \aware{} from previous broadcast waves.
%alerts. 
These agents have yet to receive an \cleartoken{} and thus will take some time before they can return to the \unaware{} state.

\begin{figure}[!h]{}
\begin{center}
  \begin{tikzpicture}[x=0.8cm,y=0.8cm]
  \draw[black, line width=0.4mm] (0.191511,0.160697) -- (0.804347,0.674927);
\draw[black, line width=0.4mm] (0.216506,-0.125) -- (1.24708,-0.72);
\draw[black, line width=0.4mm] (1.06289,0.594777) -- (1.39655,-0.604153);
\draw[black, line width=0.4mm] (1.23078,0.750119) -- (2.22686,0.387578);
\draw[black, line width=0.4mm] (1.6277,-0.656409) -- (2.29766,0.113481);
\draw[black, line width=0.4mm] (2.65329,0.462769) -- (3.4653,1.14412);
\draw[black, line width=0.4mm] (3.84832,1.14412) -- (4.66033,0.462769);
\draw[black, line width=0.4mm] (4.78551,0.0610323) -- (4.60251,-0.60396);
\draw[black, line width=0.4mm] (3.24608,-1.07248) -- (4.28998,-0.888412);
\draw[black, line width=0.4mm] (1.70978,-0.888412) -- (2.75368,-1.07248);
\draw[black, line width=0.4mm] (2.55048,0.0683371) -- (2.91118,-0.882156);
\draw[black, line width=0.4mm] (2.71178,0.302072) -- (4.60184,0.302072);
\draw[black, line width=0.4mm] (0,0) circle (0.25);
\node[align=left] at (0,-0.5) {\scriptsize $\stateA$};
\draw[black, line width=0.4mm] (0.995858,0.835624) circle (0.25);
\node[align=left] at (0.845858,0.335624) {\scriptsize $\stateA$};
\draw[black, line width=0.4mm] (1.46358,-0.845) circle (0.25);
\node[align=left] at (1.46358,-1.345) {\scriptsize $\stateAW$};
\draw[black, line width=0.4mm] (2.46178,0.302072) circle (0.25);
\node[align=left] at (2.36178,-0.197928) {\scriptsize $\stateU$};
\draw[black, line width=0.4mm] (2.99988,-1.11589) circle (0.25);
\node[align=left] at (2.99988,-1.61589) {\scriptsize $\stateU$};
\draw[black, line width=0.4mm] (3.65681,1.30482) circle (0.25);
\node[align=left] at (3.65681,0.804821) {\scriptsize $\stateAW$};
\draw[black, line width=0.4mm] (4.85184,0.302072) circle (0.25);
\node[align=left] at (5.05184,-0.197928) {\scriptsize $\stateA$};
\draw[black, line width=0.4mm] (4.53618,-0.845) circle (0.25);
\node[align=left] at (4.53618,-1.345) {\scriptsize $\stateAC$};
\draw[black!20!red, line width=0.6mm, ] (3.65681,1.30482) circle (0.36);
  \end{tikzpicture}
  \end{center}
  \vspace{-.2in}
  \caption{A configuration with two residual components.
  %(maximal connected sets of \aware{} agents). 
  The component on the left is a residual component despite having a witness in it because it contains an agent with the \allclear~flag; the component on the right is a residual component since it contains an agent in state $\stateAC$.}
\end{figure}
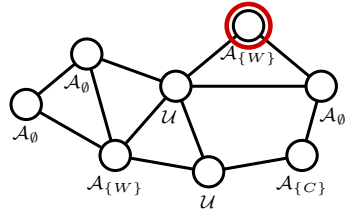

\begin{definition}[Residuals]
A \emph{residual component} is a component that satisfies at least one of the following two criteria:
\begin{enumerate}
\item It contains an agent in state $\stateAC$.%~(\cleartoken).
\item It contains an agent in state $\stateAW$~or $\stateAAW$~that is not a witness.
%\item It contains a witness that is not in state $\stateAW$~or $\stateAAW$.
%\item NOTE 1: If no residuals and no witnesses, all agents must be in the unaware state.
%\item NOTE 2: Every residual component has an agentwhich when activated decreases the potential by 1.
\end{enumerate}
We call agents belonging to residual components \emph{residuals}.
\end{definition}

When starting from arbitrary configurations, 
there is likely to be a large number of residual components that need to be cleared out.
Residuals are problematic as they do no stay in the aware state in the long term. This means they do not actually contribute to the main cluster of \aware{} state agents, and the presence of residuals causes even more residuals to form. Furthermore, later on when we allow the graph to reconfigure itself, residuals may obstruct \unaware{} state agents from coming into contact with non-residual components, which may prevent \alerttokens{} from reaching them.
Our main tool to show that all residuals eventually vanish is a {\em potential function} that decreases more quickly than it increases.

\begin{definition}[Potential]
\label{defn:potential}
For a configuration $\sigma$, we define its potential $\Phi(\sigma)$ as
$$\Phi(\sigma) := \Phi_{A}(\sigma) + \Phi_{AT}(\sigma)$$
where $\Phi_{A}(\sigma)$ and $\Phi_{AT}(\sigma)$ represent total the number of \aware{} agents and the number of \aware{} agents with an \alerttoken{} respectively.
\end{definition}
We use the following lemma 
%(Lemma~\ref{lemma:generallemmafinal}) 
to guarantee, in Lemma~\ref{lemma:reachpotentialzero}, that the expected number of steps before all residuals are removed is polynomial.
%The proof of Lemma~\ref{lemma:generallemmafinal} is given in Appendix~\ref{appendix:staticdetails}. 
%
%<INSERT> {Proof of Lemma~\ref{lemma:generallemmafinal}}
%
\begin{lemma}
\label{lemma:generallemmafinal}
Assume $n \geq 2$, and let $0 < \eta < 1$. Consider two random sequences of probabilities $(p_t)_{t \in \mathbbNzero}$ and $(q_t)_{t \in \mathbbNzero}$, with the properties that $\frac{1}{n} \leq p_t \leq 1$ and $0 \leq q_t \leq \frac{\eta}{n}$, and $p_t + q_t \leq 1$. Now consider a sequence $(X_t)_{t \in \mathbbNzero}$, where
%Note: $X_{t+1} \leq X_t-1$ rather than equal. Proofs have been adjusted, check.
\begin{align*}
X_{t+1} \begin{cases}
\leq X_t-1 &\text{with probability $p_t$}\\
= X_t+1 &\text{with probability $q_t$}\\
= X_t &\text{with probability $1-p_t-q_t$}.
\end{cases}.
\end{align*}
Then $\ex\left[\min\{t \geq 0 \mid X_t = 0\}\right] \leq \frac{nX_0}{1-\eta}$.
\end{lemma}

\begin{proof}
For each $k \in \mathbbNzero$, we define a random sequence $(Y^{(k)}_t)_{t \in \mathbbNzero}$ such that $Y^{(k)}_0 = k$ and
\begin{align*}
Y^{(k)}_{t+1} \begin{cases}
= Y^{(k)}_t-1 &\text{with probability $\frac{1}{n}$}\\
= Y^{(k)}_t+1 &\text{with probability $\frac{\eta}{n}$}\\
= Y^{(k)}_t &\text{otherwise}.
\end{cases}
\end{align*}
For each such $k$, let $S_k := \ex\left[\min\{t \geq 0 \mid Y^{(k)}_t = 0\}\right]$.
Let $T^{(k)}_{k+1} = 0$ and for $i \in \{k,k-1,\dots,1\}$, let $T^{(k)}_i = \min\{t \geq 0 \mid X^{(k)}_t \leq i\} - T^{(k)}_{i+1}$ denote the number of time steps after $T^{(k)}_{i+1}$ before the first time step $t$ where $Y^{(k)}_t \leq i$. We observe that each $T^{(k)}_i$ is identically distributed, with $\ex T^{(k)}_i = \ex T^{(k)}_1 = S_1$.
Also, we observe that $S_k = \ex [\sum_{i=1}^{k} T^{(k)}_i]$, so $S_k = k\cdot S_1$ for all $k$.
We can thus compute $S_1$ by conditioning on the first step:
\begin{align*}
&S_1 = \frac{1}{n}(1) + \frac{\eta}{n}(S_2+1) + \left(1-\frac{1+\eta}{n}\right)(S_1+1) = 1 + \frac{\eta}{n}\cdot 2S_1 + \left(1-\frac{1+\eta}{n}\right) S_1
%\implies &S_1 = \frac{n}{1 - \eta}.
\end{align*}
This implies $S_1 = \frac{n}{1-\eta}$ and thus
 $\ex\left[\min\{t \geq 0 \mid Y^{(k)}_t = 0\}\right] = S_k = \frac{kn}{1-\eta}$.

We can then couple $(X_t)_{t \in \mathbbNzero}$ and $(Y^{(X_0)}_t)_{t \in \mathbbNzero}$ in a way such that $Y^{(X_0)}_{t+1} = Y^{(X_0)}_t + 1$ whenever $X_{t+1} > X_t$, and $Y^{(X_0)}_{t+1} = Y^{(X_0)}_t-1$ whenever $X_{t+1} < X_t$. We thus have $Y^{(X_0)}_t \geq X_t$ always, and so $\ex\left[\min\{t \geq 0 \mid X_t = 0\}\right] \leq \ex\left[\min\{t \geq 0 \mid Y^{(X_0)}_t = 0\}\right] \leq \frac{kX_0}{(1-\eta)}$.
\end{proof}

We show in Lemma~\ref{lemma:reachpotentialzero} that after a polynomial number of steps in expectation, we will reach a configuration with no residual components. 
%Lemma~\ref{lemma:generallemmafinal} allows us next to bound the time to remove residuals.
%
%<INSERT> {Elimination of Residuals}
%

\begin{lemma}
\label{lemma:reachpotentialzero}
We start from a configuration satisfying the state invariant over a static connected graph $G$ with no more than $w$ witnesses at any point.
%We start from a configuration of the dynamic stimuli problem where the state invariant currently holds.
%Suppose that there are no more than $w$ witnesses. 
Then the expected number of steps before we reach a configuration with no residual components is at most ${2n^2}/{(1-wp)}$.
\end{lemma}

\begin{proof}
We apply Lemma~\ref{lemma:generallemmafinal} to the sequence of potentials $(\Phi(\sigma_t))_{t \in \mathbbNzero}$ where $\sigma_t$ is the configuration after iteration $t$.
%
%By the definition of a residual component, 
As long as there exists a residual component, there will be at least one agent will switch to the \unaware{} state on activation. This gives a probability of at least $1/n$ in any iteration of decreasing the current potential by at least $1$.

%We then look at what can cause the potential to increase.
%
There are only two ways for the potential to increase. The first is when a new \alerttoken{} is generated by a witness, and the second is when an \unaware{} witness switches to an \aware{} state.
The activation of a witness thus increases the current potential by exactly $1$, with probability $p$.
%Each of these increase the current potential by exactly $1$.
%Both of these can only happen through the activation of witnesses, and only happen with probability $p$.
As there are at most $w$ witnesses,
this happens
%the potential increases (by at least $1$)
with probability at most ${wp}/{n} < {1}/{n}$.
Note that the consumption of an \alerttoken{} to add a new \aware{} state agent to a residual component does not change the current potential. Neither does switching agents to the \cleartoken{} state affect the potential.

By Lemma~\ref{lemma:generallemmafinal}, as $\Phi(\sigma_0) \leq 2n$, within ${2n^2}/{(1-wp)}$ steps in expectation, we will either reach a configuration $\sigma$ with $\Phi(\sigma) = 0$, or a configuration with no residual components, whichever comes first. Note that if $\Phi(\sigma) = 0$, then $\sigma$ cannot have any residual components, completing the proof.
\end{proof}

We now show that as long as no agent is removed from the witness set, after all residual components are eliminated, no new ones will be generated:
%(Lemma~\ref{lemma:residualscantregenerate}).

\begin{lemma}
\label{lemma:residualscantregenerate}
We start from a configuration satisfying the state invariant over a static connected graph $G$, and
%We start from a configuration of the dynamic stimuli problem where the state invariant currently holds.
assume that no agent will be removed from the witness set from the current point on.
If there are currently no residuals, then a residual cannot be generated.
\end{lemma}

\begin{proof}
With no residual components in the current iteration, there will be no agents in state $\stateAC$, and all agents in state $\stateAW$~or $\stateAAW$~will be witnesses. This means that no agent on activation will switch another agent to the $\stateAC$~state. All agents in states $\stateAW$~or $\stateAAW$~will continue to be witnesses by assumption of the lemma, and agents will only switch to states $\stateAW$~or $\stateAAW$~if they are witnesses. Thus no component will be residual in the next iteration.
%With no residual components, all agents in the configuration will have the food bit set if and only if they are adjacent to food. There are also no $DT$ agents, so no agent on activation will switch to the $DT$ state. When an agent moves, it will set its food bit accordingly, so no residual components will form.
\end{proof}

When $|\witset_T| = 0$ and no witnesses exist in the long term, if the state invariant holds, then all agents will be in the \unaware{} state, giving us Theorem~\ref{theorem:mainresultnostimuli}.
%
% Note: The state invariant is still needed for this as residuals
On the other hand, in order to show Theorem~\ref{theorem:mainresultwithstimuli},if a witness persists in the long term ,
%(i.e., $|\witset_T| > 0$), 
 we need to show that all agents eventually switch to the \aware{} state. Lemma~\ref{lemma:allswitchtoawarestate} establishes this as without residuals, no \aware{} state agent can revert to the \unaware{} state.
%This proof can be found in Appendix~\ref{appendix:staticdetails} and,
%With Lemmas~\ref{lemma:reachpotentialzero} and~\ref{lemma:residualscantregenerate}, it allows us to prove Theorem~\ref{theorem:mainresultwithstimuli}.

\begin{lemma}
\label{lemma:allswitchtoawarestate}
We start from a configuration satisfying the state invariant over a static connected graph $G$, and
assume that there are no residual agents, 
the witness set is nonempty and no agent will be removed from the witness set from the current point on.
%In the dynamic stimuli problem, suppose that the witness set is nonempty and no agent will be removed from the witness set from the current point on. Furthermore, assume that the state invariant holds and that there are no residual agents in the configuration.
Then the expected number of iterations before the next agent switches from the \unaware{} to the \aware{} state is at most $O(n^4)$.
\end{lemma}

\begin{proof}
In the case where there is no \aware{} agent, 
there must be an \unaware{} witness, which on activation switches to the \aware{} state with probability $p$.
The expected time before this happens is no more than $O(\frac{n}{p})$. Thus for the rest of the proof, we may assume every witness is already in the \aware{} state with the witness flag set.

If we have a bound on the expected time before we reach a configuration where every \aware{} agent holds an \alerttoken{}, all it takes following that is for any \unaware{} agent to be activated while holding an \alerttoken{}. 
As the graph $G$ is connected with least one \unaware{} and one \aware{} agent,
there will be some \unaware{} agents neighboring an \aware{} agent, so the expected number of iterations before this occurs will is $O(n)$.

%We consider the amount of time it takes for the next agent to switch to the \aware{} state.
%For a non-witness to switch to the \aware{} state, it must be activated while adjacent to an agent holding an \alerttoken. We can upper bound the expected amount of time before this happens by the amount of time it takes for all \aware{} state agents to obtain an \alerttoken, followed by the amount of time it takes for an agent to be activated while adjacent to an agent with an \alerttoken{} and switch to the \aware{} state.

We now bound the amount of time it takes for all \aware{} state agents to hold \alerttokens{}.
As an agent can only hold one \alerttoken{} at a time, a new \alerttoken{} can only be generated when a witness does not hold an \alerttoken.
Assume that there is at least one \aware{} state agent that does not have an \alerttoken. Mark one such agent. 
%%If $u$ is the marked agent and a neighboring agent $v$ holding an \alerttoken{} is activated and transfers its \alerttoken{} to $u$, we swap the mark to $v$. 
%If an agent $v$ without an \alerttoken{} is activated, for the sake of our analysis, we still have it randomly pick an outgoing neighbor and receive the mark if the chosen neighbor is $u$.
Suppose that $u$ is the marked agent and that the next agent $v$ to be activated is a neighbor of $u$. If $v$ has an \alerttoken{} and $v$ randomly chooses $u$ as its outgoing neighbor (per the algorithm), then $v$ transfers its \alerttoken{} to $u$ and receives the mark from $u$. Otherwise, 
%$v$ does not have an \alerttoken{}, 
for the sake of our analysis, we still have $v$ pick an outgoing neighbor at random and receive the mark if the chosen neighbor is $u$.
%If a neighboring agent that does not have an \alerttoken{} is activated, for the sake of our analysis we can still have it randomly pick an outgoing the neighbor the same way it would if it had an \alerttoken{}, and receive the mark if the marked agent happens to be the picked neighbor (even though no \alerttoken{} is transferred).

%As long as no new agents are switching to the \aware{} state,
The mark moving in this manner is equivalent to following a $d_{max}$-random walk over the subgraph induced by the \aware{} agents, which is static as long as no new agents are switching to the \aware{} state.
As the configuration satisfies the state invariant and there are no residuals, the component of \aware{} agents the mark is in must contain at least one witness.
The worst case hitting time of the $d_{max}$-random walk over this subgraph is $O(n^2)$, %movements of the mark, 
and thus the expected number of iterations before the mark lands on a witness is $O(n^3)$, allowing a new \alerttoken{} to be generated
%, is $O(n^3)$ 
(this new \alerttoken{} is generated with a constant probability $p \in (0,1)$).
As there are at most $n$ \aware{} agents, all \aware{} agents will be holding an \alerttoken{} after $O(n^4)$ iterations in expectation, which translates to an expected time bound of $O(n^4)$ before a new \aware{} agent is added.
Note that this is a loose bound - the bound has not been optimized for clarity of explanation.
\end{proof}

%<INSERT> {Proof of Lemma~\ref{lemma:allswitchtoawarestate}

\section{Reconfigurable topologies}
%\section{Responding to stimuli dynamic graphs}
\label{sec:reconfig}

We show that the same results hold if we allow some degree of reconfigurability of the edge set and relax the requirement that the graph must be connected. This gives us enough flexibility to implement a wider range of behaviors, like free aggregation or compression and dispersion in Section~\ref{sec:foraging}.
%, for agents in different states. 
When needed, we may refer to this as the \emph{reconfigurable dynamic stimuli problem}, in order to clearly differentiate from the dynamic stimuli problem on static graphs that we considered in Section~\ref{sec:static}.

%\subsection{Allowing for reconfigurations and disconnections}
%We start from the dynamic stimuli problem defined in Section~\ref{section:model}. 
%This time however, i
Instead of a static graph $G$, as we considered in Section~\ref{sec:static}, we now allow the edge set of the graph to be locally modified over time.
These reconfigurations can be initiated by the agents themselves or controlled by an adversary, and they can be randomized or deterministic, but we require some restrictions on what reconfigurations are allowed, and what information an algorithm carrying out these reconfigurations may have access to.
This will result in a restricted class of dynamic graphs, but will be general enough to be applied to the problem of foraging that we describe in Section~\ref{sec:foraging}.
%but cannot be decided based on state knowledge beyond what decides their reconfiguration behaviors.

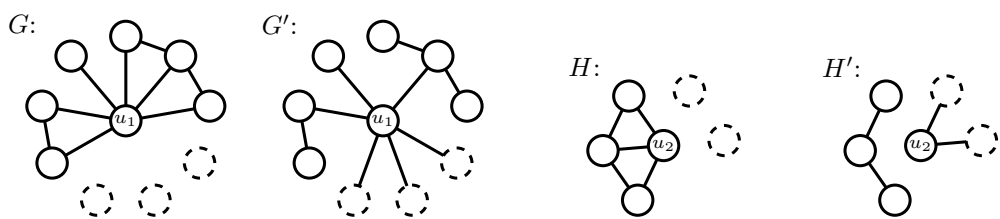
\begin{figure}
\begin{minipage}[b]{.48\linewidth}
\begin{center}
\begin{subfigure}[b]{0.5\linewidth}
  \begin{center}
  \begin{tikzpicture}[x=0.8cm,y=0.8cm]
  \node[align=left] at (-1.7,1.6) {\normalsize $G$:};
\draw[black, line width=0.4mm] (1.53081e-17,0.25) -- (7.04172e-17,1.15);
\draw[black, line width=0.4mm] (0.160697,0.191511) -- (0.739206,0.880951);
\draw[black, line width=0.4mm] (0.246202,0.043412) -- (1.13253,0.199695);
\draw[black, line width=0.4mm] (-0.216506,-0.125) -- (-0.995929,-0.575);
\draw[black, line width=0.4mm] (-0.246202,0.043412) -- (-1.13253,0.199695);
\draw[black, line width=0.4mm] (-0.160697,0.191511) -- (-0.739206,0.880951);
\draw[black, line width=0.4mm] (-1.25585,-0.453798) -- (-1.33532,-0.00309449);
\draw[black, line width=0.4mm] (0.234923,1.31449) -- (0.664979,1.15797);
\draw[black, line width=0.4mm] (1.0249,0.855956) -- (1.25373,0.459614);
\draw[black, line width=0.4mm] (0,0) circle (0.25);
\node[align=left] at (0,0) {\scriptsize $u_1$};
\draw[black, line width=0.4mm] (8.57253e-17,1.4) circle (0.25);
\draw[black, line width=0.4mm] (0.899903,1.07246) circle (0.25);
\draw[black, line width=0.4mm] (1.37873,0.243107) circle (0.25);
\draw[black, dashed, line width=0.4mm] (1.21244,-0.7) circle (0.25);
\draw[black, dashed, line width=0.4mm] (0.478828,-1.31557) circle (0.25);
\draw[black, dashed, line width=0.4mm] (-0.478828,-1.31557) circle (0.25);
\draw[black, line width=0.4mm] (-1.21244,-0.7) circle (0.25);
\draw[black, line width=0.4mm] (-1.37873,0.243107) circle (0.25);
\draw[black, line width=0.4mm] (-0.899903,1.07246) circle (0.25);
  \end{tikzpicture}
  \end{center}
\end{subfigure}%
\begin{subfigure}[b]{0.5\linewidth}
  \begin{center}
  \begin{tikzpicture}[x=0.8cm,y=0.8cm]
  \node[align=left] at (-1.7,1.6) {\normalsize $G^\prime$:};
\draw[black, line width=0.4mm] (0.160697,0.191511) -- (0.739206,0.880951);
\draw[black, line width=0.4mm] (-0.246202,0.043412) -- (-1.13253,0.199695);
\draw[black, line width=0.4mm] (-0.160697,0.191511) -- (-0.739206,0.880951);
\draw[black, line width=0.4mm] (-1.25585,-0.453798) -- (-1.33532,-0.00309449);
\draw[black, line width=0.4mm] (0.234923,1.31449) -- (0.664979,1.15797);
\draw[black, line width=0.4mm] (1.0249,0.855956) -- (1.25373,0.459614);
\draw[black, line width=0.4mm] (0.216506,-0.125) -- (0.995929,-0.575);
\draw[black, line width=0.4mm] (0.085505,-0.234923) -- (0.393323,-1.08065);
\draw[black, line width=0.4mm] (-0.085505,-0.234923) -- (-0.393323,-1.08065);
\draw[black, line width=0.4mm] (0,0) circle (0.25);
\node[align=left] at (0,0) {\scriptsize $u_1$};
\draw[black, line width=0.4mm] (8.57253e-17,1.4) circle (0.25);
\draw[black, line width=0.4mm] (0.899903,1.07246) circle (0.25);
\draw[black, line width=0.4mm] (1.37873,0.243107) circle (0.25);
\draw[black, dashed, line width=0.4mm] (1.21244,-0.7) circle (0.25);
\draw[black, dashed, line width=0.4mm] (0.478828,-1.31557) circle (0.25);
\draw[black, dashed, line width=0.4mm] (-0.478828,-1.31557) circle (0.25);
\draw[black, line width=0.4mm] (-1.21244,-0.7) circle (0.25);
\draw[black, line width=0.4mm] (-1.37873,0.243107) circle (0.25);
\draw[black, line width=0.4mm] (-0.899903,1.07246) circle (0.25);
  \end{tikzpicture}
  \end{center}
\end{subfigure}
\end{center}
%\caption{A locally reconnected reconfiguration of a vertex $u_1$ in a graph $G$}
\end{minipage}%
~~
\begin{minipage}[b]{.48\linewidth}
\begin{center}
\begin{subfigure}[b]{0.5\linewidth}
  \begin{center}
  \begin{tikzpicture}[x=0.8cm,y=0.8cm]
  \node[align=left] at (-1.3,1.3) {\normalsize $H$:};
\draw[black, line width=0.4mm] (-0.106445,-0.226207) -- (-0.319334,-0.67862);
\draw[black, line width=0.4mm] (-0.249123,-0.0209195) -- (-0.74737,-0.0627584);
\draw[black, line width=0.4mm] (-0.142678,0.205287) -- (-0.428035,0.615862);
\draw[black, line width=0.4mm] (-0.568458,-0.69954) -- (-0.853814,-0.288965);
\draw[black, line width=0.4mm] (-0.890048,0.142529) -- (-0.677158,0.594942);
\draw[black, line width=0.4mm] (0,0) circle (0.25);
\node[align=left] at (0,0) {\scriptsize $u_2$};
\draw[black, line width=0.4mm] (-0.425779,-0.904827) circle (0.25);
\draw[black, line width=0.4mm] (-0.996493,-0.0836778) circle (0.25);
\draw[black, line width=0.4mm] (-0.570714,0.821149) circle (0.25);
\draw[black, dashed, line width=0.4mm] (0.425779,0.904827) circle (0.25);
\draw[black, dashed, line width=0.4mm] (0.996493,0.0836778) circle (0.25);
  \end{tikzpicture}
  \end{center}
\end{subfigure}%
\begin{subfigure}[b]{0.5\linewidth}
  \begin{center}
  \begin{tikzpicture}[x=0.8cm,y=0.8cm]
  \node[align=left] at (-1.3,1.3) {\normalsize $H^\prime$:};
\draw[black, line width=0.4mm] (-0.568458,-0.69954) -- (-0.853814,-0.288965);
\draw[black, line width=0.4mm] (-0.890048,0.142529) -- (-0.677158,0.594942);
\draw[black, line width=0.4mm] (0.106445,0.226207) -- (0.319334,0.67862);
\draw[black, line width=0.4mm] (0.249123,0.0209195) -- (0.74737,0.0627584);
\draw[black, line width=0.4mm] (0,0) circle (0.25);
\node[align=left] at (0,0) {\scriptsize $u_2$};
\draw[black, line width=0.4mm] (-0.425779,-0.904827) circle (0.25);
\draw[black, line width=0.4mm] (-0.996493,-0.0836778) circle (0.25);
\draw[black, line width=0.4mm] (-0.570714,0.821149) circle (0.25);
\draw[black, dashed, line width=0.4mm] (0.425779,0.904827) circle (0.25);
\draw[black, dashed, line width=0.4mm] (0.996493,0.0836778) circle (0.25);
  \end{tikzpicture}
  \end{center}
\end{subfigure}
\end{center}
%\caption{test}
\end{minipage}
\caption{Two locally connected reconfigurations of the vertices $u_1$ (from $G$ to $G'$) and $u_2$ (from $H$ to $H'$) respectively, where only the 
%restriction of the graph to the 
\aware{} neighbors of the reconfigured vertex are shown. Vertices with dashed outlines are newly introduced neighbors.}% (not in the original neighborhood).}
%not in the original neighborhood of the reconfigured vertex (new neighbors).}
\label{fig:locallyconnected}
\end{figure}

The basic primitive for (local) reconfiguration of our graph by an agent $u$ is  replacing the edges incident to vertex $u$ with new edges. We call this {a reconfiguration of vertex $u$} and  define {local connectivity} to formalize  which reconfigurations are allowed.

\begin{definition}[Locally Connected Reconfigurations]
\label{defn:locallyconnected}
For any graph $G=(V,E)$, let $G'$ be a the graph resulting from a reconfiguration of vertex $u \in V$
%, and let $G'$ be the resulting graph. 
and let $N_\stateAbase(u)$ be the \aware{} neighbors of $u$ in $G$.
We say that this reconfiguration is {\it locally connected} if $u$ has at least one \aware{} neighbor in $G'$ 
and if for every pair of vertices $v_1, v_2$ in $N_\stateAbase(u)$ with a path from $v_1$ to $v_2$ in the induced subgraph $G[N_\stateAbase(u)\cup\{u\}]$, there is also a path from $v_1$ to $v_2$ in $G'[N_\stateAbase(u)\cup\{u\}]$. 
\end{definition}
Examples of locally connected reconfigurations are given in Figure~\ref{fig:locallyconnected}.

In the reconfigurable dynamic stimuli problem, we want to be able to define reconfiguration behaviors for agents in the \aware{} (without \cleartoken) and \unaware{} states.
To define what reconfigurations of an agent $u$ are valid, we group our set of agent states into three subsets which we refer to as \emph{\behaviorgroups}. The three \behaviorgroups{} are $\stateBU:=\{\stateU\}$, $\stateBA:=\{\stateA, \stateAA\}$ and $\stateBI:=\{\stateAW, \stateAAW, \stateAC\}$, which we call \unawareb, \mobileb{} and \immobileb{} respectively (the latter two referring to \aware{} state agents which are and are not allowed to change their neighboring edges respectively).
We say that a locally connected reconfiguration of an agent $u$ is \emph{valid} if it is not allowed to reconfigure in the \immobileb{} \behaviorgroup, and if $u$ is in the \mobileb{} \behaviorgroup, the reconfiguration of $u$ must be locally connected (Definition~\ref{defn:locallyconnected}). We show the following lemma:% in Appendix~\ref{appendix:reconfigurabledetails}:
%We use a function $\bstates:V \to \{\stateBA, \stateBU, \stateBI\}$ to denote the \behaviorgroup{} vector of a configuration.

\begin{lemma}
\label{lem:stateinvariantholdsdynamic}
Let $G=(V,E)$ be a graph and let $G' = (V,E')$ be the graph resulting from a valid locally connnected reconfiguration of a vertex $u \in V$.
If a configuration satisfies the state invariant on $G$, then the same state assignments satisfy the state invariant on $G'$.
\end{lemma}
\begin{proof}
%Denote by $\bstates:V \to \{\stateBA, \stateBU, \stateBI\}$ the current \behaviorgroup~vector of the configuration.
As locally connected reconfigurations of \unawareb{} agents do not affect the invariant and \immobileb{} agents cannot be reconfigured, it suffices to show that locally connected reconfigurations of \mobileb{} agents maintain the state invariant.
%, we make use of the requirement that these reconfigurations require local connectivity.

% u is reconfigured.
%Suppose that an \mobileb{} agent $u$ is to be reconfigured.
To show that the state invariant holds on $G'$, we show that any \mobileb{} agent has a path to an \immobileb{} agent in $G'$. Consider any such \mobileb{} agent $v \neq u$.
%As $(G, \bstates)$ satisfies the state invariant,
As the state invariant is satisfied on $G$, 
there exists a path in $G$ over \aware{} vertices from $v$ to an \immobileb{} agent $w$.
If this path does not contain the agent $u$, $v$ has a path to $w$ in $G'$.
On the other hand, if this path contains the agent $u$, consider the vertices $u_1$ and $u_2$ before and after $u$ respectively in this path. By local connectivity, there must still be a path from $u_1$ to $u_2$ in the induced subgraph $G'[N_\stateAbase(u) \cup \{u\}]$, and so a path exists from $v$ to $w$ over $G'$.
It remains to check that $u$ also has a path to an \immobileb{} agent in $G'$. Once again by local connectivity, $u$ must have an \aware{} neighbor in $G'$, which must have a path to an \immobileb{} agent over $G'$.
\end{proof}

In the reconfigurable version of the dynamic stimuli problem, we have a random sequence of graphs $(G_0,G_1,G_2,\ldots)$ where $G_t$ for $t \geq 1$ denotes the graph used in iteration $t$. These graphs share a common vertex set $V$, the set of agents, but the edges may change from iteration to iteration. We do not consider fully arbitrary sequence of graphs, but instead one that is generated by what we call a \emph{(valid) reconfiguration adversary} $\mathcal X$.
%
%In general, as our ultimate goal is to have the flexibility to execute different behaviors for agents from different \behaviorgroups{} (in particular \unawareb{} and \mobileb), we may see the reconfiguration behavior of the graph as one controlled by an adversary
Let the random sequence $(X_0,X_1,X_2, \ldots)$ denote the information available to the reconfiguration adversary on each iteration, and for each $t \geq 1$ we let the vector $\bstates_t : V \to \{\stateBU, \stateBA, \stateBI\}$ denote the \behaviorgroups{} of the agents at the end of iteration $t$ (i.e., after the activated agent performed any computation/change of states of its neighborhood at iteration $t$).
At iteration $t$, before an agent is activated in the stimuli algorithm, the new graph $G_t$ and the next value $X_t$ of the sequence are drawn as a pair from the distribution $\mathcal{X}(X_{t-1}, \bstates_{t-1})$,
%We call this scheduler a \emph{valid (locally connected) reconfiguration adversary} if for each $t \in \{1,2,3,\ldots\}$, 
which assigns non-zero probabilities only  to
%with the condition that the only graphs that can be drawn from $\mathcal{X}(X_{t-1}, \bstates_{t-1})$ (with non-zero probability) are those 
graphs that can be obtained through
some sequence of valid locally connected reconfigurations of the vertices of $G_{t-1}$.

We note that the reconfiguration adversary can be deterministic or randomized (it can even be in control of the agents themselves), and is specifically defined to act based on the behavior group vectors $\bstates_t : V \to \{\stateBU, \stateBA, \stateBI\}$ and  {\em not} on the state vector of the agents. We explicitly do not give the reconfiguration adversary access to full state information, as convergence time bounds require that the reconfiguration adversary of the graph be oblivious to the movements of \alerttokens.  
%Our way of defining the reconfiguration adversary %allows graph reconfigurations to be randomized or %adversarial, while not giving the adversary enough %power to adaptively reconfigure the graph to ``trap'' %\alerttokens{} moving around the system. 
As a special case, our results hold for any sequence of graphs $(G_0,G_1,G_2,\ldots)$ pre-determined by an oblivious adversary.
An example of a valid reconfiguration adversary that takes full advantage of the generality of our definition can be seen in the Adaptive $\alpha$-Compression Algorithm, an algorithm that we will later introduce to solve the problem of foraging.

In the static version of the problem, the graph is required to be connected to ensure that agents will always be able to communicate with each other. Without this requirement, we can imagine simple examples of graphs or graph sequences where no algorithm will work. In particular, if a set of agents that contains a witness never forms an edge to an agent outside of the set, there would be no way to transmit information about the existence of the witnesses to the nodes outside the set.
However, as the foraging problem will require disconnections to some extent, we relax this requirement that each graph $G_t$ is connected, and instead quantify how frequently \unaware{} state agents come into contact with \aware{} state agents.

%Remark: This is only for alerting, not useful for reverting
We say an \unaware{} agent is \emph{active} if it is adjacent to an \aware{} agent.  One way to quantify how frequently agents become active is to divide the iterations into ``batches'' of bounded expected duration, with at least some amount of active agents in each batch.
%where the total number of active agents in each batch in expectation is lower bounded by some value.
The random variables $(D_1, D_2, D_3 \ldots)$ denote the durations (in iterations) of these batches, and the random variables $(C_1, C_2, C_3, \ldots)$ denote the number of active agents in the respective batches. Definition~\ref{dfn:recurring} formalizes this notion.

\begin{definition}[Recurring Sequences]% and Reconfigurable Graphs]
\label{dfn:recurring}
Let $\mathcal{X}$ be a fixed valid reconfiguration adversary.
We say that this $\mathcal{X}$ is {\em $(U_D,U_C)$-recurring} (for $U_D \geq 1$ and $U_C \in (0,1)$) if for each possible starting iteration $t$ and fixed \behaviorgroup{} $\bstates$ with at least one \unawareb{} and one \immobileb{} agent, we have the following property:
There exists sequences of random variables $(D_1, D_2, D_3 \ldots)$ and $(C_1, C_2, C_3 \ldots)$ where for each $k \in \{1,2,3,\ldots\}$,
\begin{enumerate}
\item $C_k$ denotes the number of active agents under the \behaviorgroup{} $\bstates$ between iterations $t + \sum_{i=1}^{k-1}D_i$ and $t + \left(\sum_{i=1}^{k}D_i\right)-1$.
\item $\mathbb{E}\left[D_k \mid D_1, D_2, \ldots D_{k-1}, C_1,C_2 \ldots C_{k-1}\right] \leq U_D$.
\item $\mathbb{E}\left[\left(1-\frac{1}{n}\right)^{C_k} \mid D_1, D_2, \ldots D_{k-1}, C_1,C_2 \ldots C_{k-1}\right] \leq U_C$.
\end{enumerate}
\end{definition}

 We can then define a {\em (valid) reconfigurable graph (or sequence)} %$(G_0,G_1,\ldots)$
%, a we call a reconfigurable graph \emph{valid} 
as one that is generated by a valid reconfiguration adversary $\mathcal{X}$ and is $(U_D,U_C)$-recurring for some $U_D \geq 1, U_C \in (0,1)$. 
This allows us to state our main results for reconfigurable graphs as Theorems~\ref{theorem:mainresultnostimulidynamic} and~\ref{theorem:mainresultwithstimulidynamic}.
The theorems we have shown for the static version of the problem (Theorems~\ref{theorem:mainresultnostimuli} and~\ref{theorem:mainresultwithstimuli}) are special cases of these two theorems.

\begin{theorem}
\label{theorem:mainresultnostimulidynamic}
Starting from any configuration satisfying the state invariant over a reconfigurable graph, 
if $|\witset_T| = 0$, then all agents will reach and remain in the \unaware{} state in 
$O(n^2)$ expected iterations, after time $T$.
\end{theorem}
\begin{theorem}
\label{theorem:mainresultwithstimulidynamic}
Starting from any configuration satisfying the state invariant over a reconfigurable graph, if $|\witset_T| > 0$ and the reconfiguration adversary is $(U_D,U_C)$-recurring, 
%(Definition~\ref{dfn:recurring}), 
then all agents will reach and stay in the \aware{} state in 
$O\left(n^6 \log n + \frac{nU_D}{1-U_C}\right)$ expected iterations, after time $T$.
\end{theorem}

In particular, if every graph $G_t$ is connected, then the reconfiguration adversary is $\left(1,\left(1-\frac{1}{n}\right)\right)$-recurring by setting $D_k = C_k = 1$ (as constant random variables) for all $k$, and we have the following corollary:

\begin{corollary}
Starting from any configuration satisfying the state invariant over a reconfigurable graph, if $|\witset_T| > 0$ and every $G_t$ is connected, 
%the reconfiguration sequence is $(U_D,U_C)$-recurring (Definition~\ref{dfn:recurring}), 
then all agents will reach and stay in the \aware{} state in 
$O(n^6 \log n)$ expected iterations, after time $T$.
\label{cor:connected-graphs}
\end{corollary}
Obviously, if $G$ is a static connected graph, it falls as a special case of the corollary; a tighter analysis allowed us to prove the $O(n^5)$ expected convergence bound in Theorem~\ref{theorem:mainresultwithstimuli}. 
%This justifies the earlier statement that Theorems~\ref{theorem:mainresultnostimuli} and~\ref{theorem:mainresultwithstimuli} are a special case of Theorems~\ref{theorem:mainresultnostimulidynamic} and~\ref{theorem:mainresultwithstimulidynamic}.

The following two lemmas, which are analogous to Lemmas~\ref{lemma:reachpotentialzero} and~\ref{lemma:residualscantregenerate} but for reconfigurable graphs, are sufficient to show Theorem~\ref{theorem:mainresultnostimulidynamic} (the case for $|\witset_T| = 0$).
%stating that all agents will return to the \unaware{} state within polynomial expected time in the case where $|\witset_T| = 0$.
The proof of Lemma~\ref{lemma:reachpotentialzerodynamic} is identical to the proof of Lemma~\ref{lemma:reachpotentialzero}, so we only show Lemma \ref{lemma:residualscantregeneratedynamic}.

\begin{lemma}
\label{lemma:reachpotentialzerodynamic}
We start from a configuration satisfying the state invariant over a reconfigurable graph with no more than $w$ witnesses at any point.
Then the expected number of steps before we reach a configuration with no residual components is at most ${2n^2}/{(1-wp)}$.
\end{lemma}
%\begin{proof}
%The proof of this Lemma is identical to the proof of Lemma %\ref{lemma:reachpotentialzero}.
%\end{proof}

\begin{lemma}
\label{lemma:residualscantregeneratedynamic}
We start from a configuration satisfying the state invariant over a reconfigurable graph, and
assume that no agent will be removed from the witness set from the current point on.
If there are currently no residuals, then a residual cannot be generated.
\end{lemma}

\begin{proof}
From the proof of Lemma~\ref{lemma:residualscantregenerate}, we know that state changes of agents do not generate a new residual.
Valid reconfigurations of agents also cannot generate a new residual, as from Lemma~\ref{lem:stateinvariantholdsdynamic}, the state invariant always holds, so all components will always have an agent in state $\stateAC$, $\stateAW$ or $\stateAAW$. Reconfigurations cannot change the fact that there will be no agent of state $\stateAC$, and that all agents in states $\stateAW$ or $\stateAAW$ will be witnesses.
\end{proof}

The key result we will show
%in Appendix~\ref{appendix:reconfigurabledetails}
is Lemma~\ref{lemma:allswitchtoawarestatedynamic}, a loose polynomial-time upper bound for the amount of time before the next agent switches to the \aware{} state.
This gives a polynomial time bound for all agents switching to the \aware{} state when $|\witset_T| > 0$, implying Theorem~\ref{theorem:mainresultwithstimulidynamic}.
\begin{lemma}
\label{lemma:allswitchtoawarestatedynamic}
We start from a configuration satisfying the state invariant over a $(U_D,U_C)$-recurring reconfigurable graph, and assume that there are no residuals, the witness set is nonempty, and no agent will be removed from the witness set from the current point on.
Then the expected number of iterations before the next agent switches from the \unaware{} to the \aware{} state is at most $O\left(n^5 \log n + \frac{U_D}{1-U_C}\right)$.
\end{lemma}
\newcommand{\Tstart}{{T_{\text{start}}}}
\newcommand{\Tfull}{{T_{\text{full}}}}
\begin{proof}[Proof Sketch] 
This proof largely follows the proof of Lemma~\ref{lemma:allswitchtoawarestate}, with some modifications to allow for reconfigurability. In the interest of space however, we will provide only a summary of the key ideas and calculations that go into the proof (which is available in the full version of the paper).
%
%Similar to the proof of Lemma~\ref{lemma:allswitchtoawarestate},
Assuming that every witness is already in the \aware{} state with the witness flag set, we upper bound the expected time it takes for all \aware{} agents to obtain an \alerttoken{}, followed by the time it takes for an agent to be activated while adjacent to an \alerttoken{} and switch to the \aware{} state.

To bound the expected time for all \aware{} agents to obtain an \alerttoken{}, we apply the same strategy, marking an agent without an alert token and passing around the mark until it lands on a witness.
We make use of the result of~\cite{RandomWalksDynamicGraphs, RandomWalksRecurringTopologies}, which states that the expected hitting time of the $d_{max}$-random walk on a connected evolving graph controlled by an \emph{oblivious adversary} is $O(n^3 \log n)$~\cite{RandomWalksRecurringTopologies}.
%This polynomial time bound is notable as there are connected evolving graphs where the simple random walk admits exponential hitting times in the worst case~\cite{RandomWalksDynamicGraphs}.
This corresponds to $O(n^4 \log n)$ iterations in expectation to generate a new \alerttoken{}, which gives an upper bound of $O(n^5 \log n)$ iterations in expectation before all agents carry \alerttokens{}.
%This is clearly a loose bound, which has not been optimized for clarity of explanation.

Two complications arise however when applying this result - the requirement for the adversary controlling the dynamic graph to be oblivious and the requirement that the dynamic graph remains connected. The first issue is dealt with with an observation that with no change in the \behaviorgroup{} vector (as long as no new agent switches to the \aware{} state), the sequence of graphs generated by agent movement is independent of the movement of \alerttokens{}. The second issue is resolved with the observation that even though each graph $G_t[A]$ induced by the set of \aware{} agents is not connected, the state invariant and the lack of residuals ensure that each of its connect components contains a witness. The witnesses are linked with imaginary edges to connect the graph, which we can do as we only care about the amount of time before the mark lands on any witness.

A new \aware{} agent is added when an active \unaware{} agent is activated. The probability of adding a new \aware{} agent on a given iteration with $k$ active agents is thus $\frac{k}{n}$, as each of the $n$ agents are activated with equal probability.
Thus, if we denote by the random sequence $K_1, K_2, K_3, \ldots$ the number of active agents on each iteration following $\Tfull$ (including $\Tfull$), we get the following expression for the expected value of $X$, which we use to denote the number of iterations following $\Tfull$ before a new \aware{} agent is added:
\begin{align*}
&\mathbb{E}[X | K_1, K_2, K_3, \ldots] 
= \sum_{x = 0}^\infty Pr\left(X > x | K_1, K_2, K_3, \ldots \right) \\
&= 1 + \sum_{x = 1}^\infty \prod_{i=1}^x \left(1 - \frac{K_i}{n}\right) 
%&= 1 + \sum_{x = 1}^\infty \left(1 - \frac{K_1}{n}\right)\left(1 - \frac{K_2}{n}\right)\cdots\left(1 - \frac{K_{x}}{n}\right)
\leq 1 + \sum_{x = 1}^\infty y^{\sum_{i=1}^{x}K_i} \text{ where $y := \left(1 - \frac{1}{n}\right) \in (0,1)$}\\
&= 1 + \underbrace{y^{K_1} + y^{K_1+K_2} + \ldots + y^{\sum_{i=1}^{D_1} K_i}}_{D_1\text{ terms}} + \underbrace{y^{C_1+K_{D_1+1}} + \ldots + y^{C_1 + \sum_{i=D_1+1}^{D_1+D_2} K_i}}_{D_2\text{ terms}} + \ldots \\
%&\leq 1 + (D_1-1) + y^{C_1} + y^{C_1}(D_2-1) + y^{C_1+C_2} + y^{C_1+C_2}(D_3-1) + y^{C_1+C_2+C_3} + \ldots \\
&\leq D_1 + D_2 \cdot y^{C_1} + D_3 \cdot y^{C_1+C_2} + \ldots \text{ (as $\sum_{i=1}^{D_1+D_2+\ldots+D_x} K_i = C_{x+1}$ for all $x$).}
\end{align*}
Thus, via the law of total expectation, we have
\begin{align*}
\mathbb{E}[X]
\leq \sum_{i=1}^\infty \mathbb{E}[D_i y^{\sum_{j=1}^{i-1}C_j}] 
&\leq \sum_{i=1}^\infty \mathbb{E}\left[ y^{\sum_{j=1}^{i-1}C_j} \mathbb{E}[D_i | C_1,C_2,\ldots,C_{i-1}] \right] \\
&\leq \sum_{i=1}^\infty U_D \mathbb{E}\left[ y^{\sum_{j=1}^{i-2}C_j} \mathbb{E}[y^{C_{i-1}} | C_1,C_2,\ldots,C_{i-2}] \right] \\
&\leq \ldots \leq \sum_{i=1}^\infty U_D (U_C)^{i-1} = \frac{U_D}{1-U_C}.
\end{align*}
This gives us an expected time bound of $O\left(n^5 \log n + \frac{U_D}{1-U_C}\right)$ to add a new \aware{} agent.
\end{proof}

\section{Foraging via self-induced phase changes}\label{sec:foraging}
Recall that in {\em the foraging problem}, we have ``ants'' (agents) that may initially be searching for ``food'' (stimuli, which can be any resource in the environment, like an energy source); once a food source is found, the ants that have learned about the food source start informing other ants, allowing them to switch their behaviors from the \emph{search mode} to the \emph{gather mode},
%and these transition to a {\em compression state}
that leads them to start to gather around the food source to consume the food.
%i.e., to incrementally enter a {\em gather phase}; however,
Once the source is depleted, the ants closer to the depleted source start a broadcast wave, gradually informing other ants that they should restart the search phase again by individually switching their states. 
%All of these "phase change waves" have to be {\em self-induced}, i.e., they are initiated by particles in the system and broadcasted locally  by the particles without any external control.
The foraging problem is very general and has several fundamental application domains, including search-and-rescue operations in swarms of nano- or micro-robots; health applications (e.g., a collective of nano-sensors that could search for, identify, and gather around a foreign body to isolate or consume it, then resume searching, etc.); and  finding and consuming/deactivating hazards in a nuclear reactor or a minefield.

%Our dynamic foraging algorithm is based on the stochastic compression algorithm of~\cite{Cannon2016}, which we summarize here. 
Our model for foraging is based on the {\it geometric Amoebot model} for programmable matter \cite{Daymude2021-canonicalamoebot, Daymude2023-canonicalamoebot}. We have $n$ anonymous agents occupying distinct sites on a $\sqrt{N}\times\sqrt{N}$ piece of the triangular lattice with periodic boundary conditions.
These agents have constant-size memory, and have no global orientation or any other global information beyond a common chirality.
Agents are activated with individual Poisson clocks, upon which they may move to adjacent unoccupied sites or change states, operating under similar constraints to the dynamic stimuli problem. An agent may only communicate with agents occupying adjacent sites of the lattice.
Food sources may be placed on any site of the lattice, removed, or shifted around at arbitrary times, possibly adversarially, and an agent can only observe the presence of the food source while occupying the lattice site containing it.
%
%It will be convenient to refer a vertex in $\Lambda$ according to hypothetical global coordinates $(x,y)$, but note that this is just for ease of exposition, since the agents are not aware of any such global coordinate system. In this convention, a vertex $(x,y) \in \Lambda$ has edges to the vertices corresponding to $(x-1,y-1)$, $(x-1,y)$, $(x-1,y)$, $(x+1,y)$, $(x,y+1)$ ,$(x+1,y+1)$, with the arithmetic taken modulo $\sqrt{N}$.
%Agents are aware of their own and their neighbors' current states and when an agent is activated, it may do a bounded amount of computation, send at most one token (not necessarily identical) to each of its neighbors, and choose one of its six neighbors in the lattice to see if it is unoccupied and move there.
This model can be viewed as a high level abstraction of the (canonical) {\em Amoebot model}~\cite{Daymude2021-canonicalamoebot, Daymude2023-canonicalamoebot} under a {\em random sequential scheduler}, where at most one agent would be active at any point in time. One should be able to port the model and algorithms presented in this paper to the Amoebot model; however a formal description on how this should be done is beyond the scope of this paper.

At any point of time, there are two main states an agent can be in,
%We will define two main states an agent can be in at any point in time,
which, at the macro-level, are to induce the collective to enter the {\em search} or {\em gather} modes respectively.
When in {\it search mode}, agents move around in a process akin to a simple exclusion process, where they perform a random walk while avoiding two agents occupying the same site.
Agents enter the {\it gather mode} when food is found and this information is propagated in the system, consequently resulting in the system compressing around the food (Figure~\ref{fig:foragingillustration}).

\begin{figure}[!h]
\begin{center}
\begin{subfigure}[b]{0.5\linewidth}
  \begin{center}
  \begin{tikzpicture}[x=0.55cm,y=0.55cm]
  \draw[lightgray] (10.3923,-1) -- (10.3923,-4);
\draw[lightgray] (2.59808,-4.5) -- (9.52628,-0.5);
\draw[lightgray] (4.33013,-4.5) -- (10.3923,-1);
\draw[lightgray] (0,-2) -- (2.59808,-0.5);
\draw[lightgray] (0,-2) -- (4.33013,-4.5);
\draw[lightgray] (0.866025,-0.5) -- (7.79423,-4.5);
\draw[lightgray] (0.866025,-0.5) -- (0.866025,-4.5);
\draw[lightgray] (6.06218,-4.5) -- (10.3923,-2);
\draw[lightgray] (2.59808,-0.5) -- (9.52628,-4.5);
\draw[lightgray] (2.59808,-0.5) -- (2.59808,-4.5);
\draw[lightgray] (4.33013,-0.5) -- (10.3923,-4);
\draw[lightgray] (4.33013,-0.5) -- (4.33013,-4.5);
\draw[lightgray] (9.52628,-4.5) -- (10.3923,-4);
\draw[lightgray] (7.79423,-4.5) -- (10.3923,-3);
\draw[lightgray] (0,-1) -- (0.866025,-0.5);
\draw[lightgray] (0,-1) -- (6.06218,-4.5);
\draw[lightgray] (0,-1) -- (0,-4);
\draw[lightgray] (0,-4) -- (6.06218,-0.5);
\draw[lightgray] (0,-4) -- (0.866025,-4.5);
\draw[lightgray] (1.73205,-1) -- (1.73205,-4);
\draw[lightgray] (5.19615,-1) -- (5.19615,-4);
\draw[lightgray] (6.06218,-0.5) -- (10.3923,-3);
\draw[lightgray] (6.06218,-0.5) -- (6.06218,-4.5);
\draw[lightgray] (3.4641,-1) -- (3.4641,-4);
\draw[lightgray] (7.79423,-0.5) -- (10.3923,-2);
\draw[lightgray] (7.79423,-0.5) -- (7.79423,-4.5);
\draw[lightgray] (6.9282,-1) -- (6.9282,-4);
\draw[lightgray] (9.52628,-0.5) -- (10.3923,-1);
\draw[lightgray] (9.52628,-0.5) -- (9.52628,-4.5);
\draw[lightgray] (0,-3) -- (4.33013,-0.5);
\draw[lightgray] (0,-3) -- (2.59808,-4.5);
\draw[lightgray] (8.66025,-1) -- (8.66025,-4);
\draw[lightgray] (0.866025,-4.5) -- (7.79423,-0.5);
\draw[black, line width=0.4mm, fill=white] (0.866025,-1.5) circle (0.288);
\draw[black, line width=0.4mm, fill=white] (0.866025,-3.5) circle (0.288);
\draw[line width=0.4mm] (2.85263,-2.24544) -- (2.34352,-2.75456);
\draw[line width=0.4mm] (2.34352,-2.24544) -- (2.85263,-2.75456);
\node[align=left] at (2.59808,-1.9) {\footnotesize food};
\draw[black, line width=0.4mm, fill=white] (3.4641,-1) circle (0.288);
\draw[black, line width=0.4mm, fill=white] (5.19615,-1) circle (0.288);
\draw[black, line width=0.4mm, fill=white] (5.19615,-3) circle (0.288);
\draw[black, line width=0.4mm, fill=white] (7.79423,-2.5) circle (0.288);
\draw[black, line width=0.4mm, fill=white] (8.66025,-4) circle (0.288);
\draw[black, line width=0.4mm, fill=white] (9.52628,-1.5) circle (0.288);
  \end{tikzpicture}
  \end{center}
\end{subfigure}%
\begin{subfigure}[b]{0.5\linewidth}
  \begin{center}
  \begin{tikzpicture}[x=0.55cm,y=0.55cm]
  \draw[lightgray] (10.3923,-1) -- (10.3923,-4);
\draw[lightgray] (2.59808,-4.5) -- (9.52628,-0.5);
\draw[lightgray] (4.33013,-4.5) -- (10.3923,-1);
\draw[lightgray] (0,-2) -- (2.59808,-0.5);
\draw[lightgray] (0,-2) -- (4.33013,-4.5);
\draw[lightgray] (0.866025,-0.5) -- (7.79423,-4.5);
\draw[lightgray] (0.866025,-0.5) -- (0.866025,-4.5);
\draw[lightgray] (6.06218,-4.5) -- (10.3923,-2);
\draw[lightgray] (2.59808,-0.5) -- (9.52628,-4.5);
\draw[lightgray] (2.59808,-0.5) -- (2.59808,-4.5);
\draw[lightgray] (4.33013,-0.5) -- (10.3923,-4);
\draw[lightgray] (4.33013,-0.5) -- (4.33013,-4.5);
\draw[lightgray] (9.52628,-4.5) -- (10.3923,-4);
\draw[lightgray] (7.79423,-4.5) -- (10.3923,-3);
\draw[lightgray] (0,-1) -- (0.866025,-0.5);
\draw[lightgray] (0,-1) -- (6.06218,-4.5);
\draw[lightgray] (0,-1) -- (0,-4);
\draw[lightgray] (0,-4) -- (6.06218,-0.5);
\draw[lightgray] (0,-4) -- (0.866025,-4.5);
\draw[lightgray] (1.73205,-1) -- (1.73205,-4);
\draw[lightgray] (5.19615,-1) -- (5.19615,-4);
\draw[lightgray] (6.06218,-0.5) -- (10.3923,-3);
\draw[lightgray] (6.06218,-0.5) -- (6.06218,-4.5);
\draw[lightgray] (3.4641,-1) -- (3.4641,-4);
\draw[lightgray] (7.79423,-0.5) -- (10.3923,-2);
\draw[lightgray] (7.79423,-0.5) -- (7.79423,-4.5);
\draw[lightgray] (6.9282,-1) -- (6.9282,-4);
\draw[lightgray] (9.52628,-0.5) -- (10.3923,-1);
\draw[lightgray] (9.52628,-0.5) -- (9.52628,-4.5);
\draw[lightgray] (0,-3) -- (4.33013,-0.5);
\draw[lightgray] (0,-3) -- (2.59808,-4.5);
\draw[lightgray] (8.66025,-1) -- (8.66025,-4);
\draw[lightgray] (0.866025,-4.5) -- (7.79423,-0.5);
\draw[black, line width=0.4mm, fill=white] (1.73205,-2) circle (0.288);
\draw[black, line width=0.4mm, fill=white] (1.73205,-3) circle (0.288);
\draw[black, line width=0.4mm, fill=white] (2.59808,-2.5) circle (0.288);
\draw[line width=0.4mm] (2.85263,-2.24544) -- (2.34352,-2.75456);
\draw[line width=0.4mm] (2.34352,-2.24544) -- (2.85263,-2.75456);
\node[align=left] at (2.59808,-1.9) {\footnotesize food};
\draw[black, line width=0.4mm, fill=white] (2.59808,-3.5) circle (0.288);
\draw[black, line width=0.4mm, fill=white] (3.4641,-2) circle (0.288);
\draw[black, line width=0.4mm, fill=white] (3.4641,-3) circle (0.288);
\draw[black, line width=0.4mm, fill=white] (4.33013,-2.5) circle (0.288);
\draw[black, line width=0.4mm, fill=white] (4.33013,-3.5) circle (0.288);
  \end{tikzpicture}
  \end{center}
\end{subfigure}
\end{center}
\vspace{-.1in}
\caption{In the diagram on the left, a food source is placed on a lattice site. The diagram on the right illustrates a desired configuration, where all agents have gathered in a low perimeter configuration around the food source. If the food source is later removed, the agents should once again disperse, returning to a configuration like the figure on the left.}
\label{fig:foragingillustration}
\end{figure}
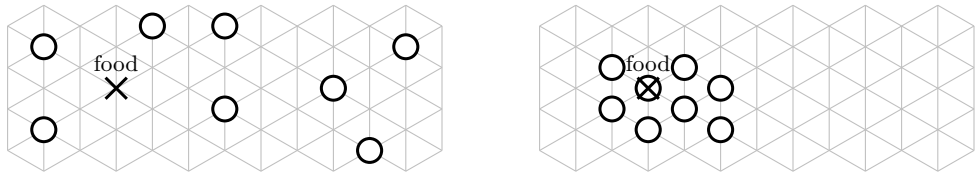

In the nonadaptive setting, Cannon {\it et al}$.$~\cite{Cannon2016} designed a rigorous   { compression/expansion algorithm} for agents that remain simply connected throughout execution, where a single parameter $\lambda$ determines a system-wide phase:  A small $\lambda$, namely $\lambda< 2.17$, provably corresponds to the search mode, which is desirable to search for food, while large $\lambda$, namely $\lambda>2+\sqrt{2}$, corresponds to the gather mode, desirable when food has been discovered. 
Likewise, Li {\it et al.} \cite{Li2021-bobbots} show a very similar bifurcation based on a bias parameter $\lambda$  in the setting when the agents are allowed to disconnect and disperse throughout the lattice.
Our goal here is to perform a system-wide adjustment in the bias parameters when one or more agents notice the presence or depletion of food to induce the appropriate global coordination to
provably transition the collective between macro-modes when required.  Informally, one can imagine individual agents adjusting their $\lambda$ parameter to be high when they are fed, encouraging compression around the food, and making $\lambda$ small when they are hungry, promoting the search for more food. 
A configuration is called {\it $\alpha$-compressed} if  the perimeter (measured by the length of the closed walk around its boundary edges) is at most $\alpha\, p_{\text{min}}(n)$, for some constant $\alpha > 1,$  where $p_{\text{min}}(n)$ denotes the minimum possible perimeter of a connected system with $n$ agents, which is the desired outcome of the gather mode.

\subparagraph*{Adaptive $\alpha$-compression.}
We present the first rigorous local distributed algorithm for the foraging problem:
The {\it Adaptive $\alpha$-Compression} algorithm  
is based on the stochastic compression algorithm of~\cite{Cannon2016}, addressing a 
%by stating it as a 
geometric application of the dynamic stimuli problem, where
%We introduce the algorithm for foraging in the context of the dynamic stimuli problem. 
the \aware{} state represents the ``gather'' mode, and the \unaware{} state represents the ``search'' mode.
A witness 
%in this case 
is an agent that occupies the same
%currently observes the food source (which is when the agent is on a 
lattice site as the food source.
The underlying dynamic graph used by the Adaptive Stimuli Algorithm is given by the adjacency graph of the agents - two agents share an edge on the graph if they occupy adjacent sites of the lattice.
%(this is exactly when these agents can communicate in the foraging problem).
As agents move around to implement behaviors like gathering and searching, their neighbor sets will change. The movement of agents thus reconfigures and oftentimes even disconnects our graph.
%These are factors we account for in the reconfigurable version of the problem.

%as two agents on our lattice can only communicate if they occupy adjacent sites of the lattice, thus the edges of the graph represent pairs of adjacent agents.

%Consider the simple case where our agents do not move and are in a connected configuration on the lattice. The static version of the dynamic stimuli problem directly applies here, giving us a polynomial time solution to switch all agents to the gather or search modes in the presence or absence of food (assuming that the food does not appear on lattice sites not occupied by agents). This is not a very practical example however, as agents will need to move to implement behaviors like gathering and searching. As the neighbors of agents change as they move around, the movement of agents reconfigures and even occasionally disconnects our graph. These are factors we account for in the reconfigurable version of the problem.
%so we would like to extend our dynamic stimuli problem to allow for flexible reconfiguration adversary by the agents on the graph.

\begin{algorithm}[t]
\caption{Adaptive $\alpha$-Compression}% $\lambda > \sqrt{2 + \sqrt{2}}$ is a parameter.
\label{alg:adaptivealphacompression}
\begin{algorithmic}[1]
%\State {\bf Input:} A $\sqrt{N}\times\sqrt{N}$ triangular lattice $\Lambda_N$ and parameters $p \in (0,1)$, $\lambda > 1$.
\Procedure{Adaptive-Alpha-Compression}{$u$}
\State $q \gets$ Random number in $[0,1]$
%\State $u.isWitness \gets \Call{Observes-Food-Source}{u}$
\State $u.isWitness \gets $\textsc{True} if $u$ observes the food source, else $u.isWitness \gets$\textsc{False}
\If{$q \leq \frac{1}{2}$} \Comment{With probability $\frac{1}{2}$, make a state update}
    \State \Call{Adaptive-Stimuli-Algorithm}{u}
\Else \Comment{With probability $\frac{1}{2}$, make a move}
    \If{$u.isWitness$ or $u.state \in \{\stateAW, \stateAAW, \stateAC\}$} \Comment{\immobileb{} agent}
        \State Do nothing
    \ElsIf{$u.state \in \{\stateA, \stateAA\}$}\Comment{\mobileb{} agent}
        \State \Call{Execute-Gather}{u}
    \ElsIf{$u.state = \stateU$} \Comment{\unawareb{} agent}
        \State \Call{Execute-Search}{u}
    \EndIf
\EndIf
\EndProcedure
\end{algorithmic}

\begin{algorithmic}[1]
\Procedure{Execute-Gather}{$u$}
\State $d \gets$ Random direction in $\{0,1,2,3,4,5\}$
\State $\ell \gets$ Current position of $u$
\State $\ell' \gets$ Neighboring lattice site of $u$ in direction $d$
\If{Moving $u$ from $\ell$ to $\ell'$ is a valid compression move (Definition~\ref{defn:validcompressionmoves})}
    \State $p \gets$ Random number in $[0,1]$
    \State $d(u) \gets$ number of neighboring agents of $u$ if $u$ were at position $\ell$
    \State $d'(u) \gets$ number of neighboring agents of $u$ if $u$ were at position $\ell'$
    \If{$p \leq \lambda^{d'(u) - d(u)}$} %\Comment{$d(u)$ and $d'(u)$ are the number of neighbors of $u$ before and after the move}
        \State Move $u$ to position $\ell'$ \Comment{Movements reconfigure the adjacency graph}
    \EndIf
\EndIf
\EndProcedure
\end{algorithmic}

\begin{algorithmic}[1]
\Procedure{Execute-Search}{$u$}
\State $d \gets$ Random direction in $\{0,1,2,3,4,5\}$
\State $\ell' \gets$ Neighboring lattice site of $u$ in direction $d$
\If{Moving $\ell'$ is an unoccupied lattice site}
    \State Move $u$ to position $\ell'$ \Comment{Movements reconfigure the adjacency graph}
\EndIf
\EndProcedure
\end{algorithmic}
\end{algorithm}

In this algorithm, agents in the \unawareb{} (search) \behaviorgroup{} execute movements akin to a simple exclusion process (\textsc{Execute-Search}) while agents in the \mobileb{} (gather) \behaviorgroup{} execute moves of the compression algorithm (\textsc{Execute-Gather}) in~\cite{Cannon2016}. We focus on the compression algorithm run by the \aware{} agents.
In a simple exclusion process, a selected agent 
picks a direction at random, and moves in that direction if and only if the immediate neighboring site in that direction is unoccupied.
In the compression algorithm~\cite{Cannon2016} on the other hand, a selected \aware{} agent first picks a direction at random to move in, and if this move is a valid compression move (according to Definition~\ref{defn:validcompressionmoves}), the agent  moves to the chosen position with the probability given in Definition~\ref{defn:compressionmoveprobability}.
With a
(far-from-trivial)
modification of the analysis in~\cite{Cannon2016} to account for the stationary witness agent, we show that in the case of a single food source, the ``gather'' movements
allow the agents to form a low perimeter cluster around the food.
We present the following definitions, adapted from~\cite{Cannon2016} to the set of \aware{} agents running the compression algorithm:
\begin{definition}[Valid Compression Moves~\cite{Cannon2016}]
\label{defn:validcompressionmoves}
Denote by $\NA(\ell)$ and $\NA(\ell')$ the sets of \aware{} neighbors of $\ell$ and $\ell'$ respectively and $\NA(\ell \cup \ell') := \NA(\ell) \cup \NA(\ell') \setminus \{\ell,\ell'\}$.
Consider the following two properties:

\smallskip
\noindent
\underbar{Property 1:} $|\NA(\ell)\cap \NA(\ell')| \geq 1$ and every agent in $\NA(\ell \cup \ell')$ is connected to an agent in $\NA(\ell)\cap \NA(\ell')$ through $\NA(\ell \cup \ell')$.

\smallskip
\noindent
\underbar{Property 2:} $|\NA(\ell)\cap \NA(\ell')| = 0$, $\ell$ and $\ell'$ each have at least one neighbor, all agents in $\NA(\ell) \setminus \{\ell'\}$ are connected by paths within the set, and all agents in $\NA(\ell') \setminus \{\ell\}$ are connected by paths within the set.

\smallskip
%\noindent 
We say the move from $\ell$ to $\ell'$ is a \emph{valid compression move} if it satisfies both properties, and $\NA(\ell)$ contains fewer than five aware state agents.
\end{definition}

\begin{definition}[Transition probabilities~\cite{Cannon2016}]
\label{defn:compressionmoveprobability}
Fix $\lambda > 2 + \sqrt{2}$, as sufficient for $\alpha$-compression.
An agent $u$ transitions through a valid movement
%Even when a movement is valid, we only make the move 
with Metropolis-Hastings~\cite{metropolis} acceptance probability $\min\{1,\lambda^{e(\sigma')-e(\sigma)}\}$,
where $\sigma$ and $\sigma'$ are the configurations before and after the movement, and $e(\cdot)$ represents the number of edges between \aware{} state agents in the configuration.
\end{definition}
Note that even though $e(\cdot)$ is a global property, 
%when an agent $u$ is being moved, 
the difference $e(\sigma')-e(\sigma)$ can be computed locally (within two hops in the lattice, or through expansions in the Amoebot model~\cite{Daymude2021-canonicalamoebot,Daymude2023-canonicalamoebot}), as it is just the change in the number of \aware{} neighbors of $u$ before and after its movement.

The condition for valid compression moves is notable as it keeps 
%a  configuration simply connected (connected and hole-free), 
a component of aware agents containing the witness agent % food
%configuration 
simply connected (connected and hole-free), which is crucial to the proof in~\cite{Cannon2016} that a low perimeter configuration, in our case around the food source, will be obtained in the long term. We also show that these valid compression moves are locally connected (as per Definition~\ref{defn:locallyconnected}), a sufficient condition for the reconfigurable dynamic stimuli problem to apply.

%%%%%%%%%%%%%%%%j  ***** moved ergodicity here *****

%However, even though the Metropolis-Hastings algorithm gives us the same stationary distribution and hence the same low-perimeter clusters, 
We first prove that the Markov chain representing the compression moves (\textsc{Execute-Gather}) is connected.
%, which is challenging when there is a fixed food source that is part of the cluster.
The proof builds upon the ergodicity argument in \cite{Cannon2016};
however, the addition of a single stationary agent, the witness, in our context makes this proof significantly more complex, and we defer it to the full version of the paper.
%We show the following lemma for irreducibility:

\begin{lemma}
\label{lemma:irreducible}
%We consider configurations of agents executing the compression algorithm with a single immobile agent.
%Suppose that there is enough space on the lattice for all agents to be laid out on a single line on any position and in any direction, without the line intersecting itself.
%If the agents  can only make the moves given in Definition~\ref{defn:validcompressionmoves} while the food source remains fixed in place,
Consider connected configurations of agents on a triangular lattice with a single agent $v$ that cannot move.
%As long as there is sufficient space on the lattice for all agents to be laid out in a single line in any direction without the line intersecting itself, t
There exists a sequence of valid compression moves 
%(Definition~\ref{defn:validcompressionmoves}) 
that transforms any 
%simply 
connected configuration of agents into any simply connected
%] hole-free
%\footnote{We call a configuration hole-free if there exists no closed path over the agents that encircles an empty site of the lattice.} 
configuration of the agents while keeping $v$ stationary. 
\end{lemma}
%
%Unfortunately, the proof that the state space is connected with a single immobile agent is significantly more complex than without one.

We may now state our main results, which verify the correctness of Adaptive $\alpha$-Compression. 
%\begin{theorem}
%\label{thm:compressionhasfood}
%If a food source exists and remains in place long enough, then all agents will reach and remain in the \aware{} state and will tightly aggregate around the food, forming a component with perimeter at most $\alpha$ times the minimum possible perimeter, for any $\alpha > 1$, with high probability.
%\end{theorem}
\begin{theorem}
\label{thm:compressionnofood}
If no food source has been identified for sufficiently long, 
then within an expected $O(n^2)$
%polynomial number of
steps, all agents will reach and remain in the \unaware{} state and will converge to the uniform distribution of nonoverlapping 
%agents 
 lattice positions.
%region.
%Furthermore, as agents in the \unaware{} state follow a simple exclusion process, the distribution of agents over the lattice subsequently converges quickly to the uniform distribution.
\end{theorem}
\begin{theorem}
\label{thm:compressionhasfoodprecise}
If at least one food source exists and remains in place for long enough, then within $O(n^6 \log n + N^2 n)$
%a polynomial number of 
steps in expectation, all agents will reach and remain in the \aware{} state, and each component of \aware{} agents will contain a food source.
In addition, if there is only one food source, the agents will converge to a configuration 
with a single $\alpha$-compressed component around the food, for any constant $\alpha > 1$, with all but an exponentially small probability, for a large enough lattice region.
%high probability.
%with perimeter at most $\alpha$ times the minimum possible perimeter, for any $\alpha > 1$, with high probability.
\end{theorem}

%%%%% Proof of correctness
The Adaptive $\alpha$-Compression Algorithm fits the requirements of the reconfigurable dynamic stimuli model. In particular, the information $X_t$ available to the reconfiguration adversary corresponds to the configuration of the lattice, and the graph $G_t$ represents the adjacency of agents on the lattice at that time $t$.
To show that a sequence of valid \aware{} agent movements in the Adaptive $\alpha$-Compression Algorithm, which determine the configurations of $G_1,G_2,\ldots$, can be modeled via a valid reconfiguration adversary $\mathcal X$, we need to show that the reconfigurations resulting from the \textsc{Execute-Gather} procedure must be locally connected.
%(recall Definition~\ref{defn:locallyconnected}).
%(Lemma~\ref{lemma:compressionlocallyconnected})
%, proof in Appendix~\ref{appendix:compressiondetails}).

\begin{lemma}
\label{lemma:compressionlocallyconnected}
The movement behavior of Adaptive $\alpha$-Compression is locally connected. %(Definition~\ref{defn:locallyconnected}).
\end{lemma}

\begin{proof}
We only need to show that the \textsc{Execute-Gather} procedure maintains local connectivity (Definition~\ref{defn:locallyconnected}).
This is true as when reconfiguring an \aware{} agent $u$ with \aware{} neighbor set $N_\stateAbase(u)$, we only allow valid compression moves (Definition~\ref{defn:validcompressionmoves}) to be made.
In the case of Property 1, all agents in $N_\stateAbase(u)$ will still have paths to $u$ in $G'$ through $S$.
In the case of Property 2, all agents in $N_\stateAbase(u)$ will still have paths to each other within $G'[N_\stateAbase(u)]$, despite no longer having local paths to $u$. The agent $u$ will have at least one \aware{} state neighbor after the move as this is a requirement of Property 2.
\end{proof}

As this is an instance of the reconfigurable dynamic stimuli problem, Theorem~\ref{thm:compressionnofood} follows immediately from Theorem~\ref{theorem:mainresultnostimulidynamic}. To show Theorem~\ref{thm:compressionhasfoodprecise} however, we need to show polynomial recurring rates %(Lemma~\ref{lemma:adaptivecompressionrecurring}) 
by arguing that the \unaware{} state agents following a simple exclusion process will regularly come into contact with the clusters of \aware{} state agents around the food sources.
%The proof of Lemma~\ref{lemma:adaptivecompressionrecurring}
%and the first part of Theorem~\ref{thm:compressionhasfoodprecise} (the case with multiple food sources)
%will be given in Appendix~\ref{appendix:compressiondetails}.

\begin{lemma}
\label{lemma:adaptivecompressionrecurring}
The movement behavior defined in the Adaptive $\alpha$-Compression Algorithm is $(U_D,U_C)$-recurring with $U_D = 2N^2 + \frac{2}{n} + 1$ and $U_C = \frac{2}{3}$.
\end{lemma}
\begin{proof} [Proof Sketch]
In the interest of space, we give only a brief summary of the proof (which is available in the full version of the paper). We define the random sequences $(D_1, D_2, D_3, \ldots)$ and $(C_1, C_2, C_3, \ldots)$ by dividing the time steps after the starting iteration $t$ 
into batches, where the $k$\textsuperscript{th} batch would take $D_k$ iterations and would see $C_k$ active agents over its duration.

A batch ends (and the next batch starts) when the agent movement places an \unaware{} agent $u$ next to an \aware{} agent $v$, then attempts a movement of $u$ or $v$ after that. The duration $D_k$ of the $k$\textsuperscript{th} batch can be computed with the hitting time of a simple exclusion process over the triangular lattice, plus a geometric random variable representing the number of iterations taken to select $u$ or $v$ after that. This gives us a uniform upper bound of $2N^2n + \frac{2}{n} + 1$ for $\mathbb{E}[D_k | D_1, D_2, \ldots D_{k-1}, C_1,C_2 \ldots C_{k-1}]$ for each batch $k \in \{1,2,\ldots\}$.

The number of active agents $C_k$ within batch $k$ is at least the number of iterations between the first time within the batch that 
%the agent movement places 
an \unaware{} agent moves next to an \aware{} agent and the end of the batch.
This is shown to stochastically dominate a geometric random variable $Y$ with success probability $p_Y$, which we show in the full paper to be:
\begin{align*}
p_Y &= \sum_{i=0}^\infty \frac{1}{2} \cdot \left(\frac{1}{2}\right)^i \left(1 - \left(1-\frac{2}{n}\right)^i \right)
= 1 - \frac{1}{2}\sum_{i=0}^\infty \left( \frac{1 - 2/n}{2} \right)^i
= \frac{2}{n}\left( \frac{1}{1 + 2/n} \right)
\end{align*}

As $C_k$ stochastically dominates $Y$ and $(1-\frac{1}{n})^x$ is a decreasing function of $x$, we have:
\begin{align*}
\mathbb{E}\left[\left(1-\frac{1}{n}\right)^{C_k} | D_1, D_2, \ldots D_{k-1}, C_1,C_2 \ldots C_{k-1}\right]
&\leq \mathbb{E}\left[\left(1-\frac{p}{n}\right)^Y\right] \\
&= \sum_{y=0}^\infty \left(1 - \frac{1}{n}\right)^y p_Y \left(1 - p_Y\right)^y \\
&= p_Y \frac{1}{1 - (1-1/n)(1-p_Y)} 
= \frac{2}{3}. \qedhere
\end{align*}
%\vspace{-.4in}
\end{proof}
To show the first half of Theorem~\ref{thm:compressionhasfoodprecise}, we start from the first iteration beyond which no additional changes in the positions (or existence) of the food sources occur. 
We first show that if there is at least one food source, it will be found.
%Assuming there is at least one food source, the agents search for a food source.
%
%If there are no agents on food sources at this point, then there are currently no witnesses, and 
As long as no food source has been found, there will be no witnesses, so every agent will return to the \unaware{} state by Theorem~\ref{theorem:mainresultnostimulidynamic}. Agents in the \unaware{} state move randomly following a simple exclusion process. Using the hitting time of a simple random walk on a regular graph (the triangular lattice) of $N$ sites, we have a simple upper bound of $O(N^2n)$ iterations in expectation before some agent finds a food source and becomes a witness.

From then on, there will be at least one witness, and the witness set can only be augmented, not reduced, as the other agents potentially find additional food sources, and since the agents already sitting on food sources are no longer allowed to move.
%Agents sitting on food sources are no longer allowed to move in the Adaptive $\alpha$-Compression Algorithm (since we assumed that the locations of food sources no longer changes).
As agent movement behaviors are recurring with polynomial bounds (Lemma~\ref{lemma:adaptivecompressionrecurring}), the reconfigurable Adaptive Stimuli Algorithm %(Theorem~\ref{theorem:mainresultwithstimulidynamic}) thus applies,
applies, yielding a polynomial bound on the expected number of iterations before all agents have switched to the \aware{} state with no residuals.
Additionally, due to the maintenance of the state invariant and as there are no residuals, every component of \aware{} agents will contain at least one witness, meaning that every cluster of agents will be around some food source. This gives us the first part of Theorem~\ref{thm:compressionhasfoodprecise}.

\begin{figure}[!ht]
\begin{subfigure}[b]{\linewidth}
  \begin{center}
  \includegraphics[width=0.95\linewidth]{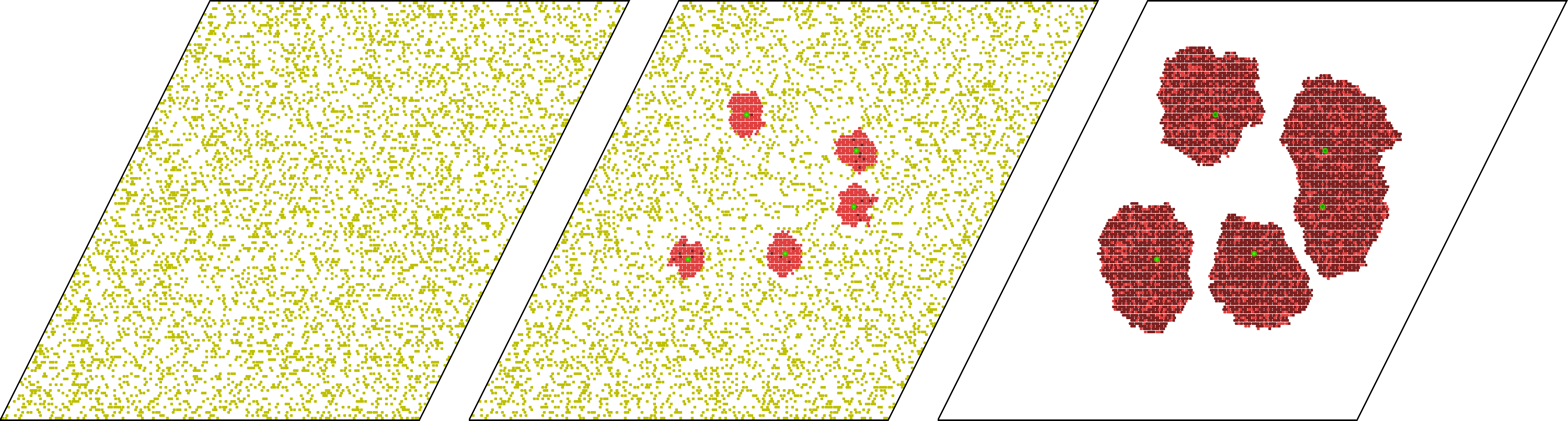}
  \end{center}
  \caption{All agents initially unaware. Five food sources added, all agents eventually switch to the \aware{} state.}
\end{subfigure}\\
~\\
\begin{subfigure}[b]{\linewidth}
  \begin{center}
  \includegraphics[width=0.95\linewidth]{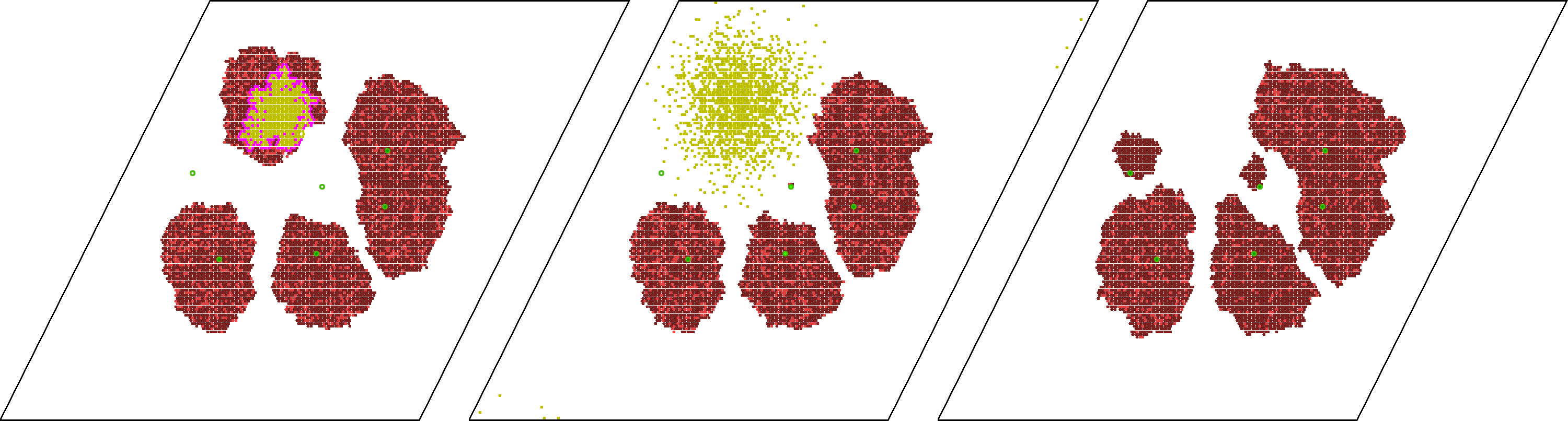}
  \end{center}
  \caption{One food source removed, two new added. A cluster disperses and agents rejoin other clusters.}
\end{subfigure}\\
~\\
\begin{subfigure}[b]{\linewidth}
  \begin{center}
  \includegraphics[width=0.95\linewidth]{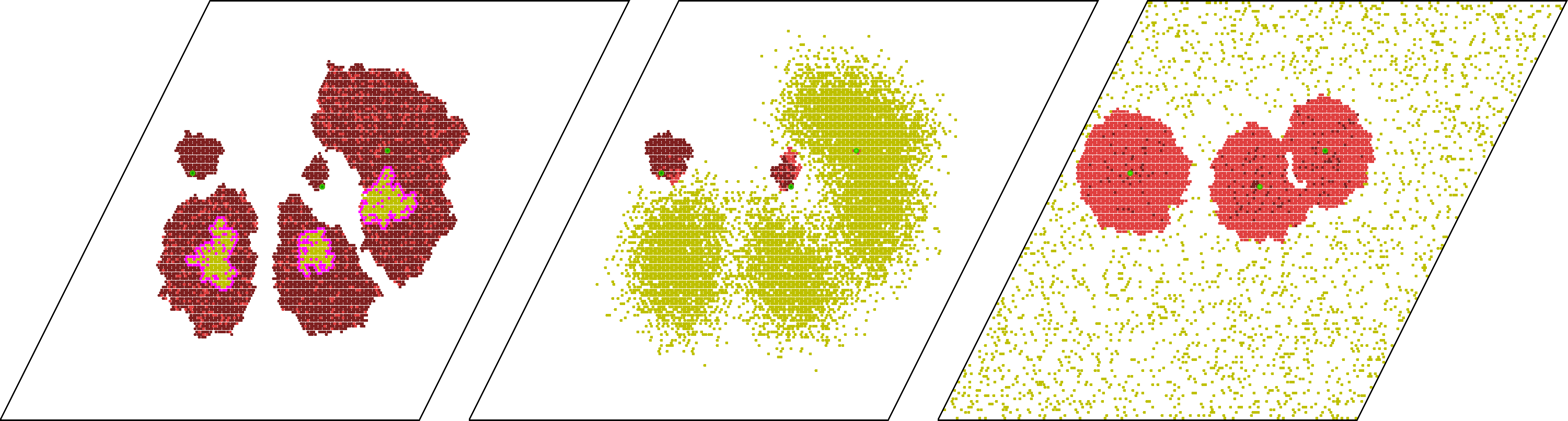}
  \end{center}
  \caption{Three food sources removed, clusters disperse and gather around the remaining food sources.}
\end{subfigure}\\
~\\
\begin{subfigure}[b]{\linewidth}
  \begin{center}
  \includegraphics[width=0.95\linewidth]{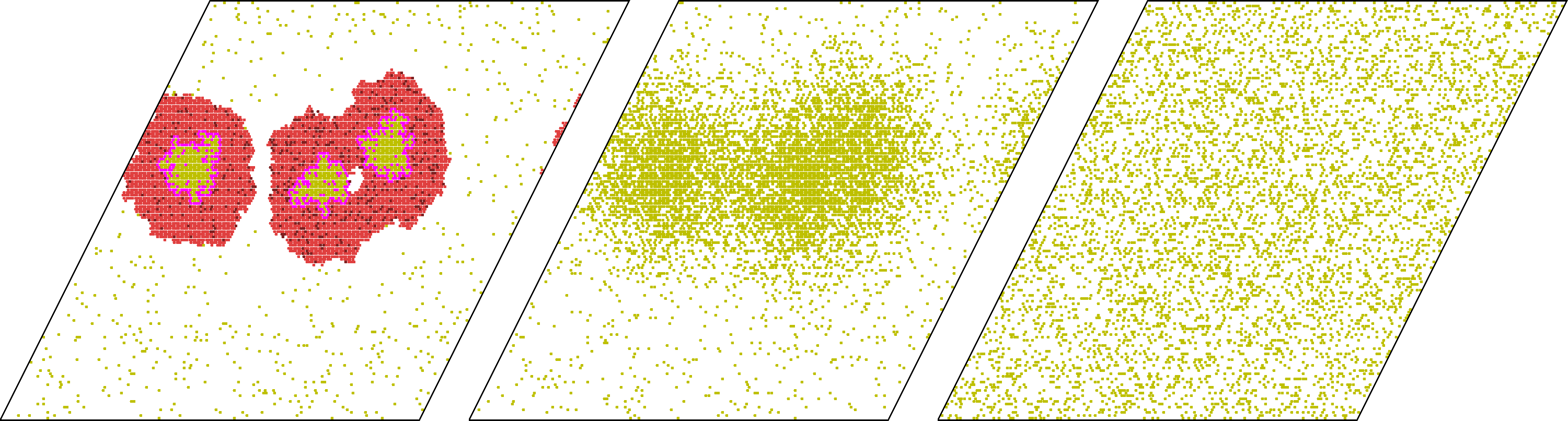}
  \end{center}
  \caption{All food sources removed, agents disperse and converge to a uniform distribution over the lattice.}
\end{subfigure}
\caption{Simulation of Adaptive $\alpha$-Compression with multiple food sources. The images are in chronological order. \unaware{} agents are yellow, \aware{} agents are red (darker red if they have an \alerttoken{}), agents with the \cleartoken{} are purple, and food sources are green.}
\label{fig:simulation}
\end{figure}

%\subparagraph{Adaptive $\alpha$-Compression in the case of a single food source.}
%Agents in the compression state execute movements from the compression algorithm of~\cite{Cannon2016}.

The second half of Theorem~\ref{thm:compressionhasfoodprecise} states that a low perimeter ($\alpha$-compressed) configuration is achievable in the case of a single food source.
As the Markov chain representing the compression moves is irreducible (Lemma~\ref{lemma:irreducible}),
%Assuming irreducibility of the Markov chain $\mathcal{M}_{\rm COM}$ representing the compression moves (\textsc{Execute-Gather} in Algorithm~\ref{alg:adaptivealphacompression}),
the results of \cite{Cannon2016} guarantee that for any $\alpha > 1$, there exists a sufficiently large constant $\lambda$ such that at stationarity, the perimeter of the cluster is at most $\alpha$ times its minimum possible perimeter with high probability.
%

%\comm{SH: ``The page limit for SAND proceedings is 20 pages, excluding title page and bibliography.'' - Does this mean we can go up to page 21?}

%\input{appendix_simulation}
\section{Simulations of the Adaptive $\alpha$-Compression Algorithm}
\label{appendix:simulation}
We demonstrate a simulation of the Adaptive $\alpha$-Compression algorithm with $5625$ agents in a $150\times 150$ triangular lattice with periodic boundary conditions. Multiple food sources (stimuli) are placed and moved around to illustrate the gather and search phases. This simulation is shown as a sequence of $12$ images in chronological order in Figure~\ref{fig:simulation}.

%\begin{acknowledgments}
\section*{Acknowledgements}
The authors thank the anonymous reviewers for useful feedback on the presentation.
%\end{acknowledgments}

%\section{Conclusions}

%Remarks:

%1. Max degree requirement can be relaxed but we use it for simplicity

%2. Open question: More that a constant number of witnesses

%3. Open question: Algorithms that function on weaker connectivity requirements

%4. Open question: witnesses that exist in the long term but with positions that constantly switch back and forth

\newpage
%%%%%%%%%%%%%%%%%%%%%%%%%%%%%%%%%%%%%
%%%%%%%%%%%%%%%%%%%%%%%%%%%%%%%%%%%%%
%%%%%%%%%%%%%%%%%%%%%%%%%%%%%%%%%%%%%
%%%%%%%%%%%%%%%%%%%%%%%%%%%%%%%%%%%%%
%%%%%%%%%%%%%%%%%%%%%%%%%%%%%%%%%%%%%
%\bibliographystyle{ACM-Reference-Format}
\bibliographystyle{plainurl}
\bibliography{bibfile}

\begin{thebibliography}{10}

\bibitem{alberti}
Simon Alberti.
\newblock Organizing living matter: The role of phase transitions in cell
  biology and disease.
\newblock {\em Biophysical journal}, 14, 2018.

\bibitem{Arroyo2018}
Marta~Andrés Arroyo, Sarah Cannon, Joshua~J. Daymude, Dana Randall, and
  Andr\'{e}a~W. Richa.
\newblock A stochastic approach to shortcut bridging in programmable matter.
\newblock In {\em 23rd International Con- ference on DNA Computing and
  Molecular Programming (DNA)}, pages 122--138, 2017.

\bibitem{RandomWalksDynamicGraphs}
Chen Avin, Michal Koucký, and Zvi Lotker.
\newblock Cover time and mixing time of random walks on dynamic graphs.
\newblock {\em Random Structures \& Algorithms}, 52(4):576--596, 2018.

\bibitem{Cannon2019}
Sarah Cannon, Joshua~J. Daymude, Cem G{\"o}kmen, Dana Randall, and
  Andr{\'e}a~W. Richa.
\newblock A local stochastic algorithm for separation in heterogeneous
  self-organizing particle systems.
\newblock In {\em Approximation, Randomization, and Combinatorial Optimization.
  Algorithms and Techniques (APPROX/RANDOM 2019)}, pages 54:1--54:22, 2019.

\bibitem{Cannon2016}
Sarah Cannon, Joshua~J. Daymude, Dana Randall, and Andr\'ea~W. Richa.
\newblock A {M}arkov chain algorithm for compression in self-organizing
  particle systems.
\newblock In {\em Proceedings of the 2016 ACM Symposium on Principles of
  Distributed Computing}, PODC '16, pages 279--288, 2016.

\bibitem{casteigts2018finding}
Arnaud Casteigts.
\newblock Finding structure in dynamic networks, 2018.
\newblock \href {http://arxiv.org/abs/1807.07801} {\path{arXiv:1807.07801}}.

\bibitem{Chazelle14}
Bernard Chazelle.
\newblock The convergence of bird flocking.
\newblock {\em J. ACM}, 61(4), 2014.

\bibitem{ClementiRT12}
Andrea Clementi, Riccardo Silvestri, and Luca Trevisan.
\newblock Information spreading in dynamic graphs.
\newblock In {\em Proceedings of the 2012 ACM Symposium on Principles of
  Distributed Computing}, page 37–46, 2012.

\bibitem{Correll2011}
Nikolaus Correll and Alcherio Martinoli.
\newblock Modeling and designing self-organized aggregation in a swarm of
  miniature robots.
\newblock {\em The Int'l J. of Robotics Research}, 30(5):615--626, 2011.

\bibitem{Dario1992}
Paolo Dario, Renzo Valleggi, Maria~Chiara Carrozza, M.~C. Montesi, and Michele
  Cocco.
\newblock Microactuators for microrobots: a critical survey.
\newblock {\em Journal of Micromechanics and Microengineering}, 2(3):141--157,
  1992.

\bibitem{Daymude2021-canonicalamoebot}
Joshua~J. Daymude, Andr{\'e}a~W. Richa, and Christian Scheideler.
\newblock {The Canonical Amoebot Model: Algorithms and Concurrency Control}.
\newblock In {\em 35th International Symposium on Distributed Computing (DISC
  2021)}, volume 209 of {\em Leibniz International Proceedings in Informatics
  (LIPIcs)}, pages 20:1--20:19, 2021.

\bibitem{Daymude2023-canonicalamoebot}
Joshua~J. Daymude, Andr{\'e}a~W. Richa, and Christian Scheideler.
\newblock The canonical amoebot model: Algorithms and concurrency control.
\newblock {\em Distributed Computing}, 2023.
\newblock To appear.

\bibitem{RandomWalksRecurringTopologies}
Oksana Denysyuk and Lu{\'i}s Rodrigues.
\newblock Random walks on evolving graphs with recurring topologies.
\newblock In {\em Distributed Computing}, pages 333--345. Springer Berlin
  Heidelberg, 2014.

\bibitem{DinitzFGN22}
Michael Dinitz, Jeremy~T. Fineman, Seth Gilbert, and Calvin Newport.
\newblock Smoothed analysis of information spreading in dynamic networks.
\newblock In Christian Scheideler, editor, {\em 36th International Symposium on
  Distributed Computing (DISC)}, volume 246 of {\em LIPIcs}, pages 18:1--18:22,
  2022.

\bibitem{dutta2013complexity}
Chinmoy Dutta, Gopal Pandurangan, Rajmohan Rajaraman, Zhifeng Sun, and Emanuele
  Viola.
\newblock On the complexity of information spreading in dynamic networks.
\newblock In {\em Proceedings of the Twenty-Fourth Annual ACM-SIAM Symposium on
  Discrete Algorithms}, pages 717--736. SIAM, 2013.

\bibitem{Fates2010}
Nazim Fat\`es.
\newblock Solving the decentralised gathering problem with a
  reaction–diffusion–chemotaxis scheme.
\newblock {\em Swarm Intelligence}, 4(2):91--115, 2010.

\bibitem{Fates2011}
Nazim Fat\`es and Nikolaos Vlassopoulos.
\newblock A robust aggregation method for quasi-blind robots in an active
  environment.
\newblock In {\em ICSI 2011}, 2011.

\bibitem{Garnier2009}
Simon Garnier, Jacques Gautrais, Masoud Asadpour, Christian Jost, and Guy
  Theraulaz.
\newblock Self-organized aggregation triggers collective decision making in a
  group of cockroach-like robots.
\newblock {\em Adaptive Behavior}, 17(2):109--133, 2009.

\bibitem{Garnier2005}
Simon Garnier, Christian Jost, Rapha\"el Jeanson, Jacques Gautrais, Masoud
  Asadpour, Gilles Caprari, and Guy Theraulaz.
\newblock Aggregation behaviour as a source of collective decision in a group
  of cockroach-like-robots.
\newblock In {\em Advances in Artificial Life}, ECAL '05, pages 169--178, 2005.

\bibitem{HaeuplerK11}
Bernhard Haeupler and David Karger.
\newblock Faster information dissemination in dynamic networks via network
  coding.
\newblock In {\em Proceedings of the 30th Annual ACM SIGACT-SIGOPS Symposium on
  Principles of Distributed Computing}, page 381–390, 2011.

\bibitem{HussakTrehan19}
Walter Hussak and Amitabh Trehan.
\newblock On termination of a flooding process.
\newblock In {\em Proc. of the 2019 ACM Symp. on Principles of Distributed
  Computing}, PODC '19, page 153–155, 2019.

\bibitem{kedia2022}
Hridesh Kedia, Shunhao Oh, and Dana Randall.
\newblock A local stochastic algorithm for alignment in self-organizing
  particle systems.
\newblock In {\em Approximation, Randomization, and Combinatorial Optimization.
  Algorithms and Techniques (APPROX/RANDOM 2022)}, volume 245, pages
  14:1--14:20, 2022.

\bibitem{KuhnLO10}
Fabian Kuhn, Nancy Lynch, and Rotem Oshman.
\newblock Distributed computation in dynamic networks.
\newblock In {\em Proceedings of the Forty-Second ACM Symposium on Theory of
  Computing}, page 513–522, 2010.

\bibitem{KuhnOshman11}
Fabian Kuhn and Rotem Oshman.
\newblock Dynamic networks: Models and algorithms.
\newblock {\em SIGACT News}, 42(1):82–96, 2011.

\bibitem{Li2021-bobbots}
Shengkai Li, Bahnisikha Dutta, Sarah Cannon, Joshua~J. Daymude, Ram Avinery,
  Enes Aydin, Andr\'{e}a~W. Richa, Daniel~I. Goldman, and Dana Randall.
\newblock Programming active granular matter with mechanically induced phase
  changes.
\newblock {\em Science Advances}, 7, 2021.

\bibitem{Liu2015}
Jintao Liu, Arthur Prindle, Jacqueline Humphries, Mar\c{c}al Gabalda-Sagarra,
  Munehiro Asally, Dong-Yeon~D. Lee, San Ly, Jordi Garcia-Ojalvo, and
  G\"urol~M. S\"uel.
\newblock Metabolic co-dependence gives rise to collective oscillations within
  biofilms.
\newblock {\em Nature}, 523(7562):550--554, 2015.

\bibitem{lovasz1993random}
L{\'a}szl{\'o} Lov{\'a}sz.
\newblock Random walks on graphs.
\newblock {\em Combinatorics, P. Erd\"{o}s is Eighty}, 2(4):1--46, 1993.

\bibitem{Magurran1990}
Anne~E. Magurran.
\newblock The adaptive significance of schooling as an anti-predator defence in
  fish.
\newblock {\em Annales Zoologici Fennici}, 27(2):51--66, 1990.

\bibitem{metropolis}
Nicholas Metropolis, Arianna~W. Rosenbluth, Marshall~N. Rosenbluth, Augusta~H.
  Teller, and Edward Teller.
\newblock {Equation of State Calculations by Fast Computing Machines}.
\newblock {\em The Journal of Chemical Physics}, 21:1087--1092, 1953.

\bibitem{Mlot2011}
Nathan~J. Mlot, Craig~A. Tovey, and David~L. Hu.
\newblock Fire ants self-assemble into waterproof rafts to survive floods.
\newblock {\em Proc. of the National Academy of Sciences}, 108(19):7669--7673,
  2011.

\bibitem{Ozdemir2018}
Anil \"Ozdemir, Melvin Gauci, Salom\'e Bonnet, and Roderich Gro\ss.
\newblock Finding consensus without computation.
\newblock {\em IEEE Robotics and Automation Letters}, 3(3):1346--1353, 2018.

\bibitem{Prindle2015}
Arthur Prindle, Jintao Liu, Munehiro Asally, San Ly, Jordi Garcia-Ojalvo, and
  G\"urol~M. S\"uel.
\newblock Ion channels enable electrical communication in bacterial
  communities.
\newblock {\em Nature}, 527(7576):59--63, 2015.

\bibitem{Sahin2005}
Erol {\c{S}}ahin.
\newblock Swarm robotics: From sources of inspiration to domains of
  application.
\newblock In {\em Swarm Robotics}, pages 10--20, 2005.

\bibitem{Savoie2018}
William Savoie, Sarah Cannon, Joshua~J. Daymude, Ross Warkentin, Shengkai Li,
  Andr\'ea~W. Richa, Dana Randall, and Daniel~I. Goldman.
\newblock Phototactic supersmarticles.
\newblock {\em Artificial Life and Robotics}, 23(4):459--468, 2018.

\bibitem{Schelling1971}
Thomas~C. Schelling.
\newblock Dynamic models of segregation.
\newblock {\em The Journal of Mathematical Sociology}, 1(2):143--186, 1971.

\bibitem{Soysal2005}
Onur Soysal and Erol {\c{S}}ahin.
\newblock Probabilistic aggregation strategies in swarm robotic systems.
\newblock In {\em Proceedings 2005 IEEE Swarm Intelligence Symposium}, SIS
  2005, pages 325--332, 2005.

\bibitem{Toffoli1991}
Tommaso Toffoli and Norman Margolus.
\newblock Programmable matter: Concepts and realization.
\newblock {\em Physica D: Nonlinear Phenomena}, 47(1):263--272, 1991.

\bibitem{Wolpert2019}
David~H. Wolpert.
\newblock The stochastic thermodynamics of computation.
\newblock {\em Journal of Physics A: Mathematical and Theoretical},
  52(19):193001, 2019.

\bibitem{Xie2019}
Hui Xie, Mengmeng Sun, Xinjian Fan, Zhihua Lin, Weinan Chen, Lei Wang, Lixin
  Dong, and Qiang He.
\newblock Reconfigurable magnetic microrobot swarm: Multimode transformation,
  locomotion, and manipulation.
\newblock {\em Science Robotics}, 4(28):eaav8006, 2019.

\end{thebibliography}

\appendix
\section{Details of proofs}
\label{appendix:proofs}
In this appendix we give in full the longer proofs which have been omitted from the main paper.
We first restate and prove Lemma~\ref{lemma:allswitchtoawarestatedynamic} by making use of Lemmas~\ref{lemma:reachpotentialzerodynamic} and~\ref{lemma:residualscantregeneratedynamic}. Showing this lemma will allow us to prove Theorem~\ref{theorem:mainresultwithstimulidynamic}.
\newtheorem*{lemmarestate1}{Lemma~\ref{lemma:allswitchtoawarestatedynamic}}
\begin{lemmarestate1}
We start from a configuration satisfying the state invariant over a $(U_D,U_C)$-recurring valid reconfigurable graph, and assume that there are no residuals, the witness set is nonempty, and no agent will be removed from the witness set from the current point on.
Then the expected number of iterations before the next agent switches from the \unaware{} to the \aware{} state is at most $O\left(n^5 \log n + \frac{U_D}{1-U_C}\right)$.
\end{lemmarestate1}

\begin{proof}
This proof largely follows the proof of Lemma~\ref{lemma:allswitchtoawarestate}, with a few modifications to allow for reconfigurability.
%In the spirit of carefulness, we will state this proof in full, without reference to the proof of Lemma~\ref{lemma:allswitchtoawarestate}.
%
%ADD: CASE WHERE THE FIRST AGENT FINDS THE WITNESS. THIS IS PROBABLY ONLY FOR COMPRESSION BUT WE NEED TO TALK ABOUT THE FIRST AGENT STILL WITHOUT COMPRESSION
Once again, if there is no \aware{} agent, a witness becomes aware in $O\left(\frac{n}{p}\right)$ time in expectation, so for the rest of the proof, we may assume every witness is already in the \aware{} state with the witness flag set.

%We consider the amount of time it takes for the next agent to switch to the \aware{} state. For a non-witness to switch to the \aware{} state, it must be activated while adjacent to an agent holding an \alerttoken. 
Similar to the proof of Lemma~\ref{lemma:allswitchtoawarestate}, we upper bound the expected amount of time it takes for all \aware{} state agents to obtain an \alerttoken, followed by the amount of time it takes for an agent to be activated while adjacent to an \alerttoken{} and switch to the \aware{} state.

To bound the amount of time it takes for all \aware{} state agents to be holding an \alerttoken, we apply the same strategy as the proof of Lemma~\ref{lemma:allswitchtoawarestate}, marking an agent without an \alerttoken{} and passing around the mark until it lands on a witness.
%
%PARA NOT NEEDED IN STATIC PROOF
However, unlike the proof of Lemma~\ref{lemma:allswitchtoawarestate}, the mark now moves over a sequence of induced subgraphs $G_{\Tstart+1}, G_{\Tstart+2}, \ldots$, where $\Tstart$ denotes the last iteration an \unaware{} state agent switched to the \aware{} state (named as such as this is the point from which we start our analysis).
%
%%%% IMPORTANT: INDEPENDENCE!!! This is a tricky part of the proof due to the connectedness and independence requirements.
%The mark moving in this manner is thus equivalent to following the $d_{max}$-random walk over
%the sequence of induced subgraphs $G_\Tstart, G_{\Tstart+1}, G_{\Tstart+2}, \ldots$,
On each iteration $t \geq \Tstart$, $G_{t+1}$ is drawn from a distribution $\mathcal{X}(X_t,\bstates_t) = \mathcal{X}(X_t,\bstates_\Tstart)$, as $\bstates_{t} = \bstates_\Tstart$ for all $t \geq \Tstart$
%
%which depends (directly and indirectly) only on past information $X_0, X_1, X_2, \ldots, X_{t-1}$ and past \behaviorgroups~$\bstates_1, \bstates_1, \ldots, \bstates_t$.
%
(note that as long as no new \unaware{} agent switches to the \aware{} state, as there are no residuals and the set of witnesses $\witset_T$ is no longer changing, 
the \behaviorgroup~vectors $\bstates_{\Tstart}, \bstates_{\Tstart+1}, \ldots, \bstates_t$ will not have changed since $\Tstart$).
%The movements of the mark are thus invisible to the behavioral algorithm, so
Thus conditioned on everything that has happened on iterations up to and including $\Tstart$, the random sequence describing the movements of the mark is independent of the random sequence describing the reconfiguration behavior of the graph.

This independence is crucial for bounding the expected time before the next \alerttoken{} is generated, as it allows us to apply the result of~\cite{RandomWalksDynamicGraphs, RandomWalksRecurringTopologies}, which states that the expected hitting time of the $d_{max}$-random walk on a connected evolving graph controlled by an \emph{oblivious adversary} is $O(n^3 \log n)$~\cite{RandomWalksRecurringTopologies}.
This polynomial time bound is notable as there are connected evolving graphs where the simple random walk admits exponential hitting times in the worst case~\cite{RandomWalksDynamicGraphs}.

There is one remaining obstacle to applying this result however. 
Denote by $A$ the set of \aware{} state agents given by $\bstates_\Tstart$. We would like to apply this result to the sequence of induced subgraphs $G_t[A], t \geq \Tstart+1$, which in general, may not be connected.
To resolve this, we observe that as the state invariant (Definition~\ref{dfn:stateinvariant}) holds, each connected component of $G_t[A]$ will have an \immobileb{} agent (which are witnesses as there are no residuals). Adding an edge between every pair of \immobileb{} agents in each $G_t[A]$ gives a sequene of connected graphs $H_t$ - we can then see the movement of the mark as a $d_{max}$-random walk over the graph sequence $H_t$ to conclude that the expected hitting time from any vertex in $A$ to any vertex corresponding to an \immobileb{} agent is $O(n^3 \log n)$.
This corresponds to $O(n^4 \log n)$ iterations in expectation to generate a new \alerttoken{} which gives an upper bound of $O(n^5 \log n)$ iterations in expectation before all agents carry \alerttokens{}.
%Note that it is also important that the random graph sequence $H_t, t \geq \Tstart$ is independent of the sequence of movements of the mark for the graph evolution to be treated as controlled by an oblivious adversary.

%A hitting time bound of $O(n^3 \log n)$ movements of the mark implies an expected time upper bound of $O(n^4 \log n)$ iterations before an witness not currently holding an \alerttoken{} is activated and generates a new \alerttoken{} (this asymptotic bound is the same even as we take into account that the probability of generating a new \alerttoken{} on activation is a non-zero constant $p < 1$).
%As there are at most $n$ \aware{} agents, all \aware{} agents will be holding an \alerttoken{} after $O(n^5 \log n)$ iterations in expectation. Note that this is a loose bound - the bound has not been optimized for clarity of explanation.

%%%%%FOR STATIC PROOF ONLY
%The remaining step is for an \unaware{} agent to be activated while adjacent to an \aware{} agent holding an \alerttoken. As the graph $G$ is connected, as long as there is at least one \unaware{} and one \aware{} agent, there will always be at least one \unaware{} agent adjacent to an \aware{} agent, so the expected number of steps before this occurs will is $O(n)$.

%PARA NOT NEEDED IN STATIC PROOF
%recurrence result here:
Now that all \aware{} agents have \alerttokens{}, all that remains is to activate an \unaware{} agent neighboring an \aware{} agent (otherwise known as an active agent) and convert it to an \aware{} agent, consuming the \alerttoken{} held by said neighbor.
Denote by $\Tfull$ the first iteration where every \aware{} agent carries an \alerttoken. As the reconfiguration adversary is $(U_D,U_C)$-recurring, the iterations following $\Tfull$ may be divided into intervals with the random variables $D_1, D_2, D_3, \ldots$ as their respective lengths, where the total numbers of active agents within the respective intervals are represented by the random variables $C_1, C_2, C_3, \ldots$.

A new \aware{} agent is added when an active \unaware{} agent is activated. The probability of adding a new \aware{} agent on a given iteration with $k$ active agents is thus $\frac{k}{n}$, as each of the $n$ agents are activated with equal probability.
Thus, if we denote by the random sequence $K_1, K_2, K_3, \ldots$ the number of active agents on each iteration following $\Tfull$ (including $\Tfull$), we get the following expression for the expected value of $X$, which we use to denote the number of iterations following $\Tfull$ before a new \aware{} agent is added:
\begin{align*}
&\mathbb{E}[X | K_1, K_2, K_3, \ldots] 
= \sum_{x = 0}^\infty Pr\left(X > x | K_1, K_2, K_3, \ldots \right) \\
&= 1 + \sum_{x = 1}^\infty \left(1 - \frac{K_1}{n}\right)\left(1 - \frac{K_2}{n}\right)\cdots\left(1 - \frac{K_{x}}{n}\right)\\
&\leq 1 + \sum_{x = 1}^\infty \left(1 - \frac{1}{n}\right)^{K_1}\left(1 - \frac{1}{n}\right)^{K_2}\cdots\left(1 - \frac{1}{n}\right)^{K_x}\\
&= 1 + \sum_{x = 1}^\infty y^{\sum_{i=1}^{x}K_i} \text{ (let $y := \left(1 - \frac{1}{n}\right) \in (0,1)$)}\\
&= 1 + \underbrace{y^{K_1} + y^{K_1+K_2} + \ldots + y^{\sum_{i=1}^{D_1} K_i}}_{D_1\text{ terms}} + \underbrace{y^{C_1+K_{D_1+1}} + \ldots + y^{C_1 + \sum_{i=D_1+1}^{D_1+D_2} K_i}}_{D_2\text{ terms}} + \ldots \\
&\leq 1 + (D_1-1) + y^{C_1} + y^{C_1}(D_2-1) + y^{C_1+C_2} + y^{C_1+C_2}(D_3-1) + y^{C_1+C_2+C_3} + \ldots \\
&\leq D_1 + D_2 \cdot y^{C_1} + D_3 \cdot y^{C_1+C_2} + \ldots
\end{align*}
Thus, via the law of total expectation, we have
\begin{align*}
\mathbb{E}[X]
\leq \sum_{i=1}^\infty \mathbb{E}[D_i y^{\sum_{j=1}^{i-1}C_j}] 
&\leq \sum_{i=1}^\infty \mathbb{E}\left[ y^{\sum_{j=1}^{i-1}C_j} \mathbb{E}[D_i | C_1,C_2,\ldots,C_{i-1}] \right] \\
&\leq \sum_{i=1}^\infty U_D \mathbb{E}\left[ y^{\sum_{j=1}^{i-2}C_j} \mathbb{E}[y^{C_{i-1}} | C_1,C_2,\ldots,C_{i-2}] \right] \\
&\leq \ldots \leq \sum_{i=1}^\infty U_D (U_C)^{i-1} = \frac{U_D}{1-U_C}.
\end{align*}
Combining the two phases, we have an expected time bound of $O\left(n^5 \log n + \frac{U_D}{1-U_C}\right)$ before a new \aware{} agent is added.
\end{proof}

We now restarte and show Lemma \ref{lemma:adaptivecompressionrecurring}, which is used in Section \ref{sec:foraging} to show Theorem~\ref{thm:compressionhasfoodprecise}.

\newtheorem*{lemmarestate1b}{Lemma~\ref{lemma:adaptivecompressionrecurring}}
\begin{lemmarestate1b}
The movement behavior defined in the Adaptive $\alpha$-Compression Algorithm is $(U_D,U_C)$-recurring with $U_D = 2N^2 + \frac{2}{n} + 1$ and $U_C = \frac{2}{3}$.
\end{lemmarestate1b}
\begin{proof}
% Arbitrary lattice, arbitrary starting state.
% Note: The reconfiguration process as a function may execute multiple steps per iteration.
% We don't run the dynamic stimuli problem in this case. We are looking at a property of the behavior sequence alone. We want to count how many times agents become active if you run the behaviors only.
We start from an arbitrary starting lattice configuration, with an arbitrary \behaviorgroup{} vector that we assume will remain fixed. We further suppose that this \behaviorgroup{} vector contains at least one \unaware{} and at least one \immobileb{} agent.
%We start from an arbitrary \behaviorgroup{} vector $\bstates$ and an arbitrary starting lattice configuration. Suppose that the \behaviorgroup{} vector remains fixed at $\bstates$ and there is at least one \unaware{} and at least one \immobileb{} agent in $\bstates$.
In adaptive $\alpha$-compression, as agent movement depends only on the current configuration of the lattice and $\bstates$, for convenience of notation and without loss of generality, we start on iteration 0.

As on each iteration the Adaptive $\alpha$-Compression Algorithm randomly chooses between executing an iteration of the dynamic stimuli problem and executing a movement on the lattice, there may be any number of movements between iterations of the dynamic stimuli problem.
Throughout this proof, to distinguish the two, an ``iteration'' refers to a iteration of the dynamic stimuli problem, while a ``movement step'' refers to to an execution of the movement algorithm.
%Throughout this proof, ``iteration'' refers to a step of the dynamic stimuli problem.

To show that the reconfiguration adversary is recurring, the random sequences $(D_1, D_2, D_3, \ldots)$ and $(C_1, C_2, C_3, \ldots)$ need to be defined - these random sequences partition the iterations into \emph{``batches''}, where the $i$\textsuperscript{th} batch would take $D_i$ iterations and would see $C_i$ active agents over its duration.
To define a single batch starting from some iteration $T$, we consider the following two events (that occur between iterations):
\begin{itemize}
\item \textbf{Step 1}: The first time after iteration $T$ where agent movement puts an \unaware{} agent $u$ next to an \aware{} agent $v$ (even if the agent moves away within the same iteration).
\item \textbf{Step 2}: The first time after step 1 happens where agent movement selects either $u$ (again) or $v$ to be moved (even if the movement behavior does not end up moving the agent).
\end{itemize}
This batch would refer to the duration starting from $T$ and ending on the first iteration $T'$ following step 2 (the next batch then starts on iteration $T'+1$). Note that it is possible that Steps 1 and 2 happen between the same two iterations, in the event that the agents $u$ and $v$ identified in the first step are again selected by the movement behavior before the next iteration occurs.

To upper bound the expected number of movement steps required for Step 1, we note that the movements of the agents in the \unaware{} state follow a simple exclusion process. By considering obstructed movements as swap moves rather than rejected moves, we can analyze the hitting time of a simple exclusion process as independent simple random walks over the lattice.
If \mobileb{} or \immobileb{} agents did not exist, as our triangular lattice is a regular graph, a simple worst case upper bound for the amount of time before some \unaware{} agent reaches the location of an \immobileb{} agent is at most $2N^2$~\cite{lovasz1993random} actions of any one agent, which translates to $2N^2n$ movement steps in expectation. This \unaware{} state agent must come into contact with an \aware{} agent (\mobileb{} or \immobileb{} agent) before this happens, so $2N^2n$ is an upper bound for Step 1.

To upper bound the expected number of movement steps required for Step 2, we note the probability of selecting $u$ or $v$ is exactly $\frac{2}{n}$. The expected value of a geometric distribution gives us an upper bound of $\frac{n}{2}+1$ movement steps in expectation.

As the number of movement steps is equal to the number of iterations of the dynamic stimuli problem in expectation, and this applies from any starting configuration, this gives us a uniform upper bound of $2N^2n + \frac{2}{n} + 1$ for $\mathbb{E}[D_k | D_1, D_2, \ldots D_{k-1}, C_1,C_2 \ldots C_{k-1}]$ for each batch $k \in \{1,2,\ldots\}$.

Similarly, as selecting either $u$ or $v$ in Step 2 is a prerequisite for $u$ and $v$ to no longer be neighbors after Step 1, $C_k$ (conditioned on past $D_i, C_i$) stochastically dominates some geometric distribution $Y$ (with some success probability we will later define). To be precise however, as $C_k$ counts the number of iterations of the dynamic stimuli problem rather than the number of movement steps, we have to take into account the fact that $Geom(\frac{1}{2})$ movement steps may occur in between any two iterations.
The success probability $p_Y$ of the geometric distribution $Y$ is:
\begin{align*}
p_y &= \sum_{i=0}^\infty \frac{1}{2} \cdot \left(\frac{1}{2}\right)^i \left(1 - \left(1-\frac{2}{n}\right)^i \right)
= 1 - \frac{1}{2}\sum_{i=0}^\infty \left( \frac{1 - 2/n}{2} \right)^i
= \frac{2}{n}\left( \frac{1}{1 + 2/n} \right)
\end{align*}

As $C_k$ stochastically dominates $Y$ and $(1-\frac{1}{n})^x$ is a decreasing function of $x$, we have:
\begin{align*}
\mathbb{E}\left[\left(1-\frac{1}{n}\right)^{C_k} | D_1, D_2, \ldots D_{k-1}, C_1,C_2 \ldots C_{k-1}\right]
&\leq \mathbb{E}\left[\left(1-\frac{p}{n}\right)^Y\right] \\
&= \sum_{y=0}^\infty \left(1 - \frac{1}{n}\right)^y p_Y \left(1 - p_Y\right)^y \\
&= p_Y \frac{1}{1 - (1-1/n)(1-p_Y)} 
= \frac{2}{3}. \qedhere
\end{align*}
%\vspace{-.4in}
\end{proof}
% appendix_static is now no longer needed.
%\input{appendix_static}

%\input{appendix_reconfigurable}

%\input{appendix_compression}

% This appendix is temporarily commented out as it slows down compile time.
%\documentclass{article}
%\usepackage[letterpaper]{geometry}
%\usepackage{amsmath}
%\usepackage{amsfonts}
%\usepackage{amsthm}
%\usepackage{subcaption}
%\usepackage{graphicx}

%\newcommand{\floor}[1]{\lfloor #1 \rfloor}
%\newcommand{\ceil}[1]{\lceil #1 \rceil}
%\newcommand{\ex}{\mathbb{E}}

%\newcommand{\com}[1]{\textbf{\textcolor{red}{#1}}}

%\newtheorem{theorem}{Theorem}[section]
%\newtheorem{definition}[theorem]{Definition}
%\newtheorem{remark}[theorem]{Remark}
%\newtheorem{lemma}[theorem]{Lemma}
%\newtheorem{corollary}[theorem]{Corollary}
%\newtheorem{proposition}[theorem]{Proposition}
%\newtheorem{claim}[theorem]{Claim}
%\newtheorem{observation}[theorem]{Observation}
%\newtheorem{fact}[theorem]{Fact}
%\newtheorem{assumption}[theorem]{Assumption}
%\newtheorem{property}[theorem]{Property}
%\newtheorem{procedure}[theorem]{Procedure}

% TERMINOLOGY - START
\newcommand{\spine}{{spine}}
\newcommand{\Spine}{{Spine}}
\newcommand{\spines}{{spines}}
\newcommand{\SpineComb}{{Spine Comb}}
\newcommand{\spinecomb}{{spine comb}}
\newcommand{\spinecombs}{{spine combs}}
\newcommand{\mainspine}{{main spine}}
\newcommand{\anchoragent}{{anchor agent}}
\newcommand{\anchoragents}{{anchor agents}}
\newcommand{\tailagent}{{tail agent}}
\newcommand{\tailagents}{{tail agents}}
\newcommand{\sourcespine}{{source \spine{}}}
\newcommand{\targetspine}{{target \spine{}}}
\newcommand{\ResidualRegion}{{Residual Region}}
\newcommand{\residualregion}{{residual region}}

% TERMINOLOGY - END

% DIAGRAMS NEEDED
% Illustration of the six spines
% Illustration of tail agents and spine lengths
% Illustration of coordinate system
%   - do this by giving coordinates of anchor agents
%     (maybe illustrate a comb too)
%     (maybe illustrate the residual region too)
%
% Combed - illustration of a combed config from (l,d)
% Comb procedure (shift operation, two cases)
% Comb procedure (line formation, no shift)
% Comb procedure (before line merging)
% Comb procedure (after line merging)
%
% L: Lemma 1.7 (Unaffected Region Above)
% L: Lemma 1.13 (Unenterable Region Below)
%
% Combable sequence (illustrate the empty line above)
% Spine Comb ^ probably use same diag
% Result of a spine comb (with and without gap in line)
%
% Hexagon with a tail
% Cases for making gap from hexagon
%    - corner (upper left, with tail)
%    - side (two cases)

%\begin{document}

\section{Ergodicity of the Markov chain for compression}
\label{apx:irreducibilityproof}
We conclude by proving Lemma~\ref{lemma:irreducible}, showing that the Markov chain for compression is irreducible (and thus ergodic) by showing that all states are reachable using compression moves in the presence of an immobile agent on the food source (as in described at the end of Section~\ref{sec:foraging}).  The proof of ergodicity given in Cannon et al. \cite{Cannon2016} without an immobile agent was already fairly involved, but with the addition of an agent that cannot itself move, the proof becomes substantially more challenging. 

The main strategy in the proof is to treat the immobile agent on the food source as the ``center'' of the configuration, and consider the lines extending from the center in each of the six possible directions. These lines, which we call ``spines'', divide the lattice into six regions.
The sequence of moves described in \cite{Cannon2016} is then modified to operate within one of these regions, with limited side effects on the two regions counterclockwise from this region. We call this sequence of moves a ``comb'', and show that there is a sequence of comb operations that can be applied to the configuration, repeatedly going round the six regions in counterclockwise order, until the resulting configuration is a single long line.
We then observe that any valid compression move transforming a hole-free configuration to a hole-free configuration is also valid in the reverse direction, giving us the statement of the lemma.

%We overcome this by introducing a ``comb'' operation that organizes agents radially around the immobile agent.

% intro of the proof
We will treat the immobile agent on the food source as the ``center'' of the configuration. From the center, there are six directions one can move in a straight line on the triangular lattice - up, down, up-left, down-left, up-right, down-right. We call these six straight lines of agents extending from the immobile agent \emph{\spines{}}. We refer to agents on the \spines{} as \spine{} agents, and agents not on \spines{} as non-\spine{} agents. We similarly use the names \spine{} and non-\spine{} locations to refer to sites on the triangular lattices.

%Define tail agents
On each \spine{}, we call the furthest out (in terms of distance from the immobile agent) agent on the \spine{} with no adjacent non-\spine{} agents the \emph{\anchoragent{}} of the \spine{}. The \spine{} agents further out than the \anchoragent{} are called \emph{\tailagents{}}. The \emph{distance} of a \spine{} location from the center refers to its shortest-path distance (which would be along the \spine{}) to the immobile agent on the triangular lattice. For each integer $r \geq 1$, the hexagon of radius $r$ refers to the regular hexagon with corners defined by the six \spine{} locations of distance $r$ from the center. The distance of a non-\spine{} location from the center would then be the radius of the smallest such hexagon it is contained within.
%Define spine lengths
An important concept that we will use in the proof is the length of a \spine{}. The length of a \spine{} is defined to be the distance of its \anchoragent{} to the center. If it has no \anchoragent{}, the length of the \spine{} is $0$.

\begin{figure*}[t]
%\centering
%\includegraphics[width=.4\linewidth]{diagrams_irreducibility/spines.png}
\begin{center}
\begin{tikzpicture}[x=0.6cm,y=0.6cm]
\input{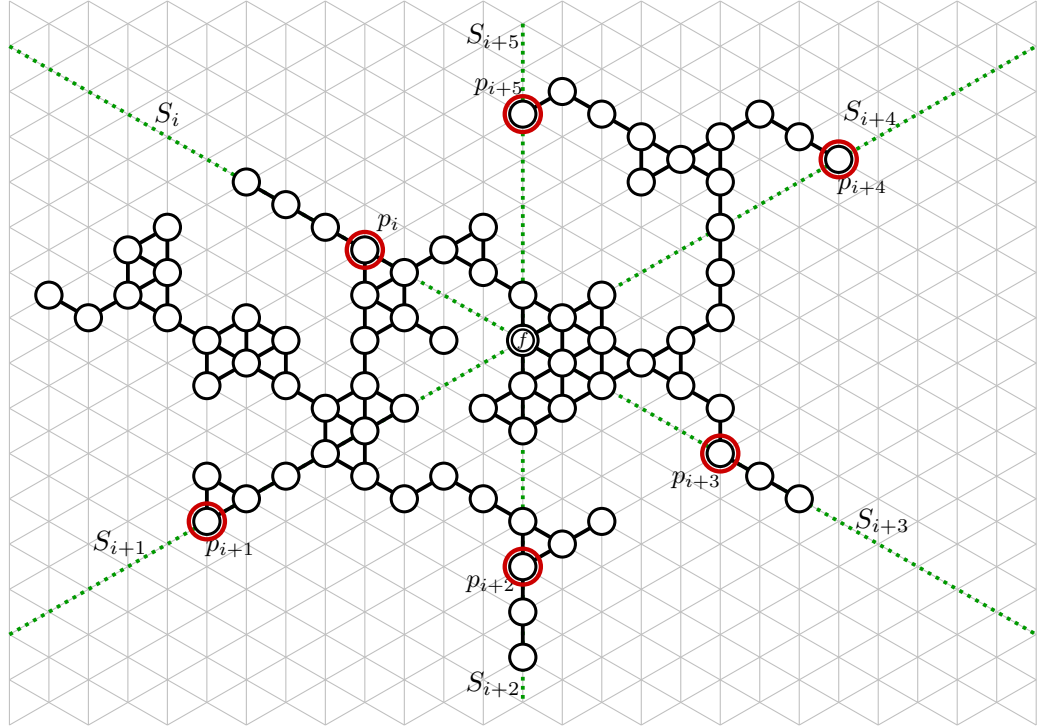}
\end{tikzpicture}
\end{center}
\caption{Illustration of the \spines{} extending from the immobile agent $f$.
The six \spines{} $S_i,\ldots,S_{i+5}$ have lengths $4,8,5,5,8,5$ respectively. These \spines{} have respective \anchoragents{} $p_i,\ldots,p_{i+5}$. As an illustration of the coordinate system, these six anchor agents are at coordinates $(4,0)$, $(8,8)$, $(0,5)$, $(-5,0)$, $(-8,-8)$ and $(0,-5)$ respectively.}
\label{fig:spines}
\end{figure*}

%Spine notation
We notate the six \spines{} using one of the \spines{} as a reference \spine{}. If the reference \spine{} is denoted $S_i$, where $i$ is an integer modulo $6$, then $S_{i+1}, S_{i+2}, \dots, S_{i+5}$ denote the subsequent \spines{} in a counterclockwise order from $S_i$.

The proof centers around a specific transformation we call a ``comb'' operation. This comb operation is applied from one \spine{} (which we refer to as the \sourcespine{}) to an adjacent \spine{} (which we refer to as the \targetspine{}), and has the effect of ``pushing'' the agents between the two \spines{} towards the \targetspine{}.

Our system exhibits reflection symmetry and 6-fold rotational symmetry, so this comb operation can be defined in $6\times 2 = 12$ different ways. However, for simplicity of discussion, we will only define the comb operation in one orientation, specifically on the left side, downwards. This is a comb from the \spine{} going in the up-left direction to the \spine{} going into the down-left direction. We rotate or reflect the configuration freely, depending on which pair of adjacent \spines{} we want to comb between. 

% Define Combing
\subsection{The comb operation}
We define our two-dimensional coordinate system $(lane,depth)$ with reference to the \sourcespine{}, assumed to be going in the up-left direction. A position $(\ell,0)$ for $\ell \geq 0$ refers to the position on the \sourcespine{} $\ell$ steps away from the immobile agent. If $\ell < 0$, this refers to the position $-\ell$ steps in the direction of the \spine{} directly opposite the \sourcespine{}. A position $(\ell,d)$ refers to the location $d$ steps downwards from position $(\ell,0)$. Thus, $d$ denotes the (signed) distance of the position from the \sourcespine{}.

\begin{figure*}
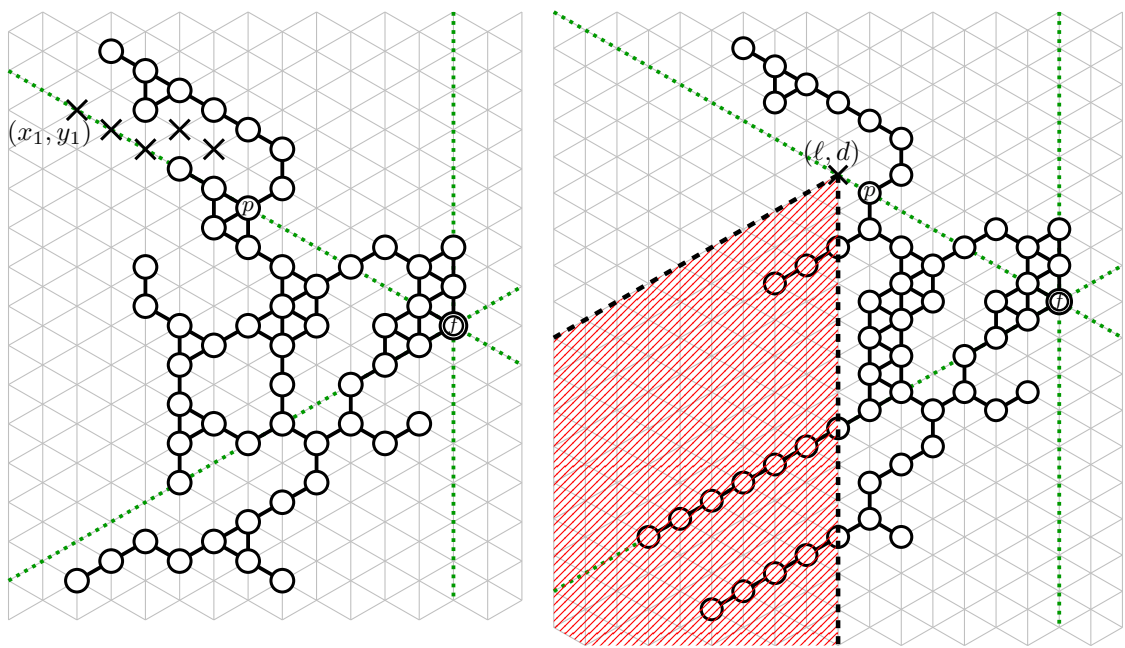

\begin{subfigure}[t]{.48\textwidth}
  %\centering
  %\includegraphics[width=.6\linewidth]{diagrams_irreducibility/spine_comb.png}
  \begin{center}
  \begin{tikzpicture}[x=0.52cm,y=0.52cm]
  \input{diagrams_tex/spine_comb.tex}
  \end{tikzpicture}
  \end{center}
  \caption{Illustration of a \spinecomb{} (Definition~\ref{proc:spinecomb}), which in this case is a comb over the sequence (Definition~\ref{proc:combingasequence}) denoted by crosses in the Figure.}
  \label{fig:spine_comb}
\end{subfigure}%
\hfill
\begin{subfigure}[t]{.48\textwidth}
  %\centering
  %\includegraphics[width=.65\linewidth]{diagrams_irreducibility/combed_residual_region.png}
  \begin{center}
  \begin{tikzpicture}[x=0.48cm,y=0.48cm]
  \input{diagrams_tex/combed_residual_region.tex}
  \end{tikzpicture}
  \end{center}
  \caption{After combing $(\ell,d)$, position $(\ell,d)$ is combed (Definition~\ref{defn:combed}). The shaded region is the \residualregion{} of $(\ell,d)$ (Definition~\ref{defn:residualregion})}
  \label{fig:combed_residual_region}
\end{subfigure}%
\caption{The comb operation is applied in to the five points marked with crosses in Figure~\ref{fig:spine_comb} from left to right in sequence, starting from $(x_1,y_1)$. Figure~\ref{fig:combed_residual_region} illustrates the end result.}
\label{fig:comb_before_after}
\end{figure*}

Before we define the comb operation, the following definitions tells us what can and cannot be combed.

\begin{definition}[\ResidualRegion{}]
\label{defn:residualregion}
Consider a position $(\ell,d)$ and the diagonal half-line extending down-left from $(\ell,d)$, including $(\ell,d)$ itself. The \emph{\residualregion{}} of this position refers to the set of all positions on or below this half-line (Figure~\ref{fig:combed_residual_region}).
\end{definition}

\begin{definition}[Combed]
\label{defn:combed}
For $\ell > 0$ and $d \geq 0$, we say a position $(\ell,d)$ is \emph{combed} (Figure~\ref{fig:combed_residual_region}) if:
\begin{enumerate}
\item All sites directly above a topmost agent of the \residualregion{} of $(\ell,d)$ are empty.
\item All agents in the \residualregion{} of $(\ell,d)$ form straight lines stretching down and left.
\item Consider the column of sites directly to the right of the \residualregion{} of $(\ell,d)$. Each of the abovementioned lines of agents stretches down and left from an agent from this column with no agent directly below.
\end{enumerate}
\end{definition}

\begin{definition}[Combable]
\label{defn:combable}
For $\ell > 0$ and $d \geq 0$, we say a position $(\ell,d)$ is \emph{combable} if:
\begin{enumerate}
\item The position $(\ell+1,d+1)$, which is one step diagonally down-left from $(\ell,d)$, is combed.
\item The site directly above $(\ell,d)$ is empty.
\end{enumerate}
\end{definition}

With this, we can define the comb procedure. 
Aside from operating only below a given depth, this comb procedure is identical to the process used the proof of irreducibility in \cite{Cannon2016}. As this operation is covered in detail in said paper, we will be brief with its explanation here.

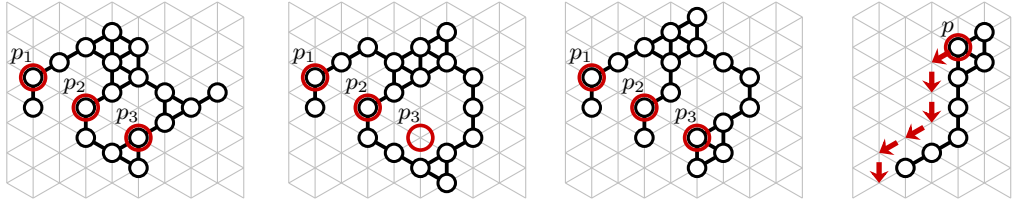
\begin{figure*}
  \begin{subfigure}{.26\textwidth}
    \centering
    \begin{center}
    \begin{tikzpicture}[x=0.4cm,y=0.4cm]
    \draw[lightgray] (2.59808,-0.5) -- (2.59808,-6.5);
\draw[lightgray] (0,-2) -- (3.4641,0);
\draw[lightgray] (0,-2) -- (7.79423,-6.5);
\draw[lightgray] (4.33013,-6.5) -- (7.79423,-4.5);
\draw[lightgray] (6.9282,0) -- (7.79423,-0.5);
\draw[lightgray] (6.9282,0) -- (6.9282,-6);
\draw[lightgray] (0.866025,-6.5) -- (7.79423,-2.5);
\draw[lightgray] (0,-3) -- (5.19615,0);
\draw[lightgray] (0,-3) -- (6.06218,-6.5);
\draw[lightgray] (3.4641,0) -- (7.79423,-2.5);
\draw[lightgray] (3.4641,0) -- (3.4641,-6);
\draw[lightgray] (7.79423,-0.5) -- (7.79423,-6.5);
\draw[lightgray] (6.06218,-6.5) -- (7.79423,-5.5);
\draw[lightgray] (2.59808,-6.5) -- (7.79423,-3.5);
\draw[lightgray] (0,-4) -- (6.9282,0);
\draw[lightgray] (0,-4) -- (4.33013,-6.5);
\draw[lightgray] (4.33013,-0.5) -- (4.33013,-6.5);
\draw[lightgray] (0,0) -- (7.79423,-4.5);
\draw[lightgray] (0,0) -- (0,-6);
\draw[lightgray] (5.19615,0) -- (7.79423,-1.5);
\draw[lightgray] (5.19615,0) -- (5.19615,-6);
\draw[lightgray] (0,-5) -- (7.79423,-0.5);
\draw[lightgray] (0,-5) -- (2.59808,-6.5);
\draw[lightgray] (0.866025,-0.5) -- (0.866025,-6.5);
\draw[lightgray] (0,-1) -- (1.73205,0);
\draw[lightgray] (0,-1) -- (7.79423,-5.5);
\draw[lightgray] (6.06218,-0.5) -- (6.06218,-6.5);
\draw[lightgray] (0,-6) -- (7.79423,-1.5);
\draw[lightgray] (0,-6) -- (0.866025,-6.5);
\draw[lightgray] (1.73205,0) -- (7.79423,-3.5);
\draw[lightgray] (1.73205,0) -- (1.73205,-6);
\draw[black, line width=0.4mm, fill=white] (0.866025,-2.5) circle (0.288);
\node[align=left] at (0.519615,-1.8) {\small $p_1$};
\draw[black, line width=0.4mm, fill=white] (0.866025,-3.5) circle (0.288);
\draw[black, line width=0.4mm, fill=white] (1.73205,-2) circle (0.288);
\draw[black, line width=0.4mm, fill=white] (2.59808,-1.5) circle (0.288);
\draw[black, line width=0.4mm, fill=white] (2.59808,-3.5) circle (0.288);
\node[align=left] at (2.25167,-2.8) {\small $p_2$};
\draw[black, line width=0.4mm, fill=white] (2.59808,-4.5) circle (0.288);
\draw[black, line width=0.4mm, fill=white] (3.4641,-1) circle (0.288);
\draw[black, line width=0.4mm, fill=white] (3.4641,-2) circle (0.288);
\draw[black, line width=0.4mm, fill=white] (3.4641,-3) circle (0.288);
\draw[black, line width=0.4mm, fill=white] (3.4641,-5) circle (0.288);
\draw[black, line width=0.4mm, fill=white] (4.33013,-1.5) circle (0.288);
\draw[black, line width=0.4mm, fill=white] (4.33013,-2.5) circle (0.288);
\draw[black, line width=0.4mm, fill=white] (4.33013,-4.5) circle (0.288);
\node[align=left] at (3.98372,-3.8) {\small $p_3$};
\draw[black, line width=0.4mm, fill=white] (4.33013,-5.5) circle (0.288);
\draw[black, line width=0.4mm, fill=white] (5.19615,-3) circle (0.288);
\draw[black, line width=0.4mm, fill=white] (5.19615,-4) circle (0.288);
\draw[black, line width=0.4mm, fill=white] (6.06218,-3.5) circle (0.288);
\draw[black, line width=0.4mm, fill=white] (6.9282,-3) circle (0.288);
\draw[black, line width=0.5mm] (1.12583,-2.35) -- (1.47224,-2.15);
\draw[black, line width=0.5mm] (5.19615,-3.7) -- (5.19615,-3.3);
\draw[black, line width=0.5mm] (5.45596,-3.85) -- (5.80237,-3.65);
\draw[black, line width=0.5mm] (3.72391,-1.15) -- (4.07032,-1.35);
\draw[black, line width=0.5mm] (0.866025,-3.2) -- (0.866025,-2.8);
\draw[black, line width=0.5mm] (2.85788,-3.35) -- (3.20429,-3.15);
\draw[black, line width=0.5mm] (4.58993,-4.35) -- (4.93634,-4.15);
\draw[black, line width=0.5mm] (4.33013,-5.2) -- (4.33013,-4.8);
\draw[black, line width=0.5mm] (3.72391,-4.85) -- (4.07032,-4.65);
\draw[black, line width=0.5mm] (3.72391,-5.15) -- (4.07032,-5.35);
\draw[black, line width=0.5mm] (6.32199,-3.35) -- (6.6684,-3.15);
\draw[black, line width=0.5mm] (2.85788,-1.35) -- (3.20429,-1.15);
\draw[black, line width=0.5mm] (2.85788,-1.65) -- (3.20429,-1.85);
\draw[black, line width=0.5mm] (4.33013,-2.2) -- (4.33013,-1.8);
\draw[black, line width=0.5mm] (4.58993,-2.65) -- (4.93634,-2.85);
\draw[black, line width=0.5mm] (5.45596,-3.15) -- (5.80237,-3.35);
\draw[black, line width=0.5mm] (3.4641,-2.7) -- (3.4641,-2.3);
\draw[black, line width=0.5mm] (3.72391,-2.85) -- (4.07032,-2.65);
\draw[black, line width=0.5mm] (1.99186,-1.85) -- (2.33827,-1.65);
\draw[black, line width=0.5mm] (3.4641,-1.7) -- (3.4641,-1.3);
\draw[black, line width=0.5mm] (3.72391,-1.85) -- (4.07032,-1.65);
\draw[black, line width=0.5mm] (3.72391,-2.15) -- (4.07032,-2.35);
\draw[black, line width=0.5mm] (2.59808,-4.2) -- (2.59808,-3.8);
\draw[black, line width=0.5mm] (2.85788,-4.65) -- (3.20429,-4.85);
\draw[black!20!red, line width=0.48mm, ] (0.866025,-2.5) circle (0.408);
\draw[black!20!red, line width=0.48mm, ] (2.59808,-3.5) circle (0.408);
\draw[black!20!red, line width=0.48mm, ] (4.33013,-4.5) circle (0.408);
    \end{tikzpicture}
    \end{center}
    \caption{Shift $p_1,p_2$ down-right.}
    \label{fig:shiftable_1}
  \end{subfigure}%
  \hfill
  \begin{subfigure}{.27\textwidth}
    \centering
    \begin{center}
    \begin{tikzpicture}[x=0.4cm,y=0.4cm]
    \draw[lightgray] (2.59808,-0.5) -- (2.59808,-6.5);
\draw[lightgray] (0,-2) -- (3.4641,0);
\draw[lightgray] (0,-2) -- (7.79423,-6.5);
\draw[lightgray] (4.33013,-6.5) -- (7.79423,-4.5);
\draw[lightgray] (6.9282,0) -- (7.79423,-0.5);
\draw[lightgray] (6.9282,0) -- (6.9282,-6);
\draw[lightgray] (0.866025,-6.5) -- (7.79423,-2.5);
\draw[lightgray] (0,-3) -- (5.19615,0);
\draw[lightgray] (0,-3) -- (6.06218,-6.5);
\draw[lightgray] (3.4641,0) -- (7.79423,-2.5);
\draw[lightgray] (3.4641,0) -- (3.4641,-6);
\draw[lightgray] (7.79423,-0.5) -- (7.79423,-6.5);
\draw[lightgray] (6.06218,-6.5) -- (7.79423,-5.5);
\draw[lightgray] (2.59808,-6.5) -- (7.79423,-3.5);
\draw[lightgray] (0,-4) -- (6.9282,0);
\draw[lightgray] (0,-4) -- (4.33013,-6.5);
\draw[lightgray] (4.33013,-0.5) -- (4.33013,-6.5);
\draw[lightgray] (0,0) -- (7.79423,-4.5);
\draw[lightgray] (0,0) -- (0,-6);
\draw[lightgray] (5.19615,0) -- (7.79423,-1.5);
\draw[lightgray] (5.19615,0) -- (5.19615,-6);
\draw[lightgray] (0,-5) -- (7.79423,-0.5);
\draw[lightgray] (0,-5) -- (2.59808,-6.5);
\draw[lightgray] (0.866025,-0.5) -- (0.866025,-6.5);
\draw[lightgray] (0,-1) -- (1.73205,0);
\draw[lightgray] (0,-1) -- (7.79423,-5.5);
\draw[lightgray] (6.06218,-0.5) -- (6.06218,-6.5);
\draw[lightgray] (0,-6) -- (7.79423,-1.5);
\draw[lightgray] (0,-6) -- (0.866025,-6.5);
\draw[lightgray] (1.73205,0) -- (7.79423,-3.5);
\draw[lightgray] (1.73205,0) -- (1.73205,-6);
\draw[black, line width=0.4mm, fill=white] (0.866025,-2.5) circle (0.288);
\node[align=left] at (0.519615,-1.8) {\small $p_1$};
\draw[black, line width=0.4mm, fill=white] (0.866025,-3.5) circle (0.288);
\draw[black, line width=0.4mm, fill=white] (1.73205,-2) circle (0.288);
\draw[black, line width=0.4mm, fill=white] (2.59808,-1.5) circle (0.288);
\draw[black, line width=0.4mm, fill=white] (2.59808,-3.5) circle (0.288);
\node[align=left] at (2.25167,-2.8) {\small $p_2$};
\draw[black, line width=0.4mm, fill=white] (2.59808,-4.5) circle (0.288);
\draw[black, line width=0.4mm, fill=white] (3.4641,-2) circle (0.288);
\draw[black, line width=0.4mm, fill=white] (3.4641,-3) circle (0.288);
\draw[black, line width=0.4mm, fill=white] (3.4641,-5) circle (0.288);
\draw[black, line width=0.4mm, fill=white] (4.33013,-1.5) circle (0.288);
\draw[black, line width=0.4mm, fill=white] (4.33013,-2.5) circle (0.288);
\node[align=left] at (3.98372,-3.8) {\small $p_3$};
\draw[black, line width=0.4mm, fill=white] (4.33013,-5.5) circle (0.288);
\draw[black, line width=0.4mm, fill=white] (5.19615,-1) circle (0.288);
\draw[black, line width=0.4mm, fill=white] (5.19615,-2) circle (0.288);
\draw[black, line width=0.4mm, fill=white] (5.19615,-5) circle (0.288);
\draw[black, line width=0.4mm, fill=white] (5.19615,-6) circle (0.288);
\draw[black, line width=0.4mm, fill=white] (6.06218,-2.5) circle (0.288);
\draw[black, line width=0.4mm, fill=white] (6.06218,-3.5) circle (0.288);
\draw[black, line width=0.4mm, fill=white] (6.06218,-4.5) circle (0.288);
\draw[black, line width=0.5mm] (6.06218,-3.2) -- (6.06218,-2.8);
\draw[black, line width=0.5mm] (0.866025,-3.2) -- (0.866025,-2.8);
\draw[black, line width=0.5mm] (5.19615,-5.7) -- (5.19615,-5.3);
\draw[black, line width=0.5mm] (5.19615,-1.7) -- (5.19615,-1.3);
\draw[black, line width=0.5mm] (5.45596,-2.15) -- (5.80237,-2.35);
\draw[black, line width=0.5mm] (4.58993,-1.35) -- (4.93634,-1.15);
\draw[black, line width=0.5mm] (4.58993,-1.65) -- (4.93634,-1.85);
\draw[black, line width=0.5mm] (1.12583,-2.35) -- (1.47224,-2.15);
\draw[black, line width=0.5mm] (2.85788,-3.35) -- (3.20429,-3.15);
\draw[black, line width=0.5mm] (4.58993,-5.35) -- (4.93634,-5.15);
\draw[black, line width=0.5mm] (4.58993,-5.65) -- (4.93634,-5.85);
\draw[black, line width=0.5mm] (1.99186,-1.85) -- (2.33827,-1.65);
\draw[black, line width=0.5mm] (2.59808,-4.2) -- (2.59808,-3.8);
\draw[black, line width=0.5mm] (2.85788,-4.65) -- (3.20429,-4.85);
\draw[black, line width=0.5mm] (3.72391,-5.15) -- (4.07032,-5.35);
\draw[black, line width=0.5mm] (3.72391,-1.85) -- (4.07032,-1.65);
\draw[black, line width=0.5mm] (3.72391,-2.15) -- (4.07032,-2.35);
\draw[black, line width=0.5mm] (5.45596,-4.85) -- (5.80237,-4.65);
\draw[black, line width=0.5mm] (2.85788,-1.65) -- (3.20429,-1.85);
\draw[black, line width=0.5mm] (6.06218,-4.2) -- (6.06218,-3.8);
\draw[black, line width=0.5mm] (3.4641,-2.7) -- (3.4641,-2.3);
\draw[black, line width=0.5mm] (3.72391,-2.85) -- (4.07032,-2.65);
\draw[black, line width=0.5mm] (4.33013,-2.2) -- (4.33013,-1.8);
\draw[black, line width=0.5mm] (4.58993,-2.35) -- (4.93634,-2.15);
\draw[black!20!red, line width=0.48mm, ] (0.866025,-2.5) circle (0.408);
\draw[black!20!red, line width=0.48mm, ] (2.59808,-3.5) circle (0.408);
\draw[black!20!red, line width=0.48mm, ] (4.33013,-4.5) circle (0.408);
    \end{tikzpicture}
    \end{center}
    \caption{Shift $p_1,p_2$ down-right.}
    \label{fig:shiftable_2}
  \end{subfigure}%
  \begin{subfigure}{.25\textwidth}
    \centering
    \begin{center}
    \begin{tikzpicture}[x=0.4cm,y=0.4cm]
    \draw[lightgray] (2.59808,-0.5) -- (2.59808,-6.5);
\draw[lightgray] (0,-2) -- (3.4641,0);
\draw[lightgray] (0,-2) -- (7.79423,-6.5);
\draw[lightgray] (4.33013,-6.5) -- (7.79423,-4.5);
\draw[lightgray] (6.9282,0) -- (7.79423,-0.5);
\draw[lightgray] (6.9282,0) -- (6.9282,-6);
\draw[lightgray] (0.866025,-6.5) -- (7.79423,-2.5);
\draw[lightgray] (0,-3) -- (5.19615,0);
\draw[lightgray] (0,-3) -- (6.06218,-6.5);
\draw[lightgray] (3.4641,0) -- (7.79423,-2.5);
\draw[lightgray] (3.4641,0) -- (3.4641,-6);
\draw[lightgray] (7.79423,-0.5) -- (7.79423,-6.5);
\draw[lightgray] (6.06218,-6.5) -- (7.79423,-5.5);
\draw[lightgray] (2.59808,-6.5) -- (7.79423,-3.5);
\draw[lightgray] (0,-4) -- (6.9282,0);
\draw[lightgray] (0,-4) -- (4.33013,-6.5);
\draw[lightgray] (4.33013,-0.5) -- (4.33013,-6.5);
\draw[lightgray] (0,0) -- (7.79423,-4.5);
\draw[lightgray] (0,0) -- (0,-6);
\draw[lightgray] (5.19615,0) -- (7.79423,-1.5);
\draw[lightgray] (5.19615,0) -- (5.19615,-6);
\draw[lightgray] (0,-5) -- (7.79423,-0.5);
\draw[lightgray] (0,-5) -- (2.59808,-6.5);
\draw[lightgray] (0.866025,-0.5) -- (0.866025,-6.5);
\draw[lightgray] (0,-1) -- (1.73205,0);
\draw[lightgray] (0,-1) -- (7.79423,-5.5);
\draw[lightgray] (6.06218,-0.5) -- (6.06218,-6.5);
\draw[lightgray] (0,-6) -- (7.79423,-1.5);
\draw[lightgray] (0,-6) -- (0.866025,-6.5);
\draw[lightgray] (1.73205,0) -- (7.79423,-3.5);
\draw[lightgray] (1.73205,0) -- (1.73205,-6);
\draw[black, line width=0.4mm, fill=white] (0.866025,-2.5) circle (0.288);
\node[align=left] at (0.519615,-1.8) {\small $p_1$};
\draw[black, line width=0.4mm, fill=white] (0.866025,-3.5) circle (0.288);
\draw[black, line width=0.4mm, fill=white] (1.73205,-2) circle (0.288);
\draw[black, line width=0.4mm, fill=white] (2.59808,-1.5) circle (0.288);
\draw[black, line width=0.4mm, fill=white] (2.59808,-3.5) circle (0.288);
\node[align=left] at (2.25167,-2.8) {\small $p_2$};
\draw[black, line width=0.4mm, fill=white] (2.59808,-4.5) circle (0.288);
\draw[black, line width=0.4mm, fill=white] (3.4641,-1) circle (0.288);
\draw[black, line width=0.4mm, fill=white] (3.4641,-2) circle (0.288);
\draw[black, line width=0.4mm, fill=white] (3.4641,-3) circle (0.288);
\draw[black, line width=0.4mm, fill=white] (4.33013,-0.5) circle (0.288);
\draw[black, line width=0.4mm, fill=white] (4.33013,-1.5) circle (0.288);
\draw[black, line width=0.4mm, fill=white] (4.33013,-4.5) circle (0.288);
\node[align=left] at (3.98372,-3.8) {\small $p_3$};
\draw[black, line width=0.4mm, fill=white] (4.33013,-5.5) circle (0.288);
\draw[black, line width=0.4mm, fill=white] (5.19615,-2) circle (0.288);
\draw[black, line width=0.4mm, fill=white] (5.19615,-4) circle (0.288);
\draw[black, line width=0.4mm, fill=white] (5.19615,-5) circle (0.288);
\draw[black, line width=0.4mm, fill=white] (6.06218,-2.5) circle (0.288);
\draw[black, line width=0.4mm, fill=white] (6.06218,-3.5) circle (0.288);
\draw[black, line width=0.5mm] (1.12583,-2.35) -- (1.47224,-2.15);
\draw[black, line width=0.5mm] (5.45596,-3.85) -- (5.80237,-3.65);
\draw[black, line width=0.5mm] (3.72391,-0.85) -- (4.07032,-0.65);
\draw[black, line width=0.5mm] (3.72391,-1.15) -- (4.07032,-1.35);
\draw[black, line width=0.5mm] (0.866025,-3.2) -- (0.866025,-2.8);
\draw[black, line width=0.5mm] (4.33013,-1.2) -- (4.33013,-0.8);
\draw[black, line width=0.5mm] (4.58993,-1.65) -- (4.93634,-1.85);
\draw[black, line width=0.5mm] (2.85788,-3.35) -- (3.20429,-3.15);
\draw[black, line width=0.5mm] (4.58993,-4.35) -- (4.93634,-4.15);
\draw[black, line width=0.5mm] (4.58993,-4.65) -- (4.93634,-4.85);
\draw[black, line width=0.5mm] (4.33013,-5.2) -- (4.33013,-4.8);
\draw[black, line width=0.5mm] (4.58993,-5.35) -- (4.93634,-5.15);
\draw[black, line width=0.5mm] (6.06218,-3.2) -- (6.06218,-2.8);
\draw[black, line width=0.5mm] (2.85788,-1.35) -- (3.20429,-1.15);
\draw[black, line width=0.5mm] (2.85788,-1.65) -- (3.20429,-1.85);
\draw[black, line width=0.5mm] (5.45596,-2.15) -- (5.80237,-2.35);
\draw[black, line width=0.5mm] (3.4641,-2.7) -- (3.4641,-2.3);
\draw[black, line width=0.5mm] (1.99186,-1.85) -- (2.33827,-1.65);
\draw[black, line width=0.5mm] (3.4641,-1.7) -- (3.4641,-1.3);
\draw[black, line width=0.5mm] (3.72391,-1.85) -- (4.07032,-1.65);
\draw[black, line width=0.5mm] (2.59808,-4.2) -- (2.59808,-3.8);
\draw[black, line width=0.5mm] (5.19615,-4.7) -- (5.19615,-4.3);
\draw[black!20!red, line width=0.48mm, ] (0.866025,-2.5) circle (0.408);
\draw[black!20!red, line width=0.48mm, ] (2.59808,-3.5) circle (0.408);
\draw[black!20!red, line width=0.48mm, ] (4.33013,-4.5) circle (0.408);
    \end{tikzpicture}
    \end{center}
    \caption{Shift $p_2,p_3$ up-left.}
    \label{fig:shiftable_3}
  \end{subfigure}%
  \begin{subfigure}{.22\textwidth}
    \centering
    \begin{center}
    \begin{tikzpicture}[x=0.4cm,y=0.4cm]
    \draw[lightgray] (3.4641,-7) -- (5.19615,-6);
\draw[lightgray] (2.59808,-0.5) -- (5.19615,-2);
\draw[lightgray] (2.59808,-0.5) -- (2.59808,-6.5);
\draw[lightgray] (0,-2) -- (2.59808,-0.5);
\draw[lightgray] (0,-2) -- (5.19615,-5);
\draw[lightgray] (0,-7) -- (5.19615,-4);
\draw[lightgray] (1.73205,-1) -- (1.73205,-7);
\draw[lightgray] (0,-3) -- (4.33013,-0.5);
\draw[lightgray] (0,-3) -- (5.19615,-6);
\draw[lightgray] (0,-4) -- (5.19615,-1);
\draw[lightgray] (0,-4) -- (5.19615,-7);
\draw[lightgray] (4.33013,-0.5) -- (5.19615,-1);
\draw[lightgray] (4.33013,-0.5) -- (4.33013,-6.5);
\draw[lightgray] (3.4641,-1) -- (3.4641,-7);
\draw[lightgray] (0,-5) -- (5.19615,-2);
\draw[lightgray] (0,-5) -- (3.4641,-7);
\draw[lightgray] (0.866025,-0.5) -- (5.19615,-3);
\draw[lightgray] (0.866025,-0.5) -- (0.866025,-6.5);
\draw[lightgray] (0,-1) -- (0.866025,-0.5);
\draw[lightgray] (0,-1) -- (5.19615,-4);
\draw[lightgray] (0,-1) -- (0,-7);
\draw[lightgray] (1.73205,-7) -- (5.19615,-5);
\draw[lightgray] (5.19615,-1) -- (5.19615,-7);
\draw[lightgray] (0,-6) -- (5.19615,-3);
\draw[lightgray] (0,-6) -- (1.73205,-7);
\draw[black, line width=0.4mm, fill=white] (1.73205,-6) circle (0.288);
\draw[black, line width=0.4mm, fill=white] (2.59808,-5.5) circle (0.288);
\draw[black, line width=0.4mm, fill=white] (3.4641,-2) circle (0.288);
\node[align=left] at (3.11769,-1.4) {\small $p$};
\draw[black, line width=0.4mm, fill=white] (3.4641,-3) circle (0.288);
\draw[black, line width=0.4mm, fill=white] (3.4641,-4) circle (0.288);
\draw[black, line width=0.4mm, fill=white] (3.4641,-5) circle (0.288);
\draw[black, line width=0.4mm, fill=white] (4.33013,-1.5) circle (0.288);
\draw[black, line width=0.4mm, fill=white] (4.33013,-2.5) circle (0.288);
\draw[black, line width=0.5mm] (1.99186,-5.85) -- (2.33827,-5.65);
\draw[black, line width=0.5mm] (3.4641,-4.7) -- (3.4641,-4.3);
\draw[black, line width=0.5mm] (3.4641,-3.7) -- (3.4641,-3.3);
\draw[black, line width=0.5mm] (3.4641,-2.7) -- (3.4641,-2.3);
\draw[black, line width=0.5mm] (3.72391,-2.85) -- (4.07032,-2.65);
\draw[black, line width=0.5mm] (3.72391,-1.85) -- (4.07032,-1.65);
\draw[black, line width=0.5mm] (3.72391,-2.15) -- (4.07032,-2.35);
\draw[black, line width=0.5mm] (4.33013,-2.2) -- (4.33013,-1.8);
\draw[black, line width=0.5mm] (2.85788,-5.35) -- (3.20429,-5.15);
\draw[black!20!red,-{Stealth[length=1.6mm,width=2.5mm]},line width=0.7mm] (0.866025,-5.81) -- (0.866025,-6.5);
\draw[black!20!red,-{Stealth[length=1.6mm,width=2.5mm]},line width=0.7mm] (1.46358,-5.155) -- (0.866025,-5.5);
\draw[black!20!red,-{Stealth[length=1.6mm,width=2.5mm]},line width=0.7mm] (2.59808,-2.81) -- (2.59808,-3.5);
\draw[black!20!red,-{Stealth[length=1.6mm,width=2.5mm]},line width=0.7mm] (2.59808,-3.81) -- (2.59808,-4.5);
\draw[black!20!red,-{Stealth[length=1.6mm,width=2.5mm]},line width=0.7mm] (2.32961,-4.655) -- (1.73205,-5);
\draw[black!20!red, line width=0.48mm, ] (3.4641,-2) circle (0.408);
\draw[black!20!red,-{Stealth[length=1.6mm,width=2.5mm]},line width=0.7mm] (3.19563,-2.155) -- (2.59808,-2.5);
    \end{tikzpicture}
    \end{center}
    \caption{$p$ is non-shiftable.}
    \label{fig:non_shiftable}
  \end{subfigure}%
\caption{Figures~\ref{fig:shiftable_1}, \ref{fig:shiftable_2} and \ref{fig:shiftable_3} illustrate different cases for the ``shift'' operation. In all of these images, sites $p_1$ and $p_2$ are shiftable agents, while $p_3$ is not.}
\end{figure*}

\vskip.1in
\noindent \underbar{\bf The comb operation:}
%\begin{definition}[Comb]
%\label{proc:comb}
The comb procedure (applied to a combable position $(\ell,d)$) has two phases, line formation and line merging.
After the line formation phase, the first two conditions for $(\ell,d)$ to be combed will be satisfied by the configuration (Figure~\ref{fig:line_merging_before}). The line merging phase gives us the third condition (Figure~\ref{fig:line_merging_after}).

\vskip.1in
\noindent \underbar{\textbf{Line formation}}
Let $L$ denote the set of agents on lane $\ell$ on or below $(\ell,d)$. The agents in $L$ can be grouped into connected components within $L$.
The line formation phase operates from top to bottom on $L$, removing the topmost agent of each component with size greater than $1$ at each turn, until every component on $L$ has size $1$. We maintain the invariant that $(\ell+1,d+1)$ is combed after each turn, while reducing the number of agents in $L$ by $1$.

We call a site ``shiftable'' if there is an agent on the site, and it has exactly two neighboring agents, one directly below and one directly up-right of it. In a turn, there are two possible cases. Denote by $p$ the topmost agent of the topmost component with size greater than $1$.

% FIGURES WOULD BE REAAALLY USEFUL HERE
If $p$ is shiftable, we apply what we call a ``shift'', which moves a set of agents on a line either down-right or up-right, so that $p$ either becomes unoccupied or non-shiftable. To apply a shift, we consider the sequence of sites $p = p_1, p_2, \dots$, where each site $p_{i+1}$ is exactly two steps down-right of site $p_i$. Let $k$ be the largest integer such that all of the sites $p_1, p_2, \dots, p_k$ are shiftable. Figures~\ref{fig:shiftable_1}, \ref{fig:shiftable_2} and \ref{fig:shiftable_3} illustrate examples where $k=2$.
Consider the first non-shiftable site $p_{k+1}$ in the sequence, and the sites directly above and directly down-left of $p_{k+1}$, which we will call $p_{k+1}^{U}$ and $p_{k+1}^{DL}$ respectively. If $p_{k+1}$ is unoccupied or either of $p_{k+1}^{U}$ or $p_{k+1}^{DL}$ are occupied, then moving $p_{k}$ one step down-right is a valid move (Figures~\ref{fig:shiftable_1}, \ref{fig:shiftable_2}). We can thus go backwards through the sequence from $p_{k}$ to $p_1$, moving each agent one step down-right, ending with shifting $p = p_1$ one step down-right, so that the site $p$ originally was occupying now becomes unoccupied.
On the other hand, if $p_{k+1}$ is occupied but $p_{k+1}^{U}$ and $p_{k+1}^{DL}$ aren't, as $p_{k+1}$ is non-shiftable, the remaining three neighbors (down, down-right and up-right) must form a single component, meaning moving $p_{k+1}$ one step up-left is a valid move (Figure~\ref{fig:shiftable_3}). We can then subsequently move each agent from $p_{k+1}$ to $p_2$ one step up-left, culminating in $p = p_1$ becoming non-shiftable, leading in to the second case which we will describe next. Note that after a shift, the invariant that $(\ell+1,d+1)$ is combed still holds.

% THIS DEFINITELY NEEDS A FIGURE
If $p$ is not shiftable, by the invariant we maintain, the sites above, up-left and down-left of $p$ must be unoccupied, while the site directly below $p$ is occupied. Thus, if the site up-right of $p$ is occupied, so must the site down-right of $p$. As $(\ell+1,d+1)$ is combed, the component on $L$ $p$ belongs to will have no agent down-left of it, except for potentially one line of agents extending from the bottommost agent of the component. The agent at $p$ can thus be moved down-left, down along the component on $L$ $p$ belongs to, then down-left to reach the end of the beforementioned line if it exists, and down once more to join the end of this line (Figure~\ref{fig:non_shiftable}). In all, this reduces the number of agents in $L$ by $1$, while maintaining the invariant that $(\ell+1,d+1)$ is combed.

Thus, after the line formation phase, every component in $L$ will have exactly $1$ agent, while $(\ell+1,d+1)$ remains combed. As the site above $(\ell,d)$ is empty, the first two conditions for $(\ell,d)$ being combed are satisfied. The line merging phase will give us the third condition. Figures~\ref{fig:line_merging_before} and \ref{fig:line_merging_after} illustrate configurations before and after the line merging phase.

\vskip.1in
\noindent \underbar{\textbf{Line merging}}
Let $C$ denote the column of sites directly to the right of $(\ell,d)$. The lines extending down and left in the residual region of $(\ell,d)$ may extend from agents in $C$ that are not the bottommost agents of their respective components. To fix this, the line merging phase processs these lines from the lowest to the highest. To move a line downwards by one step, the agents of the line are shifted down one by one, starting from the rightmost agent of the line and ending with the agent on the end of the line. These moves are always possible as long as there is no line directly below the current line. If there is a line directly below, we merge the current line into the line below by moving the agents one at the time to the end of the line below with a straightforward sequence of moves, starting with the leftmost agent of the current line.
%\end{definition}

\begin{figure}[h]
\begin{subfigure}{.5\linewidth}
  \centering
  \begin{center}
  \begin{tikzpicture}[x=0.4cm,y=0.4cm]
  \draw[lightgray] (0,-7) -- (6.06218,-3.5);
\draw[lightgray] (0,-7) -- (5.19615,-10);
\draw[lightgray] (0,-10) -- (6.06218,-6.5);
\draw[lightgray] (3.4641,0) -- (6.06218,-1.5);
\draw[lightgray] (3.4641,0) -- (3.4641,-10);
\draw[lightgray] (5.19615,-10) -- (6.06218,-9.5);
\draw[lightgray] (0,-4) -- (6.06218,-0.5);
\draw[lightgray] (0,-4) -- (6.06218,-7.5);
\draw[lightgray] (3.4641,-10) -- (6.06218,-8.5);
\draw[lightgray] (5.19615,0) -- (6.06218,-0.5);
\draw[lightgray] (5.19615,0) -- (5.19615,-10);
\draw[lightgray] (0,-1) -- (1.73205,0);
\draw[lightgray] (0,-1) -- (6.06218,-4.5);
\draw[lightgray] (0,-5) -- (6.06218,-1.5);
\draw[lightgray] (0,-5) -- (6.06218,-8.5);
\draw[lightgray] (0.866025,-0.5) -- (0.866025,-9.5);
\draw[lightgray] (0,-8) -- (6.06218,-4.5);
\draw[lightgray] (0,-8) -- (3.4641,-10);
\draw[lightgray] (0,-2) -- (3.4641,0);
\draw[lightgray] (0,-2) -- (6.06218,-5.5);
\draw[lightgray] (2.59808,-0.5) -- (2.59808,-9.5);
\draw[lightgray] (4.33013,-0.5) -- (4.33013,-9.5);
\draw[lightgray] (6.06218,-0.5) -- (6.06218,-9.5);
\draw[lightgray] (0,-6) -- (6.06218,-2.5);
\draw[lightgray] (0,-6) -- (6.06218,-9.5);
\draw[lightgray] (0,-9) -- (6.06218,-5.5);
\draw[lightgray] (0,-9) -- (1.73205,-10);
\draw[lightgray] (0,-3) -- (5.19615,0);
\draw[lightgray] (0,-3) -- (6.06218,-6.5);
\draw[lightgray] (0,0) -- (6.06218,-3.5);
\draw[lightgray] (0,0) -- (0,-10);
\draw[lightgray] (1.73205,-10) -- (6.06218,-7.5);
\draw[lightgray] (1.73205,0) -- (6.06218,-2.5);
\draw[lightgray] (1.73205,0) -- (1.73205,-10);
\draw[black, line width=0.4mm, fill=white] (2.59808,-5.5) circle (0.288);
\draw[black, line width=0.4mm, fill=white] (3.4641,-3) circle (0.288);
\draw[black, line width=0.4mm, fill=white] (3.4641,-5) circle (0.288);
\draw[black, line width=0.4mm, fill=white] (3.4641,-8) circle (0.288);
\draw[line width=0.4mm] (4.58469,-1.24544) -- (4.07557,-1.75456);
\draw[line width=0.4mm] (4.07557,-1.24544) -- (4.58469,-1.75456);
\node[align=left] at (3.81051,-0.8) {\normalsize $(\ell,d)$};
\draw[black, line width=0.4mm, fill=white] (4.33013,-2.5) circle (0.288);
\draw[black, line width=0.4mm, fill=white] (4.33013,-4.5) circle (0.288);
\draw[black, line width=0.4mm, fill=white] (4.33013,-7.5) circle (0.288);
\draw[black, line width=0.4mm, fill=white] (5.19615,-2) circle (0.288);
\draw[black, line width=0.4mm, fill=white] (5.19615,-3) circle (0.288);
\draw[black, line width=0.4mm, fill=white] (5.19615,-4) circle (0.288);
\draw[black, line width=0.4mm, fill=white] (5.19615,-5) circle (0.288);
\draw[black, line width=0.4mm, fill=white] (5.19615,-7) circle (0.288);
\draw[black, line width=0.4mm, fill=white] (5.19615,-8) circle (0.288);
\draw[black, line width=0.5mm] (5.19615,-3.7) -- (5.19615,-3.3);
\draw[black, line width=0.5mm] (4.58993,-4.35) -- (4.93634,-4.15);
\draw[black, line width=0.5mm] (4.58993,-4.65) -- (4.93634,-4.85);
\draw[black, line width=0.5mm] (5.19615,-7.7) -- (5.19615,-7.3);
\draw[black, line width=0.5mm] (3.72391,-7.85) -- (4.07032,-7.65);
\draw[black, line width=0.5mm] (3.72391,-4.85) -- (4.07032,-4.65);
\draw[black, line width=0.5mm] (4.58993,-7.35) -- (4.93634,-7.15);
\draw[black, line width=0.5mm] (4.58993,-7.65) -- (4.93634,-7.85);
\draw[black, line width=0.5mm] (5.19615,-2.7) -- (5.19615,-2.3);
\draw[black, line width=0.5mm] (3.72391,-2.85) -- (4.07032,-2.65);
\draw[black, line width=0.5mm] (4.58993,-2.35) -- (4.93634,-2.15);
\draw[black, line width=0.5mm] (4.58993,-2.65) -- (4.93634,-2.85);
\draw[black, line width=0.5mm] (5.19615,-4.7) -- (5.19615,-4.3);
\draw[black, line width=0.5mm] (2.85788,-5.35) -- (3.20429,-5.15);
\path [pattern=north east lines, pattern color=yellow] (4.33013,-1.5) -- (0,-4) -- (0,-10) -- (3.4641,-10) -- (4.33013,-9.5) -- cycle;
\draw[dashed, line width=0.6mm] (4.33013,-1.5) -- (0,-4);
\draw[dashed, line width=0.6mm] (4.33013,-1.5) -- (4.33013,-9.5);
  \end{tikzpicture}
  \end{center}
  \caption{After line formation, before line merging.}
  \label{fig:line_merging_before}
\end{subfigure}%
\hfill
\begin{subfigure}{.5\linewidth}
  \centering
  \begin{center}
  \begin{tikzpicture}[x=0.4cm,y=0.4cm]
  \draw[lightgray] (0,-7) -- (6.06218,-3.5);
\draw[lightgray] (0,-7) -- (5.19615,-10);
\draw[lightgray] (0,-10) -- (6.06218,-6.5);
\draw[lightgray] (3.4641,0) -- (6.06218,-1.5);
\draw[lightgray] (3.4641,0) -- (3.4641,-10);
\draw[lightgray] (5.19615,-10) -- (6.06218,-9.5);
\draw[lightgray] (0,-4) -- (6.06218,-0.5);
\draw[lightgray] (0,-4) -- (6.06218,-7.5);
\draw[lightgray] (3.4641,-10) -- (6.06218,-8.5);
\draw[lightgray] (5.19615,0) -- (6.06218,-0.5);
\draw[lightgray] (5.19615,0) -- (5.19615,-10);
\draw[lightgray] (0,-1) -- (1.73205,0);
\draw[lightgray] (0,-1) -- (6.06218,-4.5);
\draw[lightgray] (0,-5) -- (6.06218,-1.5);
\draw[lightgray] (0,-5) -- (6.06218,-8.5);
\draw[lightgray] (0.866025,-0.5) -- (0.866025,-9.5);
\draw[lightgray] (0,-8) -- (6.06218,-4.5);
\draw[lightgray] (0,-8) -- (3.4641,-10);
\draw[lightgray] (0,-2) -- (3.4641,0);
\draw[lightgray] (0,-2) -- (6.06218,-5.5);
\draw[lightgray] (2.59808,-0.5) -- (2.59808,-9.5);
\draw[lightgray] (4.33013,-0.5) -- (4.33013,-9.5);
\draw[lightgray] (6.06218,-0.5) -- (6.06218,-9.5);
\draw[lightgray] (0,-6) -- (6.06218,-2.5);
\draw[lightgray] (0,-6) -- (6.06218,-9.5);
\draw[lightgray] (0,-9) -- (6.06218,-5.5);
\draw[lightgray] (0,-9) -- (1.73205,-10);
\draw[lightgray] (0,-3) -- (5.19615,0);
\draw[lightgray] (0,-3) -- (6.06218,-6.5);
\draw[lightgray] (0,0) -- (6.06218,-3.5);
\draw[lightgray] (0,0) -- (0,-10);
\draw[lightgray] (1.73205,-10) -- (6.06218,-7.5);
\draw[lightgray] (1.73205,0) -- (6.06218,-2.5);
\draw[lightgray] (1.73205,0) -- (1.73205,-10);
\draw[black, line width=0.4mm, fill=white] (0.866025,-7.5) circle (0.288);
\draw[black, line width=0.4mm, fill=white] (1.73205,-7) circle (0.288);
\draw[black, line width=0.4mm, fill=white] (2.59808,-6.5) circle (0.288);
\draw[black, line width=0.4mm, fill=white] (3.4641,-6) circle (0.288);
\draw[black, line width=0.4mm, fill=white] (3.4641,-9) circle (0.288);
\draw[line width=0.4mm] (4.58469,-1.24544) -- (4.07557,-1.75456);
\draw[line width=0.4mm] (4.07557,-1.24544) -- (4.58469,-1.75456);
\node[align=left] at (3.81051,-0.8) {\normalsize $(\ell,d)$};
\draw[black, line width=0.4mm, fill=white] (4.33013,-5.5) circle (0.288);
\draw[black, line width=0.4mm, fill=white] (4.33013,-8.5) circle (0.288);
\draw[black, line width=0.4mm, fill=white] (5.19615,-2) circle (0.288);
\draw[black, line width=0.4mm, fill=white] (5.19615,-3) circle (0.288);
\draw[black, line width=0.4mm, fill=white] (5.19615,-4) circle (0.288);
\draw[black, line width=0.4mm, fill=white] (5.19615,-5) circle (0.288);
\draw[black, line width=0.4mm, fill=white] (5.19615,-7) circle (0.288);
\draw[black, line width=0.4mm, fill=white] (5.19615,-8) circle (0.288);
\draw[black, line width=0.5mm] (1.99186,-6.85) -- (2.33827,-6.65);
\draw[black, line width=0.5mm] (5.19615,-3.7) -- (5.19615,-3.3);
\draw[black, line width=0.5mm] (3.72391,-8.85) -- (4.07032,-8.65);
\draw[black, line width=0.5mm] (3.72391,-5.85) -- (4.07032,-5.65);
\draw[black, line width=0.5mm] (4.58993,-5.35) -- (4.93634,-5.15);
\draw[black, line width=0.5mm] (5.19615,-7.7) -- (5.19615,-7.3);
\draw[black, line width=0.5mm] (5.19615,-2.7) -- (5.19615,-2.3);
\draw[black, line width=0.5mm] (5.19615,-4.7) -- (5.19615,-4.3);
\draw[black, line width=0.5mm] (2.85788,-6.35) -- (3.20429,-6.15);
\draw[black, line width=0.5mm] (1.12583,-7.35) -- (1.47224,-7.15);
\draw[black, line width=0.5mm] (4.58993,-8.35) -- (4.93634,-8.15);
\path [pattern=north east lines, pattern color=yellow] (4.33013,-1.5) -- (0,-4) -- (0,-10) -- (3.4641,-10) -- (4.33013,-9.5) -- cycle;
\draw[dashed, line width=0.6mm] (4.33013,-1.5) -- (0,-4);
\draw[dashed, line width=0.6mm] (4.33013,-1.5) -- (4.33013,-9.5);
  \end{tikzpicture}
  \end{center}
  \caption{After line merging.}
  \label{fig:line_merging_after}
\end{subfigure}%
\caption{The results after the line formation and line merging phases when the comb procedure is applied to a combable position $(\ell,d)$.}
\end{figure}
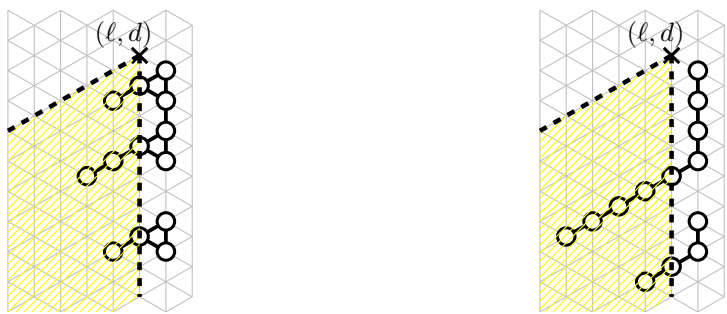

The description of the comb procedure above gives us the following Lemma:

\begin{lemma}
\label{lem:combmakescombed}
After executing a comb operation on a combable position $(\ell,d)$, the position $(\ell,d)$ will be combed (Definition~\ref{defn:combed}).
\end{lemma}

The following Lemma states that combs only affect sites below it.

\begin{figure}[h]
%\centering
%\includegraphics[width=.3\linewidth]{diagrams_irreducibility/untouched_region.png}
\begin{center}
\begin{tikzpicture}[x=0.45cm,y=0.45cm]
\draw[lightgray] (5.19615,-9) -- (10.3923,-6);
\draw[lightgray] (0,-7) -- (10.3923,-1);
\draw[lightgray] (0,-7) -- (3.4641,-9);
\draw[lightgray] (3.4641,0) -- (10.3923,-4);
\draw[lightgray] (3.4641,0) -- (3.4641,-9);
\draw[lightgray] (7.79423,-0.5) -- (7.79423,-8.5);
\draw[lightgray] (0,-4) -- (6.9282,0);
\draw[lightgray] (0,-4) -- (8.66025,-9);
\draw[lightgray] (5.19615,0) -- (10.3923,-3);
\draw[lightgray] (5.19615,0) -- (5.19615,-9);
\draw[lightgray] (9.52628,-0.5) -- (9.52628,-8.5);
\draw[lightgray] (0,-1) -- (1.73205,0);
\draw[lightgray] (0,-1) -- (10.3923,-7);
\draw[lightgray] (6.9282,-9) -- (10.3923,-7);
\draw[lightgray] (8.66025,-9) -- (10.3923,-8);
\draw[lightgray] (0,-5) -- (8.66025,0);
\draw[lightgray] (0,-5) -- (6.9282,-9);
\draw[lightgray] (0.866025,-0.5) -- (0.866025,-8.5);
\draw[lightgray] (0,-8) -- (10.3923,-2);
\draw[lightgray] (0,-8) -- (1.73205,-9);
\draw[lightgray] (0,-2) -- (3.4641,0);
\draw[lightgray] (0,-2) -- (10.3923,-8);
\draw[lightgray] (2.59808,-0.5) -- (2.59808,-8.5);
\draw[lightgray] (6.9282,0) -- (10.3923,-2);
\draw[lightgray] (6.9282,0) -- (6.9282,-9);
\draw[lightgray] (4.33013,-0.5) -- (4.33013,-8.5);
\draw[lightgray] (8.66025,0) -- (10.3923,-1);
\draw[lightgray] (8.66025,0) -- (8.66025,-9);
\draw[lightgray] (6.06218,-0.5) -- (6.06218,-8.5);
\draw[lightgray] (10.3923,0) -- (10.3923,-9);
\draw[lightgray] (0,-6) -- (10.3923,0);
\draw[lightgray] (0,-6) -- (5.19615,-9);
\draw[lightgray] (0,-9) -- (10.3923,-3);
\draw[lightgray] (0,-3) -- (5.19615,0);
\draw[lightgray] (0,-3) -- (10.3923,-9);
\draw[lightgray] (3.4641,-9) -- (10.3923,-5);
\draw[lightgray] (1.73205,-9) -- (10.3923,-4);
\draw[lightgray] (0,0) -- (10.3923,-6);
\draw[lightgray] (0,0) -- (0,-9);
\draw[lightgray] (1.73205,0) -- (10.3923,-5);
\draw[lightgray] (1.73205,0) -- (1.73205,-9);
\draw[black!40!green, dotted, line width=0.5mm] (0,-9) -- (10.3923,-3);
\draw[black!40!green, dotted, line width=0.5mm] (1.73205,0) -- (10.3923,-5);
\draw[black!40!green, dotted, line width=0.5mm] (8.66025,0) -- (8.66025,-9);
\draw[line width=0.4mm] (3.71866,-2.74544) -- (3.20954,-3.25456);
\draw[line width=0.4mm] (3.20954,-2.74544) -- (3.71866,-3.25456);
\node[align=left] at (4.24352,-3.6) {\normalsize $(\ell,d)$};
\draw[fill=black] (6.9282,-7) circle (0.288);
\node[align=left] at (6.9282,-6.3) {\normalsize $R_{(\ell_2,d_2)}$};
\draw[black, line width=0.4mm, fill=white] (8.66025,-4) circle (0.336);
\draw[black, line width=0.32mm] (8.66025,-4) circle (0.24);
\node[align=left] at (8.66025,-4) {\scriptsize $f$};
\path [pattern=north east lines, pattern color=blue] (0,-4) -- (3.4641,-2) -- (10.3923,-6) -- (10.3923,0) -- (0,0) -- cycle;
\draw[dashed, line width=0.6mm] (3.4641,-2) -- (0,-4);
\draw[dashed, line width=0.6mm] (3.4641,-2) -- (10.3923,-6);
\draw[dashed, line width=0.6mm] (3.4641,-3) -- (0,-5);
\draw[dashed, line width=0.6mm] (3.4641,-3) -- (3.4641,-9);
\path [pattern=north east lines, pattern color=red] (3.4641,-9) -- (6.9282,-7) -- (10.3923,-9) -- cycle;
\draw[dashed, line width=0.6mm] (6.9282,-7) -- (3.4641,-9);
\draw[dashed, line width=0.6mm] (6.9282,-7) -- (10.3923,-9);
\end{tikzpicture}
\end{center}
\caption{Illustration of the unaffected region above the position to be combed $(\ell,d)$ from Lemma~\ref{lem:aboveunaffected}, and an unenterable region $R_{(\ell_2,d_2)}$ from Lemma~\ref{lem:unenterableregion} corresponding to some position $(\ell_2,d_2)$ strictly to the right of $(\ell,d)$. Both of these regions include the boundaries drawn in the Figure.}
\label{fig:untouched_region}
\end{figure}
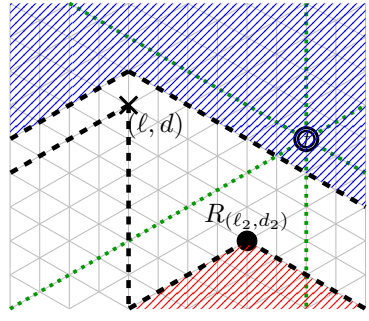

\begin{lemma}[Unaffected Region Above]
\label{lem:aboveunaffected}
Consider the two half-lines extending down-left and down-right from a combable position $(\ell,d)$ as in Figure~\ref{fig:untouched_region}. A comb operation on $(\ell,d)$ will not affect (will not move any agent into or out of) any site above these lines, not including the lines themselves.
In addition, if the site $(\ell-1,d)$ (one step directly down-right) is occupied, no site on the half-line going down-right from $(\ell-1,d)$, including $(\ell-1,d)$ itself, will be affected either.
\end{lemma}

\begin{proof}
We observe that in the comb procedure, aside from the ``shift'' moves, all of the moves occur only within the residual region of $(\ell,d)$. The shift moves only go down-right or up-left, and if the shift originates from some position $p$, the shift does not move any agent up-left from $p$. Hence, the shifts also do not affect any site above the two half-lines described in the Lemma.

We observe that the only part of the procedure that can affect agents on the half-line going down-right from $(\ell-1,d)$ is a potential shift move on an agent on position $(\ell,d)$. However, if $(\ell-1,d)$ is occupied, $(\ell,d)$ will not be shiftable, and so this shift move will not occur.
\end{proof}

To apply a sequence of comb operations to ``push'' agents down towards the \targetspine{}, we only need to find a (not necessarily straight) ``line'' of vacant positions. The following definition and lemma makes this more formal.

\begin{lemma}[Combable Sequence]
Consider a sequence of pairs $((x_1,y_1),(x_2,y_2),\dots,(x_k,y_k))$, where each item in the sequence represents a $(lane,depth)$ pair. We call this a \emph{combable sequence} if:
\begin{enumerate}
\item $x_1$ vertically coincides with the leftmost agent of the configuration.
\item $x_{i+1} = x_i-1$ for all $i \in \{1,2,\dots,k-1\}$ and $x_k > 0$.
\item $y_i \geq 0$ for all $i \in \{1,2,\dots,k\}$.
\item $y_{i+1} \in \{y_i,y_i-1\}$ for all $i \in \{2,\dots,k\}$.
\item The locations $(x_i,y_i-1)$ are all vacant.
\item For each $i \in \{1,2,\dots,k\}$, if $y_i=0$, then the position $(x_i-1,0)$ must be occupied by an agent.
\end{enumerate}
An example of such a sequence is illustrated in Figure~\ref{fig:spine_comb}.
\end{lemma}

\begin{definition}[Combing a Sequence]
\label{proc:combingasequence}
Consider a combable sequence $((x_1,y_1),(x_2,y_2),\dots,(x_k,y_k))$. Combing this sequence refers to combing each pair $(x_i,y_i)$ in succession. The following Lemma justifies that this is always possible.
\end{definition}

\begin{lemma}[Combability of Each Step in a Sequence]
When combing a combable sequence $((x_1,y_1),(x_2,y_2),\dots,(x_k,y_k))$ as described in Definition~\ref{proc:combingasequence}, when $(x_i,y_i)$ is the next position to be combed, $(x_i,y_i)$ will be combable.
\end{lemma}

\begin{proof}
We first note that for any $i$,
by the definition of a combable sequence,
the location directly above $(x_i,y_i)$ must be empty,
and if $y_i = 0$, then $(x_i-1,0)$ is occupied.
These two conditions continue to be true even as combs $1,\dots,i-1$ are executed, as by Lemma~\ref{lem:aboveunaffected}, none of these prior combs will affect $(x_i,y_i-1)$ or $(x_i-1,y_i)$.
This covers two of the conditions necessary for $(x_i,y_i)$ to be combable.

When $i=1$, $(x_1,y_1)$ is clearly combable as there are no agents to the left of $x_1$, and the site directly above $(x_1,y_1)$ is empty.

For $i \geq 2$, as $(x_{i-1},y_{i-1})$ was combed in the previous step, $(x_{i-1},y_{i-1})$ is combed (Lemma~\ref{lem:combmakescombed}). There are two cases for $y_i$.
If $y_i = y_{i-1}-1$, then $(x_i+1,y_i+1) = (x_{i-1},y_{i-1})$ is combed.
On the other hand, if $y_i = y_{i-1}$, then $(x_i+1,y_i) = (x_{i-1},y_{i-1})$ is combed. As $(x_i,y_i-1)$ is unoccupied, every location starting from $(x_i+1,y_i)$ extending down-left must also be unoccupied, which implies $(x_i+1,y_i+1)$ is combed. This gives us our final condition, so $(x_i,y_i)$ is combable.
\end{proof}

% Properties of the Comb Operation
%\subsection{Properties of the comb operation}
In addition to the two properties of the comb operation given as Lemmas~\ref{lem:combmakescombed} and \ref{lem:aboveunaffected}, we state and show a few more properties of the comb operation that we will use later in the proof.

\begin{lemma}[Combing and \Spine{} Lengths]
\label{lem:combingspinelength}
After executing a comb on some position $(\ell,d)$ where $d < \ell$, the length of the \spine{} going down-left will be at most $\ell-1$.
\end{lemma}

\begin{proof}
Let $S$ denote the \spine{} going down-left. The coordinates $(\ell,\ell)$ denotes the position on the \spine{} $S$ vertically below $(\ell,d)$. If $(\ell,\ell)$ is empty after the comb, then every site down-left of $(\ell,\ell)$ is also empty, so the \spine{} $S$ has length at most $\ell-1$. If $(\ell,\ell)$ is not empty after the comb, $(\ell,d)$ being combed ensures that $(\ell,\ell)$ and every agent on the \spine{} $S$ down-left of $(\ell,\ell)$ are tail agents, so \spine{} $S$ has length at most $\ell-1$.
\end{proof}

\begin{lemma}[Preservation of the Rightmost extent]
\label{lem:rightmostextent}
Let $\ell$ denote the lane (x-coordinate) of the rightmost agent of a configuration. After a comb operation is applied, the lanes to the right of $\ell$ (sites with lane less than $\ell$) will continue to be empty.
\end{lemma}

\begin{proof}
Consider a comb applied to some position $(\ell^*,d)$. As we enforce that $\ell^* > 0$ for a comb, this position is necessarily to the left of the immobile agent, while the rightmost agent of the configuration must be either on the same lane as the immobile agent or further right.
All moves aside from the ``shift'' moves in a comb procedure of a position $(\ell^*,d)$ operate only within the residual region of $(\ell^*,d)$, and so will not affect any site on lane $\ell$ or further right.

In the shift moves, an agent is only moved rightward (down-right) if it is shiftable. A shiftable agent must have an agent directly up-right of it, so a shift move cannot move an agent rightward of $\ell$.
\end{proof}

\begin{lemma}[Unenterable Region Below]
\label{lem:unenterableregion}
Consider a position $(\ell,d)$ and the diagonal half-lines extending down-left and down-right from $(\ell,d)$. Consider the region $R_{\ell,d}$ containing every location on or below these lines (Figure~\ref{fig:untouched_region}).

If the region $R_{\ell,d}$ is unoccupied, if a comb operation is applied on a lane strictly to the left of lane $\ell$, $R_{\ell,d}$ will continue to be unoccupied after the comb.
\end{lemma}

\begin{proof}
For the shift movements in the line formation phase, an agent is only moved down-right if it is shiftable, which means it must have an agent directly below it. Thus, $R_{\ell,d}$ cannot be entered from the left side (left of lane $\ell$) by this movement unless there is already an agent in $R_{\ell,d}$.
In addition, as the shift movements only move agents down-right or up-left, it cannot move agents into $R_{\ell,d}$ from the right side (right of lane $\ell$).

For the other movements in the comb operation, we only need to consider possibly entry into $R_{\ell,d}$ from the left side, as these movements occur only within the residual region of the comb, which is strictly to the left of $(\ell,d)$.

In the line formation phase, only down-left and downward movements are used. Down-left movements cannot enter $R_{\ell,d}$, and downward movements only occur when there is an agent directly down-right of the agent to be moved.

In the line merging phase, we simply need to consider the end state of the comb. The end state consists of straight lines stretching down-left from the lane (column) $\ell'$ one lane right of the lane to be combed. All of the agents in this lane are above $R_{\ell,d}$ after the line formation phase, and as $\ell' \leq \ell$, all lines stretching down-left from these agents will also be above $R_{\ell,d}$.
\end{proof}

% M
\subsection{Using combs to show ergodicity}
Our objective is to show that from any (connected) configuration, there exists a sequence of valid moves to transform the configuration into a straight line with the immobile agent at one end. As valid moves cannot introduce holes into a configuration, and all valid moves between hole-free configurations are reversible, this would imply that one can transform any connected configuration of agents into any hole-free configuration of agents using only valid moves.  We proceed by showing that we can always reduce the minimum spine length of any configuration with a series of moves to reach a straight line of agents.

%\subsubsection{Reducing the Minimum \Spine{} Length}
\begin{lemma}
\label{lem:reducespinelength}
If the minimum \spine{} length of a configuration is at least $1$, there exists a sequence of moves to reduce the minimum \spine{} length of the configuration.
\end{lemma}

To reduce the minimum \spine{} length of the configuration, we execute \spinecombs{} in a specific order. Pick a \spine{} of minimum length and denote it as $S_0$. We apply \spinecombs{} in a counterclockwise order, from $S_0$ to $S_1$, followed by $S_1$ to $S_2$, and so on. When applying comb a comb operation from a \sourcespine{} $S_i$ to a \targetspine{} $S_{i+1}$, as always, for ease of analysis, we will treat $S_i$ and $S_{i+1}$ as the \spines{} going in the up-left and down-left directions from the immobile agent respectively.

\begin{definition}[\SpineComb{}]
\label{proc:spinecomb}
%A \spinecomb{} can either be one or two sequences of combs.
Let $r$ denote the length of the \sourcespine{} $S_0$. Let $r_t$ denote the distance of the furthest out agent on the \sourcespine{} from the immobile agent (hence the agents of distances $r+1,\dots,r_t$ are the tail agents).

A \spinecomb{} applies a comb on the combable sequence $((x_1,y_1),(x_2,y_2),\dots,(x_k,y_k))$, where $x_1$ vertically coincides with the leftmost agent of the configuration, $x_i = x_1-i+1$ for each $i \in \{2,3,\dots,k\}$, $x_k = r+1$, $y_i = 1$ whenever $x_i > r_t$, and $y_i = 0$ when $r+1 \leq x_i \leq r_t$. From the definition of tail agents one can easily verify that this $(x_i,y_i-1)$ is vacant for each $i \in \{1,2,\dots,k\}$. Figure~\ref{fig:comb_before_after} illustrates configurations before and after a \spinecomb{} is applied.
\end{definition}

\begin{lemma}
\label{lem:spinecomb}
Consider a \spinecomb{} from a \sourcespine{} of length $r$.
After the \spinecomb{}, the position $(r+1,1)$ will be combed.
Also, the region between (and including) the two half-lines extending up-left and down-left indefinitely from the position $(r+1,0)$ will be empty.
\end{lemma}

\begin{proof}
The last comb operation is on position $(r+1,1)$ or $(r+1,0)$. If it is the former, $(r+1,1)$ will be combed (Lemma~\ref{lem:combmakescombed}). If it is on the latter, as none of the comb operations will affect the site $(r,0)$ directly down-right of the last comb position, $(r,0)$ will remain occupied by an agent. As $(r+1,0)$ is combed, there will be no agents on the diagonal stretching down-left from $(r+1,0)$. Hence, $(r+1,1)$ is also combed.

The region described in the lemma can be divided into ``files'', diagonal lines going down-left. Consider any site $(x,y)$ in this region. If $(x,y)$ is in the lowest file of the region (on the diagonal extending down-left from $(r+1,0)$), as $(r+1,1)$ is combed, $(x,y)$ must be unoccupied. Otherwise, if $(x,y+1)$ is in the same file as some position in the combable sequence, let $(x_i,y_i)$ be the last position in the sequence in the same file as $(x,y+1)$. After $(x_i,y_i)$ is combed, $(x,y)$ must be empty. Subsequent combs will not affect $(x,y)$ by Lemma~\ref{lem:aboveunaffected}. If $(x,y+1)$ is not in the same file as any position in the combable sequence, $(x,y)$ must be empty as $(x_1,y_1)$ is vertically aligned with the leftmost agent of the configuration. Similarly by Lemma~\ref{lem:aboveunaffected}, none of the combs will move an agent into $(x,y)$. Hence, $(x,y)$ will be empty after the \spinecomb{} in all cases.
\end{proof}

After a \spinecomb{} is applied, there are two possible cases, having a gap in the line (defined next), and not having a gap in the line. If there is a gap in the line, we show that we can directly reduce the minimum \spine{} length of the configuration from here, giving us the result of Lemma~\ref{lem:reducespinelength}. Hence, we can proceed with the rest of the proof of Lemma~\ref{lem:reducespinelength} assuming that there will never be a gap in the line.

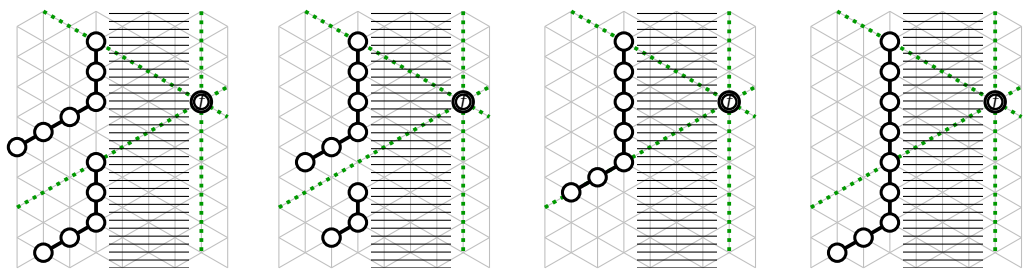
\begin{figure*}
\begin{subfigure}{.25\textwidth}
  %\centering
  %\includegraphics[width=.7\linewidth]{diagrams_irreducibility/gap_yes1.png}
  \begin{center}
  \begin{tikzpicture}[x=0.4cm,y=0.4cm]
  \draw[lightgray] (5.19615,-9) -- (6.9282,-8);
\draw[lightgray] (0,-7) -- (6.9282,-3);
\draw[lightgray] (0,-7) -- (3.4641,-9);
\draw[lightgray] (0,-4) -- (6.06218,-0.5);
\draw[lightgray] (0,-4) -- (6.9282,-8);
\draw[lightgray] (0,-1) -- (0.866025,-0.5);
\draw[lightgray] (0,-1) -- (6.9282,-5);
\draw[lightgray] (0,-1) -- (0,-9);
\draw[lightgray] (1.73205,-1) -- (1.73205,-9);
\draw[lightgray] (3.4641,-1) -- (3.4641,-9);
\draw[lightgray] (0,-5) -- (6.9282,-1);
\draw[lightgray] (0,-5) -- (6.9282,-9);
\draw[lightgray] (0.866025,-0.5) -- (6.9282,-4);
\draw[lightgray] (0.866025,-0.5) -- (0.866025,-8.5);
\draw[lightgray] (0,-8) -- (6.9282,-4);
\draw[lightgray] (0,-8) -- (1.73205,-9);
\draw[lightgray] (5.19615,-1) -- (5.19615,-9);
\draw[lightgray] (0,-2) -- (2.59808,-0.5);
\draw[lightgray] (0,-2) -- (6.9282,-6);
\draw[lightgray] (2.59808,-0.5) -- (6.9282,-3);
\draw[lightgray] (2.59808,-0.5) -- (2.59808,-8.5);
\draw[lightgray] (4.33013,-0.5) -- (6.9282,-2);
\draw[lightgray] (4.33013,-0.5) -- (4.33013,-8.5);
\draw[lightgray] (6.06218,-0.5) -- (6.9282,-1);
\draw[lightgray] (6.06218,-0.5) -- (6.06218,-8.5);
\draw[lightgray] (0,-6) -- (6.9282,-2);
\draw[lightgray] (0,-6) -- (5.19615,-9);
\draw[lightgray] (0,-9) -- (6.9282,-5);
\draw[lightgray] (0,-3) -- (4.33013,-0.5);
\draw[lightgray] (0,-3) -- (6.9282,-7);
\draw[lightgray] (3.4641,-9) -- (6.9282,-7);
\draw[lightgray] (1.73205,-9) -- (6.9282,-6);
\draw[lightgray] (6.9282,-1) -- (6.9282,-9);
\draw[black!40!green, dotted, line width=0.5mm] (0,-7) -- (6.9282,-3);
\draw[black!40!green, dotted, line width=0.5mm] (0.866025,-0.5) -- (6.9282,-4);
\draw[black!40!green, dotted, line width=0.5mm] (6.06218,-0.5) -- (6.06218,-8.5);
\draw[black, line width=0.4mm, fill=white] (0,-5) circle (0.288);
\draw[black, line width=0.4mm, fill=white] (0.866025,-4.5) circle (0.288);
\draw[black, line width=0.4mm, fill=white] (0.866025,-8.5) circle (0.288);
\draw[black, line width=0.4mm, fill=white] (1.73205,-4) circle (0.288);
\draw[black, line width=0.4mm, fill=white] (1.73205,-8) circle (0.288);
\draw[black, line width=0.4mm, fill=white] (2.59808,-1.5) circle (0.288);
\draw[black, line width=0.4mm, fill=white] (2.59808,-2.5) circle (0.288);
\draw[black, line width=0.4mm, fill=white] (2.59808,-3.5) circle (0.288);
\draw[black, line width=0.4mm, fill=white] (2.59808,-5.5) circle (0.288);
\draw[black, line width=0.4mm, fill=white] (2.59808,-6.5) circle (0.288);
\draw[black, line width=0.4mm, fill=white] (2.59808,-7.5) circle (0.288);
\draw[black, line width=0.4mm, fill=white] (6.06218,-3.5) circle (0.336);
\draw[black, line width=0.32mm] (6.06218,-3.5) circle (0.24);
\node[align=left] at (6.06218,-3.5) {\scriptsize $f$};
\draw[black, line width=0.5mm] (2.59808,-2.2) -- (2.59808,-1.8);
\draw[black, line width=0.5mm] (2.59808,-3.2) -- (2.59808,-2.8);
\draw[black, line width=0.5mm] (1.99186,-7.85) -- (2.33827,-7.65);
\draw[black, line width=0.5mm] (1.12583,-4.35) -- (1.47224,-4.15);
\draw[black, line width=0.5mm] (1.12583,-8.35) -- (1.47224,-8.15);
\draw[black, line width=0.5mm] (0.259808,-4.85) -- (0.606218,-4.65);
\draw[black, line width=0.5mm] (2.59808,-6.2) -- (2.59808,-5.8);
\draw[black, line width=0.5mm] (2.59808,-7.2) -- (2.59808,-6.8);
\draw[black, line width=0.5mm] (1.99186,-3.85) -- (2.33827,-3.65);
\path [pattern=horizontal lines, pattern color=black] (3.03109,-0.5) -- (5.62917,-0.5) -- (5.62917,-9) -- (3.03109,-9) -- cycle;
  \end{tikzpicture}
  \end{center}
  \caption{Gap between \spines{}.}
  \label{fig:gap_yes1}
\end{subfigure}%
\begin{subfigure}{.25\textwidth}
  %\centering
  %\includegraphics[width=.7\linewidth]{diagrams_irreducibility/gap_yes2.png}
  \begin{center}
  \begin{tikzpicture}[x=0.4cm,y=0.4cm]
  \draw[lightgray] (5.19615,-9) -- (6.9282,-8);
\draw[lightgray] (0,-7) -- (6.9282,-3);
\draw[lightgray] (0,-7) -- (3.4641,-9);
\draw[lightgray] (0,-4) -- (6.06218,-0.5);
\draw[lightgray] (0,-4) -- (6.9282,-8);
\draw[lightgray] (0,-1) -- (0.866025,-0.5);
\draw[lightgray] (0,-1) -- (6.9282,-5);
\draw[lightgray] (0,-1) -- (0,-9);
\draw[lightgray] (1.73205,-1) -- (1.73205,-9);
\draw[lightgray] (3.4641,-1) -- (3.4641,-9);
\draw[lightgray] (0,-5) -- (6.9282,-1);
\draw[lightgray] (0,-5) -- (6.9282,-9);
\draw[lightgray] (0.866025,-0.5) -- (6.9282,-4);
\draw[lightgray] (0.866025,-0.5) -- (0.866025,-8.5);
\draw[lightgray] (0,-8) -- (6.9282,-4);
\draw[lightgray] (0,-8) -- (1.73205,-9);
\draw[lightgray] (5.19615,-1) -- (5.19615,-9);
\draw[lightgray] (0,-2) -- (2.59808,-0.5);
\draw[lightgray] (0,-2) -- (6.9282,-6);
\draw[lightgray] (2.59808,-0.5) -- (6.9282,-3);
\draw[lightgray] (2.59808,-0.5) -- (2.59808,-8.5);
\draw[lightgray] (4.33013,-0.5) -- (6.9282,-2);
\draw[lightgray] (4.33013,-0.5) -- (4.33013,-8.5);
\draw[lightgray] (6.06218,-0.5) -- (6.9282,-1);
\draw[lightgray] (6.06218,-0.5) -- (6.06218,-8.5);
\draw[lightgray] (0,-6) -- (6.9282,-2);
\draw[lightgray] (0,-6) -- (5.19615,-9);
\draw[lightgray] (0,-9) -- (6.9282,-5);
\draw[lightgray] (0,-3) -- (4.33013,-0.5);
\draw[lightgray] (0,-3) -- (6.9282,-7);
\draw[lightgray] (3.4641,-9) -- (6.9282,-7);
\draw[lightgray] (1.73205,-9) -- (6.9282,-6);
\draw[lightgray] (6.9282,-1) -- (6.9282,-9);
\draw[black!40!green, dotted, line width=0.5mm] (0,-7) -- (6.9282,-3);
\draw[black!40!green, dotted, line width=0.5mm] (0.866025,-0.5) -- (6.9282,-4);
\draw[black!40!green, dotted, line width=0.5mm] (6.06218,-0.5) -- (6.06218,-8.5);
\draw[black, line width=0.4mm, fill=white] (0.866025,-5.5) circle (0.288);
\draw[black, line width=0.4mm, fill=white] (1.73205,-5) circle (0.288);
\draw[black, line width=0.4mm, fill=white] (1.73205,-8) circle (0.288);
\draw[black, line width=0.4mm, fill=white] (2.59808,-1.5) circle (0.288);
\draw[black, line width=0.4mm, fill=white] (2.59808,-2.5) circle (0.288);
\draw[black, line width=0.4mm, fill=white] (2.59808,-3.5) circle (0.288);
\draw[black, line width=0.4mm, fill=white] (2.59808,-4.5) circle (0.288);
\draw[black, line width=0.4mm, fill=white] (2.59808,-6.5) circle (0.288);
\draw[black, line width=0.4mm, fill=white] (2.59808,-7.5) circle (0.288);
\draw[black, line width=0.4mm, fill=white] (6.06218,-3.5) circle (0.336);
\draw[black, line width=0.32mm] (6.06218,-3.5) circle (0.24);
\node[align=left] at (6.06218,-3.5) {\scriptsize $f$};
\draw[black, line width=0.5mm] (2.59808,-2.2) -- (2.59808,-1.8);
\draw[black, line width=0.5mm] (2.59808,-3.2) -- (2.59808,-2.8);
\draw[black, line width=0.5mm] (1.99186,-7.85) -- (2.33827,-7.65);
\draw[black, line width=0.5mm] (1.12583,-5.35) -- (1.47224,-5.15);
\draw[black, line width=0.5mm] (2.59808,-7.2) -- (2.59808,-6.8);
\draw[black, line width=0.5mm] (1.99186,-4.85) -- (2.33827,-4.65);
\draw[black, line width=0.5mm] (2.59808,-4.2) -- (2.59808,-3.8);
\path [pattern=horizontal lines, pattern color=black] (3.03109,-0.5) -- (5.62917,-0.5) -- (5.62917,-9) -- (3.03109,-9) -- cycle;
  \end{tikzpicture}
  \end{center}
  \caption{Gap on \targetspine{}.}
  \label{fig:gap_yes2}
\end{subfigure}%
\begin{subfigure}{.25\textwidth}
  %\centering
  %\includegraphics[width=.7\linewidth]{diagrams_irreducibility/gap_no1.png}
  \begin{center}
  \begin{tikzpicture}[x=0.4cm,y=0.4cm]
  \draw[lightgray] (5.19615,-9) -- (6.9282,-8);
\draw[lightgray] (0,-7) -- (6.9282,-3);
\draw[lightgray] (0,-7) -- (3.4641,-9);
\draw[lightgray] (0,-4) -- (6.06218,-0.5);
\draw[lightgray] (0,-4) -- (6.9282,-8);
\draw[lightgray] (0,-1) -- (0.866025,-0.5);
\draw[lightgray] (0,-1) -- (6.9282,-5);
\draw[lightgray] (0,-1) -- (0,-9);
\draw[lightgray] (1.73205,-1) -- (1.73205,-9);
\draw[lightgray] (3.4641,-1) -- (3.4641,-9);
\draw[lightgray] (0,-5) -- (6.9282,-1);
\draw[lightgray] (0,-5) -- (6.9282,-9);
\draw[lightgray] (0.866025,-0.5) -- (6.9282,-4);
\draw[lightgray] (0.866025,-0.5) -- (0.866025,-8.5);
\draw[lightgray] (0,-8) -- (6.9282,-4);
\draw[lightgray] (0,-8) -- (1.73205,-9);
\draw[lightgray] (5.19615,-1) -- (5.19615,-9);
\draw[lightgray] (0,-2) -- (2.59808,-0.5);
\draw[lightgray] (0,-2) -- (6.9282,-6);
\draw[lightgray] (2.59808,-0.5) -- (6.9282,-3);
\draw[lightgray] (2.59808,-0.5) -- (2.59808,-8.5);
\draw[lightgray] (4.33013,-0.5) -- (6.9282,-2);
\draw[lightgray] (4.33013,-0.5) -- (4.33013,-8.5);
\draw[lightgray] (6.06218,-0.5) -- (6.9282,-1);
\draw[lightgray] (6.06218,-0.5) -- (6.06218,-8.5);
\draw[lightgray] (0,-6) -- (6.9282,-2);
\draw[lightgray] (0,-6) -- (5.19615,-9);
\draw[lightgray] (0,-9) -- (6.9282,-5);
\draw[lightgray] (0,-3) -- (4.33013,-0.5);
\draw[lightgray] (0,-3) -- (6.9282,-7);
\draw[lightgray] (3.4641,-9) -- (6.9282,-7);
\draw[lightgray] (1.73205,-9) -- (6.9282,-6);
\draw[lightgray] (6.9282,-1) -- (6.9282,-9);
\draw[black!40!green, dotted, line width=0.5mm] (0,-7) -- (6.9282,-3);
\draw[black!40!green, dotted, line width=0.5mm] (0.866025,-0.5) -- (6.9282,-4);
\draw[black!40!green, dotted, line width=0.5mm] (6.06218,-0.5) -- (6.06218,-8.5);
\draw[black, line width=0.4mm, fill=white] (0.866025,-6.5) circle (0.288);
\draw[black, line width=0.4mm, fill=white] (1.73205,-6) circle (0.288);
\draw[black, line width=0.4mm, fill=white] (2.59808,-1.5) circle (0.288);
\draw[black, line width=0.4mm, fill=white] (2.59808,-2.5) circle (0.288);
\draw[black, line width=0.4mm, fill=white] (2.59808,-3.5) circle (0.288);
\draw[black, line width=0.4mm, fill=white] (2.59808,-4.5) circle (0.288);
\draw[black, line width=0.4mm, fill=white] (2.59808,-5.5) circle (0.288);
\draw[black, line width=0.4mm, fill=white] (6.06218,-3.5) circle (0.336);
\draw[black, line width=0.32mm] (6.06218,-3.5) circle (0.24);
\node[align=left] at (6.06218,-3.5) {\scriptsize $f$};
\draw[black, line width=0.5mm] (2.59808,-2.2) -- (2.59808,-1.8);
\draw[black, line width=0.5mm] (1.99186,-5.85) -- (2.33827,-5.65);
\draw[black, line width=0.5mm] (2.59808,-3.2) -- (2.59808,-2.8);
\draw[black, line width=0.5mm] (1.12583,-6.35) -- (1.47224,-6.15);
\draw[black, line width=0.5mm] (2.59808,-4.2) -- (2.59808,-3.8);
\draw[black, line width=0.5mm] (2.59808,-5.2) -- (2.59808,-4.8);
\path [pattern=horizontal lines, pattern color=black] (3.03109,-0.5) -- (5.62917,-0.5) -- (5.62917,-9) -- (3.03109,-9) -- cycle;
  \end{tikzpicture}
  \end{center}
  \caption{No gap.}
  \label{fig:gap_no1}
\end{subfigure}%
\begin{subfigure}{.25\textwidth}
  %\centering
  %\includegraphics[width=.7\linewidth]{diagrams_irreducibility/gap_no2.png}
  \begin{center}
  \begin{tikzpicture}[x=0.4cm,y=0.4cm]
  \draw[lightgray] (5.19615,-9) -- (6.9282,-8);
\draw[lightgray] (0,-7) -- (6.9282,-3);
\draw[lightgray] (0,-7) -- (3.4641,-9);
\draw[lightgray] (0,-4) -- (6.06218,-0.5);
\draw[lightgray] (0,-4) -- (6.9282,-8);
\draw[lightgray] (0,-1) -- (0.866025,-0.5);
\draw[lightgray] (0,-1) -- (6.9282,-5);
\draw[lightgray] (0,-1) -- (0,-9);
\draw[lightgray] (1.73205,-1) -- (1.73205,-9);
\draw[lightgray] (3.4641,-1) -- (3.4641,-9);
\draw[lightgray] (0,-5) -- (6.9282,-1);
\draw[lightgray] (0,-5) -- (6.9282,-9);
\draw[lightgray] (0.866025,-0.5) -- (6.9282,-4);
\draw[lightgray] (0.866025,-0.5) -- (0.866025,-8.5);
\draw[lightgray] (0,-8) -- (6.9282,-4);
\draw[lightgray] (0,-8) -- (1.73205,-9);
\draw[lightgray] (5.19615,-1) -- (5.19615,-9);
\draw[lightgray] (0,-2) -- (2.59808,-0.5);
\draw[lightgray] (0,-2) -- (6.9282,-6);
\draw[lightgray] (2.59808,-0.5) -- (6.9282,-3);
\draw[lightgray] (2.59808,-0.5) -- (2.59808,-8.5);
\draw[lightgray] (4.33013,-0.5) -- (6.9282,-2);
\draw[lightgray] (4.33013,-0.5) -- (4.33013,-8.5);
\draw[lightgray] (6.06218,-0.5) -- (6.9282,-1);
\draw[lightgray] (6.06218,-0.5) -- (6.06218,-8.5);
\draw[lightgray] (0,-6) -- (6.9282,-2);
\draw[lightgray] (0,-6) -- (5.19615,-9);
\draw[lightgray] (0,-9) -- (6.9282,-5);
\draw[lightgray] (0,-3) -- (4.33013,-0.5);
\draw[lightgray] (0,-3) -- (6.9282,-7);
\draw[lightgray] (3.4641,-9) -- (6.9282,-7);
\draw[lightgray] (1.73205,-9) -- (6.9282,-6);
\draw[lightgray] (6.9282,-1) -- (6.9282,-9);
\draw[black!40!green, dotted, line width=0.5mm] (0,-7) -- (6.9282,-3);
\draw[black!40!green, dotted, line width=0.5mm] (0.866025,-0.5) -- (6.9282,-4);
\draw[black!40!green, dotted, line width=0.5mm] (6.06218,-0.5) -- (6.06218,-8.5);
\draw[black, line width=0.4mm, fill=white] (0.866025,-8.5) circle (0.288);
\draw[black, line width=0.4mm, fill=white] (1.73205,-8) circle (0.288);
\draw[black, line width=0.4mm, fill=white] (2.59808,-1.5) circle (0.288);
\draw[black, line width=0.4mm, fill=white] (2.59808,-2.5) circle (0.288);
\draw[black, line width=0.4mm, fill=white] (2.59808,-3.5) circle (0.288);
\draw[black, line width=0.4mm, fill=white] (2.59808,-4.5) circle (0.288);
\draw[black, line width=0.4mm, fill=white] (2.59808,-5.5) circle (0.288);
\draw[black, line width=0.4mm, fill=white] (2.59808,-6.5) circle (0.288);
\draw[black, line width=0.4mm, fill=white] (2.59808,-7.5) circle (0.288);
\draw[black, line width=0.4mm, fill=white] (6.06218,-3.5) circle (0.336);
\draw[black, line width=0.32mm] (6.06218,-3.5) circle (0.24);
\node[align=left] at (6.06218,-3.5) {\scriptsize $f$};
\draw[black, line width=0.5mm] (2.59808,-2.2) -- (2.59808,-1.8);
\draw[black, line width=0.5mm] (2.59808,-3.2) -- (2.59808,-2.8);
\draw[black, line width=0.5mm] (1.99186,-7.85) -- (2.33827,-7.65);
\draw[black, line width=0.5mm] (1.12583,-8.35) -- (1.47224,-8.15);
\draw[black, line width=0.5mm] (2.59808,-6.2) -- (2.59808,-5.8);
\draw[black, line width=0.5mm] (2.59808,-7.2) -- (2.59808,-6.8);
\draw[black, line width=0.5mm] (2.59808,-4.2) -- (2.59808,-3.8);
\draw[black, line width=0.5mm] (2.59808,-5.2) -- (2.59808,-4.8);
\path [pattern=horizontal lines, pattern color=black] (3.03109,-0.5) -- (5.62917,-0.5) -- (5.62917,-9) -- (3.03109,-9) -- cycle;
  \end{tikzpicture}
  \end{center}
  \caption{No gap.}
  \label{fig:gap_no2}
\end{subfigure}%
\caption{Illustration of the line between the \spines{} going up-left and down-left from the immobile agent. Figures~\ref{fig:gap_yes1} and \ref{fig:gap_yes2} have gaps in the line (Definition~\ref{defn:gapintheline}), while Figures~\ref{fig:gap_no1} and \ref{fig:gap_no2} do not. Observe that in the cases with no gap, the length of the \targetspine{} matches that of the \sourcespine{}.}
\end{figure*}

\begin{definition}[Gap in the line]
\label{defn:gapintheline}
Let $r$ be the length of the \sourcespine{} $S_i$. Consider the vertical line segment of sites from the location of the \anchoragent{} of $S_i$ down to the site on the \targetspine{} $S_{i+1}$ of distance $r$ from the center, including the two \spine{} location endpoints. If there is a site on this line segment that is unoccupied by agents, we say that there is a \emph{gap} in the line from the \sourcespine{} to the \targetspine{}.
\end{definition}

\begin{lemma}[Reducing minimum \spine{} length using a gap]
\label{lem:reducingusinggap}
After a \spinecomb{} is applied from a \sourcespine{} of minimum length, if there is a gap in the line from the \sourcespine{} to the \targetspine{}, there exists a sequence of moves to reduce the minimum \spine{} length of the configuration. 
\end{lemma}

\begin{proof}
Let $r$ denote the length of the \sourcespine{} $S_i$. Suppose that there is an unoccupied site $(r,d)$ on this line segment.
If the unoccupied site on the line segment is on the \sourcespine{} $S_i$ (which actually never happens), as every site on the \spine{} of distance greater than $r$ will be unoccupied by Lemma~\ref{lem:spinecomb}, the new minimum \spine{} length would be at most $r-1$. If the unoccupied site is on the \targetspine{} $S_{i+1}$, as $(r+1,1)$ is combed, every site of the \targetspine{} of distance greater than this unoccupied site would also be unoccupied. Hence the minimum \spine{} length would also have decreased to at most $r-1$ in this case.

If the unoccupied site $(r,d)$ lies strictly between the \sourcespine{} and the \targetspine{}, we apply one more comb on position $(r,d+1)$. Position $(r,d+1)$ is combable as $(r,d)$ is empty, and $(r+1,1)$ being combed ensures that all sites on the half-line extending down-left from $(r,d)$ are also empty, so $(r+1,d+1)$ is also combed. By Lemma~\ref{lem:combingspinelength}, combing $(r,d+1)$ results in the \targetspine{} having length at most $r-1$.
\end{proof}

Thus, from now on we may assume that whenever a \spinecomb{} is executed from a \sourcespine{} of minimum length, there will be no gaps in the line between the \sourcespine{} and the \targetspine{}.
In addition, we assume that the comb operations do not cause any other \spine{} (in particular the \spine{} going downwards) to end up with a \spine{} length below the current minimum spine length, as in this case we have already achieved the result of Lemma~\ref{lem:reducespinelength}.
The following Lemma shows that such a \spinecomb{} ``pushes'' all of the agents of distance greater than the minimum \spine{} length towards the \targetspine{} ore beyond.

\begin{lemma}[Resulting configuration assuming no gap exists]
\label{lem:nogapresult}
Suppose that a \spinecomb{} is executed from a \sourcespine{} $S_i$ (of minimal length $r$) to a \targetspine{} $S_{i+1}$, and assume that there are no gaps in the line between $S_i$ and $S_{i+1}$.
In the resulting configuration, there will be no agent of distance greater than $r$ strictly between the \sourcespine{} and \targetspine{}, or on the \sourcespine{} itself. Furthermore, the lengths of both the \sourcespine{} and \targetspine{} will now be exactly $r$.
\end{lemma}

\begin{proof}
After the \spinecomb{}, position $(r+1,1)$ will be combed, and the region between the down-left and up-left diagonals extending from $(r+1,0)$ as described in Lemma~\ref{lem:spinecomb} will be empty. As there is no gap in the line from $S_i$ to $S_{i+1}$, the only lines extending left and down in the residual region of $(r+1,1)$ will be on the \targetspine{} or below, giving us the first part of this Lemma.

The length of the \sourcespine{} is $r$ as position $(r,0)$ is occupied while no position on the \sourcespine{} beyond that is. For the length of the \targetspine{}, the position on the \targetspine{} of distance $r$ from the immobile agent is occupied and is not a tail agent, and by Lemma~\ref{lem:reducespinelength}, $(r+1,1)$ being combed implies that the \targetspine{} has length at most $r$.
\end{proof}

As a \spinecomb{} sets the length of the \targetspine{} $S_{i+1}$ to be the same as that of the \sourcespine{} $S_i$, which has minimum length, we can continue executing \spinecombs{} in a counterclockwise fashion, from $S_{i+1}$ to $S_{i+2}$, followed by $S_{i+2}$ to $S_{i+3}$, and so on. We show that after seven of these \spinecombs{} which do not create gaps, we will reach a type of configuration we will call a \emph{hexagon with a tail}. Figure~\ref{fig:hexagon_tail} illustrates examples of these ``hexagon with a tail'' configurations, though one should note that it is not necessary for all sites on the outer hexagon to be filled.

\begin{figure*}[t]
\begin{subfigure}{.33\textwidth}
  %\centering
  %\includegraphics[width=.8\linewidth]{diagrams_irreducibility/hexagon_tail1.png}
  \begin{center}
  \begin{tikzpicture}[x=0.35cm,y=0.35cm]
  \draw[lightgray] (10.3923,-1) -- (10.3923,-10);
\draw[lightgray] (0,-7) -- (11.2583,-0.5);
\draw[lightgray] (0,-7) -- (6.06218,-10.5);
\draw[lightgray] (0,-10) -- (12.1244,-3);
\draw[lightgray] (0,-10) -- (0.866025,-10.5);
\draw[lightgray] (7.79423,-0.5) -- (12.1244,-3);
\draw[lightgray] (7.79423,-0.5) -- (7.79423,-10.5);
\draw[lightgray] (12.1244,-1) -- (12.1244,-10);
\draw[lightgray] (0,-4) -- (6.06218,-0.5);
\draw[lightgray] (0,-4) -- (11.2583,-10.5);
\draw[lightgray] (9.52628,-0.5) -- (12.1244,-2);
\draw[lightgray] (9.52628,-0.5) -- (9.52628,-10.5);
\draw[lightgray] (11.2583,-10.5) -- (12.1244,-10);
\draw[lightgray] (0,-1) -- (0.866025,-0.5);
\draw[lightgray] (0,-1) -- (12.1244,-8);
\draw[lightgray] (0,-1) -- (0,-10);
\draw[lightgray] (11.2583,-0.5) -- (12.1244,-1);
\draw[lightgray] (11.2583,-0.5) -- (11.2583,-10.5);
\draw[lightgray] (1.73205,-1) -- (1.73205,-10);
\draw[lightgray] (9.52628,-10.5) -- (12.1244,-9);
\draw[lightgray] (0.866025,-10.5) -- (12.1244,-4);
\draw[lightgray] (3.4641,-1) -- (3.4641,-10);
\draw[lightgray] (0,-5) -- (7.79423,-0.5);
\draw[lightgray] (0,-5) -- (9.52628,-10.5);
\draw[lightgray] (0.866025,-0.5) -- (12.1244,-7);
\draw[lightgray] (0.866025,-0.5) -- (0.866025,-10.5);
\draw[lightgray] (0,-8) -- (12.1244,-1);
\draw[lightgray] (0,-8) -- (4.33013,-10.5);
\draw[lightgray] (4.33013,-10.5) -- (12.1244,-6);
\draw[lightgray] (5.19615,-1) -- (5.19615,-10);
\draw[lightgray] (0,-2) -- (2.59808,-0.5);
\draw[lightgray] (0,-2) -- (12.1244,-9);
\draw[lightgray] (2.59808,-0.5) -- (12.1244,-6);
\draw[lightgray] (2.59808,-0.5) -- (2.59808,-10.5);
\draw[lightgray] (2.59808,-10.5) -- (12.1244,-5);
\draw[lightgray] (4.33013,-0.5) -- (12.1244,-5);
\draw[lightgray] (4.33013,-0.5) -- (4.33013,-10.5);
\draw[lightgray] (6.06218,-0.5) -- (12.1244,-4);
\draw[lightgray] (6.06218,-0.5) -- (6.06218,-10.5);
\draw[lightgray] (0,-6) -- (9.52628,-0.5);
\draw[lightgray] (0,-6) -- (7.79423,-10.5);
\draw[lightgray] (0,-9) -- (12.1244,-2);
\draw[lightgray] (0,-9) -- (2.59808,-10.5);
\draw[lightgray] (0,-3) -- (4.33013,-0.5);
\draw[lightgray] (0,-3) -- (12.1244,-10);
\draw[lightgray] (6.9282,-1) -- (6.9282,-10);
\draw[lightgray] (6.06218,-10.5) -- (12.1244,-7);
\draw[lightgray] (7.79423,-10.5) -- (12.1244,-8);
\draw[lightgray] (8.66025,-1) -- (8.66025,-10);
\draw[black!40!green, dotted, line width=0.5mm] (0,-10) -- (12.1244,-3);
\draw[black!40!green, dotted, line width=0.5mm] (0,-1) -- (12.1244,-8);
\draw[black!40!green, dotted, line width=0.5mm] (7.79423,-0.5) -- (7.79423,-10.5);
\draw[black, line width=0.4mm, fill=white] (0.866025,-1.5) circle (0.288);
\draw[black, line width=0.4mm, fill=white] (1.73205,-2) circle (0.288);
\draw[black, line width=0.4mm, fill=white] (2.59808,-2.5) circle (0.288);
\draw[black, line width=0.4mm, fill=white] (3.4641,-3) circle (0.288);
\draw[black, line width=0.4mm, fill=white] (4.33013,-3.5) circle (0.288);
\node[align=left] at (4.50333,-2.8) {\footnotesize $p_1$};
\draw[black, line width=0.4mm, fill=white] (4.33013,-4.5) circle (0.288);
\draw[black, line width=0.4mm, fill=white] (4.33013,-5.5) circle (0.288);
\draw[black, line width=0.4mm, fill=white] (4.33013,-6.5) circle (0.288);
\draw[black, line width=0.4mm, fill=white] (4.33013,-7.5) circle (0.288);
\draw[black, line width=0.4mm, fill=white] (5.19615,-3) circle (0.288);
\draw[black, line width=0.4mm, fill=white] (5.19615,-4) circle (0.288);
\draw[black, line width=0.4mm, fill=white] (5.19615,-8) circle (0.288);
\draw[black, line width=0.4mm, fill=white] (6.06218,-2.5) circle (0.288);
\draw[black, line width=0.4mm, fill=white] (6.06218,-4.5) circle (0.288);
\draw[black, line width=0.4mm, fill=white] (6.06218,-8.5) circle (0.288);
\draw[black, line width=0.4mm, fill=white] (6.9282,-2) circle (0.288);
\draw[black, line width=0.4mm, fill=white] (6.9282,-4) circle (0.288);
\draw[black, line width=0.4mm, fill=white] (6.9282,-6) circle (0.288);
\draw[black, line width=0.4mm, fill=white] (6.9282,-7) circle (0.288);
\draw[black, line width=0.4mm, fill=white] (6.9282,-9) circle (0.288);
\draw[black, line width=0.4mm, fill=white] (7.79423,-1.5) circle (0.288);
\draw[black, line width=0.4mm, fill=white] (7.79423,-4.5) circle (0.288);
\draw[black, line width=0.4mm, fill=white] (7.79423,-5.5) circle (0.336);
\draw[black, line width=0.32mm] (7.79423,-5.5) circle (0.24);
\node[align=left] at (7.79423,-5.5) {\scriptsize $f$};
\draw[black, line width=0.4mm, fill=white] (7.79423,-6.5) circle (0.288);
\draw[black, line width=0.4mm, fill=white] (7.79423,-9.5) circle (0.288);
\draw[black, line width=0.4mm, fill=white] (8.66025,-2) circle (0.288);
\draw[black, line width=0.4mm, fill=white] (8.66025,-7) circle (0.288);
\draw[black, line width=0.4mm, fill=white] (8.66025,-9) circle (0.288);
\draw[black, line width=0.4mm, fill=white] (9.52628,-2.5) circle (0.288);
\draw[black, line width=0.4mm, fill=white] (9.52628,-8.5) circle (0.288);
\draw[black, line width=0.4mm, fill=white] (10.3923,-3) circle (0.288);
\draw[black, line width=0.4mm, fill=white] (10.3923,-8) circle (0.288);
\draw[black, line width=0.4mm, fill=white] (11.2583,-3.5) circle (0.288);
\draw[black, line width=0.4mm, fill=white] (11.2583,-4.5) circle (0.288);
\draw[black, line width=0.4mm, fill=white] (11.2583,-5.5) circle (0.288);
\draw[black, line width=0.4mm, fill=white] (11.2583,-6.5) circle (0.288);
\draw[black, line width=0.4mm, fill=white] (11.2583,-7.5) circle (0.288);
\draw[black, line width=0.5mm] (8.05404,-1.65) -- (8.40045,-1.85);
\draw[black, line width=0.5mm] (4.33013,-6.2) -- (4.33013,-5.8);
\draw[black, line width=0.5mm] (9.78609,-2.65) -- (10.1325,-2.85);
\draw[black, line width=0.5mm] (11.2583,-4.2) -- (11.2583,-3.8);
\draw[black, line width=0.5mm] (7.18801,-9.15) -- (7.53442,-9.35);
\draw[black, line width=0.5mm] (6.32199,-2.35) -- (6.6684,-2.15);
\draw[black, line width=0.5mm] (4.33013,-5.2) -- (4.33013,-4.8);
\draw[black, line width=0.5mm] (8.92006,-2.15) -- (9.26647,-2.35);
\draw[black, line width=0.5mm] (5.45596,-2.85) -- (5.80237,-2.65);
\draw[black, line width=0.5mm] (11.2583,-7.2) -- (11.2583,-6.8);
\draw[black, line width=0.5mm] (1.99186,-2.15) -- (2.33827,-2.35);
\draw[black, line width=0.5mm] (7.18801,-5.85) -- (7.53442,-5.65);
\draw[black, line width=0.5mm] (7.18801,-6.15) -- (7.53442,-6.35);
\draw[black, line width=0.5mm] (1.12583,-1.65) -- (1.47224,-1.85);
\draw[black, line width=0.5mm] (8.92006,-8.85) -- (9.26647,-8.65);
\draw[black, line width=0.5mm] (10.6521,-3.15) -- (10.9985,-3.35);
\draw[black, line width=0.5mm] (10.6521,-7.85) -- (10.9985,-7.65);
\draw[black, line width=0.5mm] (5.19615,-3.7) -- (5.19615,-3.3);
\draw[black, line width=0.5mm] (5.45596,-4.15) -- (5.80237,-4.35);
\draw[black, line width=0.5mm] (2.85788,-2.65) -- (3.20429,-2.85);
\draw[black, line width=0.5mm] (4.33013,-4.2) -- (4.33013,-3.8);
\draw[black, line width=0.5mm] (4.58993,-4.35) -- (4.93634,-4.15);
\draw[black, line width=0.5mm] (7.18801,-1.85) -- (7.53442,-1.65);
\draw[black, line width=0.5mm] (11.2583,-6.2) -- (11.2583,-5.8);
\draw[black, line width=0.5mm] (6.9282,-6.7) -- (6.9282,-6.3);
\draw[black, line width=0.5mm] (7.18801,-6.85) -- (7.53442,-6.65);
\draw[black, line width=0.5mm] (7.79423,-6.2) -- (7.79423,-5.8);
\draw[black, line width=0.5mm] (8.05404,-6.65) -- (8.40045,-6.85);
\draw[black, line width=0.5mm] (4.58993,-3.35) -- (4.93634,-3.15);
\draw[black, line width=0.5mm] (4.58993,-3.65) -- (4.93634,-3.85);
\draw[black, line width=0.5mm] (5.45596,-8.15) -- (5.80237,-8.35);
\draw[black, line width=0.5mm] (9.78609,-8.35) -- (10.1325,-8.15);
\draw[black, line width=0.5mm] (4.33013,-7.2) -- (4.33013,-6.8);
\draw[black, line width=0.5mm] (4.58993,-7.65) -- (4.93634,-7.85);
\draw[black, line width=0.5mm] (8.05404,-9.35) -- (8.40045,-9.15);
\draw[black, line width=0.5mm] (11.2583,-5.2) -- (11.2583,-4.8);
\draw[black, line width=0.5mm] (6.32199,-4.35) -- (6.6684,-4.15);
\draw[black, line width=0.5mm] (3.72391,-3.15) -- (4.07032,-3.35);
\draw[black, line width=0.5mm] (7.79423,-5.2) -- (7.79423,-4.8);
\draw[black, line width=0.5mm] (6.32199,-8.65) -- (6.6684,-8.85);
\draw[black, line width=0.5mm] (7.18801,-4.15) -- (7.53442,-4.35);
\draw[black!20!red, line width=0.36mm, ] (0.866025,-1.5) circle (0.408);
\draw[black!20!red,-{Stealth[length=1.6mm,width=2.5mm]},line width=0.7mm] (0.597558,-1.655) -- (0,-2);
\draw[black!20!red, line width=0.36mm, ] (1.73205,-2) circle (0.408);
\draw[black!20!red,-{Stealth[length=1.6mm,width=2.5mm]},line width=0.7mm] (1.46358,-2.155) -- (0.866025,-2.5);
\draw[black!20!red, line width=0.36mm, ] (2.59808,-2.5) circle (0.408);
\draw[black!20!red,-{Stealth[length=1.6mm,width=2.5mm]},line width=0.7mm] (2.32961,-2.655) -- (1.73205,-3);
\draw[black!20!red, line width=0.36mm, ] (3.4641,-3) circle (0.408);
\draw[black!20!red,-{Stealth[length=1.6mm,width=2.5mm]},line width=0.7mm] (3.19563,-3.155) -- (2.59808,-3.5);
\draw[black!20!red, line width=0.36mm, ] (4.33013,-3.5) circle (0.408);
\draw[black!20!red,-{Stealth[length=1.6mm,width=2.5mm]},line width=0.7mm] (4.06166,-3.655) -- (3.4641,-4);
  \end{tikzpicture}
  \end{center}
  \caption{Moving a corner outward.}
  \label{fig:hexagon_tail1}
\end{subfigure}%
\begin{subfigure}{.33\textwidth}
  %\centering
  %\includegraphics[width=.8\linewidth]{diagrams_irreducibility/hexagon_tail2.png}
  \begin{center}
  \begin{tikzpicture}[x=0.35cm,y=0.35cm]
  \draw[lightgray] (10.3923,-1) -- (10.3923,-10);
\draw[lightgray] (0,-7) -- (11.2583,-0.5);
\draw[lightgray] (0,-7) -- (6.06218,-10.5);
\draw[lightgray] (0,-10) -- (12.1244,-3);
\draw[lightgray] (0,-10) -- (0.866025,-10.5);
\draw[lightgray] (7.79423,-0.5) -- (12.1244,-3);
\draw[lightgray] (7.79423,-0.5) -- (7.79423,-10.5);
\draw[lightgray] (12.1244,-1) -- (12.1244,-10);
\draw[lightgray] (0,-4) -- (6.06218,-0.5);
\draw[lightgray] (0,-4) -- (11.2583,-10.5);
\draw[lightgray] (9.52628,-0.5) -- (12.1244,-2);
\draw[lightgray] (9.52628,-0.5) -- (9.52628,-10.5);
\draw[lightgray] (11.2583,-10.5) -- (12.1244,-10);
\draw[lightgray] (0,-1) -- (0.866025,-0.5);
\draw[lightgray] (0,-1) -- (12.1244,-8);
\draw[lightgray] (0,-1) -- (0,-10);
\draw[lightgray] (11.2583,-0.5) -- (12.1244,-1);
\draw[lightgray] (11.2583,-0.5) -- (11.2583,-10.5);
\draw[lightgray] (1.73205,-1) -- (1.73205,-10);
\draw[lightgray] (9.52628,-10.5) -- (12.1244,-9);
\draw[lightgray] (0.866025,-10.5) -- (12.1244,-4);
\draw[lightgray] (3.4641,-1) -- (3.4641,-10);
\draw[lightgray] (0,-5) -- (7.79423,-0.5);
\draw[lightgray] (0,-5) -- (9.52628,-10.5);
\draw[lightgray] (0.866025,-0.5) -- (12.1244,-7);
\draw[lightgray] (0.866025,-0.5) -- (0.866025,-10.5);
\draw[lightgray] (0,-8) -- (12.1244,-1);
\draw[lightgray] (0,-8) -- (4.33013,-10.5);
\draw[lightgray] (4.33013,-10.5) -- (12.1244,-6);
\draw[lightgray] (5.19615,-1) -- (5.19615,-10);
\draw[lightgray] (0,-2) -- (2.59808,-0.5);
\draw[lightgray] (0,-2) -- (12.1244,-9);
\draw[lightgray] (2.59808,-0.5) -- (12.1244,-6);
\draw[lightgray] (2.59808,-0.5) -- (2.59808,-10.5);
\draw[lightgray] (2.59808,-10.5) -- (12.1244,-5);
\draw[lightgray] (4.33013,-0.5) -- (12.1244,-5);
\draw[lightgray] (4.33013,-0.5) -- (4.33013,-10.5);
\draw[lightgray] (6.06218,-0.5) -- (12.1244,-4);
\draw[lightgray] (6.06218,-0.5) -- (6.06218,-10.5);
\draw[lightgray] (0,-6) -- (9.52628,-0.5);
\draw[lightgray] (0,-6) -- (7.79423,-10.5);
\draw[lightgray] (0,-9) -- (12.1244,-2);
\draw[lightgray] (0,-9) -- (2.59808,-10.5);
\draw[lightgray] (0,-3) -- (4.33013,-0.5);
\draw[lightgray] (0,-3) -- (12.1244,-10);
\draw[lightgray] (6.9282,-1) -- (6.9282,-10);
\draw[lightgray] (6.06218,-10.5) -- (12.1244,-7);
\draw[lightgray] (7.79423,-10.5) -- (12.1244,-8);
\draw[lightgray] (8.66025,-1) -- (8.66025,-10);
\draw[black!40!green, dotted, line width=0.5mm] (0,-10) -- (12.1244,-3);
\draw[black!40!green, dotted, line width=0.5mm] (0,-1) -- (12.1244,-8);
\draw[black!40!green, dotted, line width=0.5mm] (7.79423,-0.5) -- (7.79423,-10.5);
\draw[black, line width=0.4mm, fill=white] (0.866025,-1.5) circle (0.288);
\draw[black, line width=0.4mm, fill=white] (1.73205,-2) circle (0.288);
\draw[black, line width=0.4mm, fill=white] (2.59808,-2.5) circle (0.288);
\draw[black, line width=0.4mm, fill=white] (3.4641,-3) circle (0.288);
\draw[black, line width=0.4mm, fill=white] (4.33013,-3.5) circle (0.288);
\draw[black, line width=0.4mm, fill=white] (4.33013,-4.5) circle (0.288);
\draw[black, line width=0.4mm, fill=white] (4.33013,-5.5) circle (0.288);
\node[align=left] at (3.5074,-5.5) {\footnotesize $p_2$};
\draw[black, line width=0.4mm, fill=white] (4.33013,-6.5) circle (0.288);
\draw[black, line width=0.4mm, fill=white] (4.33013,-7.5) circle (0.288);
\draw[black, line width=0.4mm, fill=white] (5.19615,-3) circle (0.288);
\draw[black, line width=0.4mm, fill=white] (5.19615,-5) circle (0.288);
\draw[black, line width=0.4mm, fill=white] (5.19615,-8) circle (0.288);
\draw[black, line width=0.4mm, fill=white] (6.06218,-2.5) circle (0.288);
\draw[black, line width=0.4mm, fill=white] (6.06218,-5.5) circle (0.288);
\draw[black, line width=0.4mm, fill=white] (6.06218,-8.5) circle (0.288);
\draw[black, line width=0.4mm, fill=white] (6.9282,-2) circle (0.288);
\draw[black, line width=0.4mm, fill=white] (6.9282,-5) circle (0.288);
\draw[black, line width=0.4mm, fill=white] (6.9282,-9) circle (0.288);
\draw[black, line width=0.4mm, fill=white] (7.79423,-1.5) circle (0.288);
\draw[black, line width=0.4mm, fill=white] (7.79423,-3.5) circle (0.288);
\draw[black, line width=0.4mm, fill=white] (7.79423,-4.5) circle (0.288);
\draw[black, line width=0.4mm, fill=white] (7.79423,-5.5) circle (0.336);
\draw[black, line width=0.32mm] (7.79423,-5.5) circle (0.24);
\node[align=left] at (7.79423,-5.5) {\scriptsize $f$};
\draw[black, line width=0.4mm, fill=white] (7.79423,-6.5) circle (0.288);
\draw[black, line width=0.4mm, fill=white] (7.79423,-9.5) circle (0.288);
\draw[black, line width=0.4mm, fill=white] (8.66025,-2) circle (0.288);
\draw[black, line width=0.4mm, fill=white] (8.66025,-4) circle (0.288);
\draw[black, line width=0.4mm, fill=white] (8.66025,-5) circle (0.288);
\draw[black, line width=0.4mm, fill=white] (8.66025,-9) circle (0.288);
\draw[black, line width=0.4mm, fill=white] (9.52628,-2.5) circle (0.288);
\draw[black, line width=0.4mm, fill=white] (9.52628,-8.5) circle (0.288);
\draw[black, line width=0.4mm, fill=white] (10.3923,-3) circle (0.288);
\draw[black, line width=0.4mm, fill=white] (10.3923,-8) circle (0.288);
\draw[black, line width=0.4mm, fill=white] (11.2583,-3.5) circle (0.288);
\draw[black, line width=0.4mm, fill=white] (11.2583,-4.5) circle (0.288);
\draw[black, line width=0.4mm, fill=white] (11.2583,-5.5) circle (0.288);
\draw[black, line width=0.4mm, fill=white] (11.2583,-6.5) circle (0.288);
\draw[black, line width=0.4mm, fill=white] (11.2583,-7.5) circle (0.288);
\draw[black, line width=0.5mm] (8.05404,-1.65) -- (8.40045,-1.85);
\draw[black, line width=0.5mm] (4.33013,-6.2) -- (4.33013,-5.8);
\draw[black, line width=0.5mm] (9.78609,-2.65) -- (10.1325,-2.85);
\draw[black, line width=0.5mm] (11.2583,-4.2) -- (11.2583,-3.8);
\draw[black, line width=0.5mm] (7.18801,-9.15) -- (7.53442,-9.35);
\draw[black, line width=0.5mm] (7.79423,-4.2) -- (7.79423,-3.8);
\draw[black, line width=0.5mm] (8.05404,-4.35) -- (8.40045,-4.15);
\draw[black, line width=0.5mm] (8.05404,-4.65) -- (8.40045,-4.85);
\draw[black, line width=0.5mm] (7.18801,-4.85) -- (7.53442,-4.65);
\draw[black, line width=0.5mm] (7.18801,-5.15) -- (7.53442,-5.35);
\draw[black, line width=0.5mm] (6.32199,-2.35) -- (6.6684,-2.15);
\draw[black, line width=0.5mm] (4.33013,-5.2) -- (4.33013,-4.8);
\draw[black, line width=0.5mm] (4.58993,-5.35) -- (4.93634,-5.15);
\draw[black, line width=0.5mm] (5.45596,-2.85) -- (5.80237,-2.65);
\draw[black, line width=0.5mm] (11.2583,-7.2) -- (11.2583,-6.8);
\draw[black, line width=0.5mm] (1.99186,-2.15) -- (2.33827,-2.35);
\draw[black, line width=0.5mm] (1.12583,-1.65) -- (1.47224,-1.85);
\draw[black, line width=0.5mm] (8.92006,-8.85) -- (9.26647,-8.65);
\draw[black, line width=0.5mm] (10.6521,-3.15) -- (10.9985,-3.35);
\draw[black, line width=0.5mm] (10.6521,-7.85) -- (10.9985,-7.65);
\draw[black, line width=0.5mm] (2.85788,-2.65) -- (3.20429,-2.85);
\draw[black, line width=0.5mm] (4.33013,-4.2) -- (4.33013,-3.8);
\draw[black, line width=0.5mm] (4.58993,-4.65) -- (4.93634,-4.85);
\draw[black, line width=0.5mm] (7.18801,-1.85) -- (7.53442,-1.65);
\draw[black, line width=0.5mm] (8.05404,-3.65) -- (8.40045,-3.85);
\draw[black, line width=0.5mm] (11.2583,-6.2) -- (11.2583,-5.8);
\draw[black, line width=0.5mm] (6.32199,-5.35) -- (6.6684,-5.15);
\draw[black, line width=0.5mm] (7.79423,-6.2) -- (7.79423,-5.8);
\draw[black, line width=0.5mm] (5.45596,-5.15) -- (5.80237,-5.35);
\draw[black, line width=0.5mm] (4.58993,-3.35) -- (4.93634,-3.15);
\draw[black, line width=0.5mm] (8.66025,-4.7) -- (8.66025,-4.3);
\draw[black, line width=0.5mm] (5.45596,-8.15) -- (5.80237,-8.35);
\draw[black, line width=0.5mm] (9.78609,-8.35) -- (10.1325,-8.15);
\draw[black, line width=0.5mm] (4.33013,-7.2) -- (4.33013,-6.8);
\draw[black, line width=0.5mm] (4.58993,-7.65) -- (4.93634,-7.85);
\draw[black, line width=0.5mm] (8.05404,-9.35) -- (8.40045,-9.15);
\draw[black, line width=0.5mm] (11.2583,-5.2) -- (11.2583,-4.8);
\draw[black, line width=0.5mm] (3.72391,-3.15) -- (4.07032,-3.35);
\draw[black, line width=0.5mm] (7.79423,-5.2) -- (7.79423,-4.8);
\draw[black, line width=0.5mm] (8.05404,-5.35) -- (8.40045,-5.15);
\draw[black, line width=0.5mm] (6.32199,-8.65) -- (6.6684,-8.85);
\draw[black, line width=0.5mm] (8.92006,-2.15) -- (9.26647,-2.35);
\draw[black!20!red, line width=0.36mm, ] (4.33013,-5.5) circle (0.408);
\draw[black!20!red,-{Stealth[length=1.6mm,width=2.5mm]},line width=0.7mm] (4.59859,-5.655) -- (5.19615,-6);
  \end{tikzpicture}
  \end{center}
  \caption{Moving a side agent inward.}
  \label{fig:hexagon_tail2}
\end{subfigure}%
\begin{subfigure}{.33\textwidth}
  %\centering
  %\includegraphics[width=.8\linewidth]{diagrams_irreducibility/hexagon_tail3.png}
  \begin{center}
  \begin{tikzpicture}[x=0.35cm,y=0.35cm]
  \draw[lightgray] (10.3923,-1) -- (10.3923,-10);
\draw[lightgray] (0,-7) -- (11.2583,-0.5);
\draw[lightgray] (0,-7) -- (6.06218,-10.5);
\draw[lightgray] (0,-10) -- (12.1244,-3);
\draw[lightgray] (0,-10) -- (0.866025,-10.5);
\draw[lightgray] (7.79423,-0.5) -- (12.1244,-3);
\draw[lightgray] (7.79423,-0.5) -- (7.79423,-10.5);
\draw[lightgray] (12.1244,-1) -- (12.1244,-10);
\draw[lightgray] (0,-4) -- (6.06218,-0.5);
\draw[lightgray] (0,-4) -- (11.2583,-10.5);
\draw[lightgray] (9.52628,-0.5) -- (12.1244,-2);
\draw[lightgray] (9.52628,-0.5) -- (9.52628,-10.5);
\draw[lightgray] (11.2583,-10.5) -- (12.1244,-10);
\draw[lightgray] (0,-1) -- (0.866025,-0.5);
\draw[lightgray] (0,-1) -- (12.1244,-8);
\draw[lightgray] (0,-1) -- (0,-10);
\draw[lightgray] (11.2583,-0.5) -- (12.1244,-1);
\draw[lightgray] (11.2583,-0.5) -- (11.2583,-10.5);
\draw[lightgray] (1.73205,-1) -- (1.73205,-10);
\draw[lightgray] (9.52628,-10.5) -- (12.1244,-9);
\draw[lightgray] (0.866025,-10.5) -- (12.1244,-4);
\draw[lightgray] (3.4641,-1) -- (3.4641,-10);
\draw[lightgray] (0,-5) -- (7.79423,-0.5);
\draw[lightgray] (0,-5) -- (9.52628,-10.5);
\draw[lightgray] (0.866025,-0.5) -- (12.1244,-7);
\draw[lightgray] (0.866025,-0.5) -- (0.866025,-10.5);
\draw[lightgray] (0,-8) -- (12.1244,-1);
\draw[lightgray] (0,-8) -- (4.33013,-10.5);
\draw[lightgray] (4.33013,-10.5) -- (12.1244,-6);
\draw[lightgray] (5.19615,-1) -- (5.19615,-10);
\draw[lightgray] (0,-2) -- (2.59808,-0.5);
\draw[lightgray] (0,-2) -- (12.1244,-9);
\draw[lightgray] (2.59808,-0.5) -- (12.1244,-6);
\draw[lightgray] (2.59808,-0.5) -- (2.59808,-10.5);
\draw[lightgray] (2.59808,-10.5) -- (12.1244,-5);
\draw[lightgray] (4.33013,-0.5) -- (12.1244,-5);
\draw[lightgray] (4.33013,-0.5) -- (4.33013,-10.5);
\draw[lightgray] (6.06218,-0.5) -- (12.1244,-4);
\draw[lightgray] (6.06218,-0.5) -- (6.06218,-10.5);
\draw[lightgray] (0,-6) -- (9.52628,-0.5);
\draw[lightgray] (0,-6) -- (7.79423,-10.5);
\draw[lightgray] (0,-9) -- (12.1244,-2);
\draw[lightgray] (0,-9) -- (2.59808,-10.5);
\draw[lightgray] (0,-3) -- (4.33013,-0.5);
\draw[lightgray] (0,-3) -- (12.1244,-10);
\draw[lightgray] (6.9282,-1) -- (6.9282,-10);
\draw[lightgray] (6.06218,-10.5) -- (12.1244,-7);
\draw[lightgray] (7.79423,-10.5) -- (12.1244,-8);
\draw[lightgray] (8.66025,-1) -- (8.66025,-10);
\draw[black!40!green, dotted, line width=0.5mm] (0,-10) -- (12.1244,-3);
\draw[black!40!green, dotted, line width=0.5mm] (0,-1) -- (12.1244,-8);
\draw[black!40!green, dotted, line width=0.5mm] (7.79423,-0.5) -- (7.79423,-10.5);
\draw[black, line width=0.4mm, fill=white] (0.866025,-1.5) circle (0.288);
\draw[black, line width=0.4mm, fill=white] (1.73205,-2) circle (0.288);
\draw[black, line width=0.4mm, fill=white] (2.59808,-2.5) circle (0.288);
\draw[black, line width=0.4mm, fill=white] (3.4641,-3) circle (0.288);
\draw[black, line width=0.4mm, fill=white] (4.33013,-3.5) circle (0.288);
\draw[black, line width=0.4mm, fill=white] (4.33013,-4.5) circle (0.288);
\draw[black, line width=0.4mm, fill=white] (4.33013,-5.5) circle (0.288);
\node[align=left] at (3.72391,-5.95) {\footnotesize $p_3$};
\draw[black, line width=0.4mm, fill=white] (4.33013,-6.5) circle (0.288);
\draw[black, line width=0.4mm, fill=white] (4.33013,-7.5) circle (0.288);
\draw[black, line width=0.4mm, fill=white] (5.19615,-3) circle (0.288);
\draw[black, line width=0.4mm, fill=white] (5.19615,-5) circle (0.288);
\draw[black, line width=0.4mm, fill=white] (5.19615,-6) circle (0.288);
\draw[black, line width=0.4mm, fill=white] (5.19615,-8) circle (0.288);
\draw[black, line width=0.4mm, fill=white] (6.06218,-2.5) circle (0.288);
\draw[black, line width=0.4mm, fill=white] (6.06218,-4.5) circle (0.288);
\draw[black, line width=0.4mm, fill=white] (6.06218,-8.5) circle (0.288);
\draw[black, line width=0.4mm, fill=white] (6.9282,-2) circle (0.288);
\draw[black, line width=0.4mm, fill=white] (6.9282,-5) circle (0.288);
\draw[black, line width=0.4mm, fill=white] (6.9282,-9) circle (0.288);
\draw[black, line width=0.4mm, fill=white] (7.79423,-1.5) circle (0.288);
\draw[black, line width=0.4mm, fill=white] (7.79423,-3.5) circle (0.288);
\draw[black, line width=0.4mm, fill=white] (7.79423,-4.5) circle (0.288);
\draw[black, line width=0.4mm, fill=white] (7.79423,-5.5) circle (0.336);
\draw[black, line width=0.32mm] (7.79423,-5.5) circle (0.24);
\node[align=left] at (7.79423,-5.5) {\scriptsize $f$};
\draw[black, line width=0.4mm, fill=white] (7.79423,-6.5) circle (0.288);
\draw[black, line width=0.4mm, fill=white] (7.79423,-9.5) circle (0.288);
\draw[black, line width=0.4mm, fill=white] (8.66025,-2) circle (0.288);
\draw[black, line width=0.4mm, fill=white] (8.66025,-4) circle (0.288);
\draw[black, line width=0.4mm, fill=white] (8.66025,-5) circle (0.288);
\draw[black, line width=0.4mm, fill=white] (8.66025,-9) circle (0.288);
\draw[black, line width=0.4mm, fill=white] (9.52628,-2.5) circle (0.288);
\draw[black, line width=0.4mm, fill=white] (9.52628,-8.5) circle (0.288);
\draw[black, line width=0.4mm, fill=white] (10.3923,-3) circle (0.288);
\draw[black, line width=0.4mm, fill=white] (10.3923,-8) circle (0.288);
\draw[black, line width=0.4mm, fill=white] (11.2583,-3.5) circle (0.288);
\draw[black, line width=0.4mm, fill=white] (11.2583,-4.5) circle (0.288);
\draw[black, line width=0.4mm, fill=white] (11.2583,-5.5) circle (0.288);
\draw[black, line width=0.4mm, fill=white] (11.2583,-6.5) circle (0.288);
\draw[black, line width=0.4mm, fill=white] (11.2583,-7.5) circle (0.288);
\draw[black, line width=0.5mm] (8.05404,-1.65) -- (8.40045,-1.85);
\draw[black, line width=0.5mm] (5.19615,-5.7) -- (5.19615,-5.3);
\draw[black, line width=0.5mm] (4.33013,-6.2) -- (4.33013,-5.8);
\draw[black, line width=0.5mm] (4.58993,-6.35) -- (4.93634,-6.15);
\draw[black, line width=0.5mm] (9.78609,-2.65) -- (10.1325,-2.85);
\draw[black, line width=0.5mm] (11.2583,-4.2) -- (11.2583,-3.8);
\draw[black, line width=0.5mm] (7.18801,-9.15) -- (7.53442,-9.35);
\draw[black, line width=0.5mm] (7.79423,-4.2) -- (7.79423,-3.8);
\draw[black, line width=0.5mm] (8.05404,-4.35) -- (8.40045,-4.15);
\draw[black, line width=0.5mm] (8.05404,-4.65) -- (8.40045,-4.85);
\draw[black, line width=0.5mm] (7.18801,-4.85) -- (7.53442,-4.65);
\draw[black, line width=0.5mm] (7.18801,-5.15) -- (7.53442,-5.35);
\draw[black, line width=0.5mm] (6.32199,-2.35) -- (6.6684,-2.15);
\draw[black, line width=0.5mm] (4.33013,-5.2) -- (4.33013,-4.8);
\draw[black, line width=0.5mm] (4.58993,-5.35) -- (4.93634,-5.15);
\draw[black, line width=0.5mm] (4.58993,-5.65) -- (4.93634,-5.85);
\draw[black, line width=0.5mm] (5.45596,-2.85) -- (5.80237,-2.65);
\draw[black, line width=0.5mm] (11.2583,-7.2) -- (11.2583,-6.8);
\draw[black, line width=0.5mm] (1.99186,-2.15) -- (2.33827,-2.35);
\draw[black, line width=0.5mm] (1.12583,-1.65) -- (1.47224,-1.85);
\draw[black, line width=0.5mm] (8.92006,-8.85) -- (9.26647,-8.65);
\draw[black, line width=0.5mm] (10.6521,-3.15) -- (10.9985,-3.35);
\draw[black, line width=0.5mm] (10.6521,-7.85) -- (10.9985,-7.65);
\draw[black, line width=0.5mm] (2.85788,-2.65) -- (3.20429,-2.85);
\draw[black, line width=0.5mm] (4.33013,-4.2) -- (4.33013,-3.8);
\draw[black, line width=0.5mm] (4.58993,-4.65) -- (4.93634,-4.85);
\draw[black, line width=0.5mm] (7.18801,-1.85) -- (7.53442,-1.65);
\draw[black, line width=0.5mm] (8.05404,-3.65) -- (8.40045,-3.85);
\draw[black, line width=0.5mm] (11.2583,-6.2) -- (11.2583,-5.8);
\draw[black, line width=0.5mm] (7.79423,-6.2) -- (7.79423,-5.8);
\draw[black, line width=0.5mm] (5.45596,-4.85) -- (5.80237,-4.65);
\draw[black, line width=0.5mm] (4.58993,-3.35) -- (4.93634,-3.15);
\draw[black, line width=0.5mm] (8.66025,-4.7) -- (8.66025,-4.3);
\draw[black, line width=0.5mm] (5.45596,-8.15) -- (5.80237,-8.35);
\draw[black, line width=0.5mm] (9.78609,-8.35) -- (10.1325,-8.15);
\draw[black, line width=0.5mm] (4.33013,-7.2) -- (4.33013,-6.8);
\draw[black, line width=0.5mm] (4.58993,-7.65) -- (4.93634,-7.85);
\draw[black, line width=0.5mm] (8.05404,-9.35) -- (8.40045,-9.15);
\draw[black, line width=0.5mm] (11.2583,-5.2) -- (11.2583,-4.8);
\draw[black, line width=0.5mm] (6.32199,-4.65) -- (6.6684,-4.85);
\draw[black, line width=0.5mm] (3.72391,-3.15) -- (4.07032,-3.35);
\draw[black, line width=0.5mm] (7.79423,-5.2) -- (7.79423,-4.8);
\draw[black, line width=0.5mm] (8.05404,-5.35) -- (8.40045,-5.15);
\draw[black, line width=0.5mm] (6.32199,-8.65) -- (6.6684,-8.85);
\draw[black, line width=0.5mm] (8.92006,-2.15) -- (9.26647,-2.35);
\draw[black!20!red, line width=0.36mm, ] (4.33013,-5.5) circle (0.408);
\draw[black!20!red,-{Stealth[length=1.6mm,width=2.5mm]},line width=0.7mm] (4.06166,-5.345) -- (3.4641,-5);
  \end{tikzpicture}
  \end{center}
  \caption{Moving a side agent outward.}
  \label{fig:hexagon_tail3}
\end{subfigure}%
\caption{Possible cases for reducing the minimum \spine{} length from a hexagon with a tail.}
\label{fig:hexagon_tail}
\end{figure*}
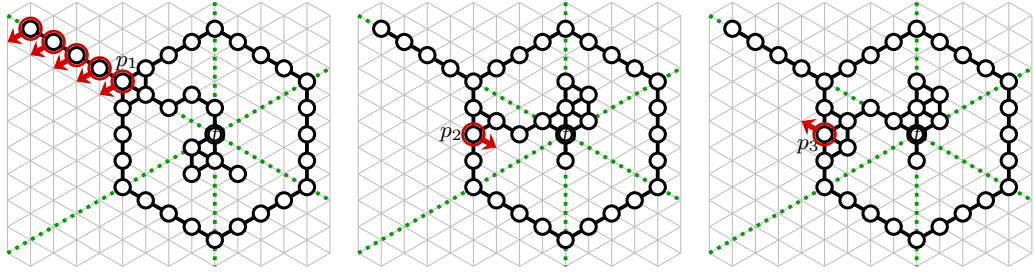

\begin{definition}[Hexagon with a Tail]
We say a configuration forms a has the ``hexagon with a tail'' arrangement of radius $r$ if:
\begin{itemize}
\item All \spines{} have length exactly $r$.
\item There are tail agents on at most one of the \spines{}.
\item Aside from these tail agents, there are no agents of distance greater than $r$ from the center.
\end{itemize}
If $r = 0$, this ``regular hexagon'' comprises of only the immobile agent. In other words, a hexagon with a tail of radius $r$ has all of the agents extending in a straight line from the immobile agent.
\end{definition}

\begin{lemma}[Reaching a Hexagon with a Tail]
\label{lem:reachinghexagon}
After seven \spinecombs{} in a counterclockwise order starting from a \spine{} of minimum length $r$, assuming that no gaps in the lines are formed and that no \spine{} ends up with length below $r$ in the process, we will end up with a hexagon with a tail arrangement of radius $r$.
\end{lemma}

\begin{proof}
We will denote the starting \spine{} as $S_0$, and name the remaining spines $S_1$ to $S_5$ in counterclockwise order. The \spinecombs{} hence go from $S_0$ to $S_1$, from $S_1$ to $S_2$ and so on, with the final (seventh) comb being from $S_0$ to $S_1$. \Spine{} $S_0$ is assumed to be a minimum length \spine{}, of length $r$.

By Lemma~\ref{lem:aboveunaffected}, a \spinecomb{} from spines $S_i$ to $S_{i+1}$ will only affect agents on spines $S_i$, $S_{i+1}$, $S_{i+2}$, and the agents between spines $S_i$ and $S_{i+1}$, between spines $S_{i+1}$ and $S_{i+2}$, and between spines $S_{i+2}$ and $S_{i+3}$. Note that this does include the agents on spine $S_{i+3}$. Hence, the first four \spinecombs{} will not affect the result of the first \spinecomb{} from $S_0$ to $S_1$. 

On the fifth \spinecomb{} from $S_4$ to $S_5$, as usual without loss of generality we take $S_4$ to be the \spine{} going up-left and $S_5$ to be the \spine{} going down-left. \Spine{} $S_0$ will thus be going downwards and \spine{} $S_1$ will be going down-right. Due to the effects of the first three combs, there will be no agent further right than the \anchoragent{} of \spine{} $S_1$. By Lemma~\ref{lem:rightmostextent}, while the fifth \spinecomb{} may move agents onto \spine{} $S_0$ or the region between \spines{} $S_0$ and $S_1$, none of these agents in the resulting configuration will be further right than the \anchoragent{} of \spine $S_1$.

On the sixth \spinecomb{} from $S_5$ to $S_0$, taking $S_5$ to be going up-left and $S_0$ to be going down-left, consider the position $(-r-1,0)$, which is one agent down-right of the \anchoragent{} of the down-right \spine{} $S_2$. The region $R_{-r-1,0}$, as defined in Lemma~\ref{lem:unenterableregion}, will be empty after the fifth \spinecomb{}, due to what we have just shown to happen after the fifth \spinecomb{}. By Lemma~\ref{lem:unenterableregion}, this region will continue to be empty after the sixth \spinecomb{}.

On the seventh and final \spinecomb{} from $S_0$ to $S_1$, take $S_0$ to be going up-left and $S_1$ to be going down-left. Consider the positions $(r,0)$ and $(r,r)$, which are on the \sourcespine{} $S_0$ and \targetspine{} $S_1$ respectively, of distance $r$ from the center. As a result of the sixth \spinecomb{} with Lemma~\ref{lem:nogapresult}, all agents of distance greater than $r$ from the center must lie between (inclusive) the two diagonal lines going up-left from the positions $(r,0)$ and $(r,r)$. Now, from the position $(r,r+1)$ which lies directly below $(r,r)$ and the position $(-r-1,0)$, which lies on \spine{} $S_3$ of distance $r+1$ from the center, we consider the two regions $R_{r,r+1}$ and $R_{-r-1,0}$ as in Lemma~\ref{lem:unenterableregion}. Both of these regions are initially empty, and so will remain empty after the seventh comb. By Lemmas~\ref{lem:aboveunaffected} and \ref{lem:nogapresult}, the only place where agents of distance greater than $r$ can be are on the \targetspine{} $S_1$.

As we had assumed that no \spine{} will have ended up with length less than $r$ in the process, this means we have reached a hexagon with a tail arrangement of radius $r$.
\end{proof}

The following Lemma then concludes the proof that we can always reduce the minimum \spine{} length, provided that the current minimum \spine{} length is at least $1$.

\begin{lemma}
From a hexagon with a tail arrangement of radius $r \geq 1$, there exists a sequence of moves to reduce the minimum \spine{} length by $1$.
\end{lemma}

\begin{proof}
Consider the set $H$ of positions of distance exactly $r$ from the center. This set of positions is in the shape of a hexagon. If one of these sites is unfilled, without loss of generality assume this site $(r,d)$ is on the left side of the hexagon (it is not on a corner as all \spines{} have length $r$). The site $(r,d+1)$ is combable, which by Lemma~\ref{lem:combingspinelength} gives us a way to reduce the length of the \spine{} going down-left to at most $r-1$.

If no such gap in $H$ currently exists, we show that we can create such a gap.
If $r=1$, pick any agent on the hexagon $H$ aside from the one on the \spine{} the tail is on. This agent can be moved to a vacant spot between two \spines{}, reducing the minimum \spine{} length to $0$.

% NOTE: This proof is hard to visualize without diagrams
Otherwise, as the configuration is assumed to be connected, there is a path of agents from the center (immobile) agent to an agent on $H$. This implies that there is an agent of $v_{-2}$ distance $r-2$ from the center adjacent to an agent $v_{-1}$ of distance $r-1$ from the center. Note that if $r = 2$, $v_{-2}$ will be the immobile agent. If $v_{-1}$ is adjacent to a corner agent of the hexagon $H$, assuming without loss of generality that this corner is on the \spine{} going up-left, we can move this corner agent one step down-left, and if there are any tail agents attached to this corner agent, they can then subsequently be moved one-by-one one step down-left as well (Figure~\ref{fig:hexagon_tail1}). This reduces the minimum \spine{} length to at most $r-1$.

If $v_{-1}$ is not adjacent to a corner agent of $H$, we note that $v_{-1}$ and $v_{-2}$ will share a neighboring site $u_{-1}$ of distance $r-1$ from the center. The sites $u_{-1}$ and $v_{-1}$ share a neighbor agent $v_0$ on $H$. If site $u_{-1}$ is unoccupied, agent $v_0$ can be moved into site $u_{-1}$, creating a gap in the hexagon $H$ (Figure~\ref{fig:hexagon_tail2}). If $u_{-1}$ is occupied, $v_0$ can be moved in the opposite direction of $u_{-1}$, to a position $u_{+1}$ of distance $r+1$ from the center, creating a gap in the hexagon $H$ (Figure~\ref{fig:hexagon_tail3}).

In both of these cases, by reflection and rotational symmetry, without loss of generality, this newly created gap $v_0$ is on the left wall of the hexagon $H$, and if the agent was moved to $u_{+1}$, $u_{+1}$ is directly up-left of $v_0$. The site directly below $v_0$ is thus combable, and by Lemma~\ref{lem:combingspinelength}, combing this reduces the minimum \spine{} length to at most $r-1$.
\end{proof}

%\subsection{Forming a Straight Line of Agents}
Finally, we show that we can reach a straight line of agents, thus showing ergodicity of the chain since the chain is reversible.
\begin{lemma}
From any connected configuration of agents with one single immobile agent, there exists a sequence of valid moves to transform this configuration into a straight line of agents with the immobile agent at one end.
\end{lemma}

\begin{proof}
Applying Lemma~\ref{lem:reducespinelength} repeatedly allows us to arrive at a configuration with minimum \spine{} length $0$. 
Applying Lemma~\ref{lem:reachinghexagon} from here gives us a hexagon with a tail arrangement of radius $0$, which is a straight line of agents with the immobile agent at one end.
\end{proof}

We observe that the direction in which the final tail faces is irrelevant, as there is a simple sequence of moves to change the direction of the tail, by moving the agents one by one to the location of the new tail, starting from the agent at the very end of the initial tail. This thus allows us to conclude Lemma~\ref{lemma:irreducible}, which also implies that the Markov chain is irreducible.

%\end{document}

\end{document}